\documentclass[aps,prd,notitlepage,nofootinbib,superscriptaddress,groupedaddress]{revtex4-2}

\usepackage[latin9]{inputenc}
\usepackage{mathrsfs}
\usepackage{bm}
\usepackage{amsmath}
\usepackage{amsthm}
\usepackage{amssymb}
\usepackage{graphicx}
\usepackage{esint}
\usepackage{hyperref}
\usepackage{tabularx}

\theoremstyle{plain}

\theoremstyle{plain}

\ifx\proof\undefined
\newenvironment{proof}[1][\protect\proofname]{\par
	\normalfont\topsep6\p@\@plus6\p@\relax
	\trivlist
	\itemindent\parindent
	\item[\hskip\labelsep\scshape #1]\ignorespaces
}{%
	\endtrivlist\@endpefalse
}
\providecommand{\proofname}{Proof}
\fi

\usepackage[table ]{ xcolor}

\usepackage{amssymb}
\usepackage{graphicx,color}

\usepackage{amsmath,amsthm}

\usepackage{mathrsfs}

\usepackage{bm}

\def\beq{\begin{equation}}
\def\eeq{\end{equation}}
\def\bi{\begin{itemize}}
\def\ei{\end{itemize}}
	\def\ba{\begin{array}}
	\def\ea{\end{array}}
	\def\bfig{\begin{figure}}
	\def\efig{\end{figure}}

	\def\C{\mathbb{C}}
	\def\R{\mathbb{R}}
	\def\Z{\mathbb{Z}}

	\newtheorem{theorem}{Theorem}[section]
	\newtheorem{definition}{Definition}[section]
	
	\newtheorem{lemma}[theorem]{Lemma}

	\newcommand{\I}{{\rm i}}

	\newcommand{\Slc}{\mathrm{SL}(2,\mathbb{C})}

	\newcommand{\Su}{\mathrm{SU}(2)}

	\def\be{\begin{eqnarray}}
	\def\ee{\end{eqnarray}}

	\newcommand{\ca}{\mathcal A}

	\newcommand{\cc}{\mathcal C}
	
	\newcommand{\ce}{\mathcal E}
	\newcommand{\cf}{\mathcal F}
	\newcommand{\cg}{\mathcal G}
	\newcommand{\ch}{\mathcal H}
	\newcommand{\ci}{\mathcal I}
	\newcommand{\cj}{\mathcal J}
	\newcommand{\ck}{\mathcal K}
	\newcommand{\cl}{\mathcal L}
	
	\newcommand{\cn}{\mathcal N}
	
	\newcommand{\calp}{\mathcal P}

	\newcommand{\cs}{\mathcal S}
	\newcommand{\ct}{\mathcal T}
	\newcommand{\cu}{\mathcal U}
	\newcommand{\cv}{\mathcal V}
	
	\newcommand{\cx}{\mathcal X}

	\newcommand{\sa}{\mathscr{A}}

	\newcommand{\ff}{\mathfrak{f}}

	\newcommand{\fl}{\mathfrak{l}}  \newcommand{\Fl}{\mathfrak{L}}
	  
	\newcommand{\fn}{\mathfrak{n}}  \newcommand{\Fn}{\mathfrak{N}}
	  
	\newcommand{\fp}{\mathfrak{p}}

	\newcommand{\fs}{\mathfrak{s}}  \newcommand{\Fs}{\mathfrak{S}}

	\renewcommand{\a}{\alpha}
	\renewcommand{\b}{\beta}
	\newcommand{\g}{\gamma}
	\newcommand{\G}{\Gamma}
	
	\newcommand{\eps}{\varepsilon}

	\newcommand{\sig}{\sigma}
	
	\renewcommand{\l}{\lambda}
	\renewcommand{\L }{\Lambda}
	\renewcommand{\o}{\omega}
	\renewcommand{\O}{\Omega}
	\renewcommand{\t}{\tau}

	\newcommand{\dd}{\mathrm d}
	\newcommand{\rmd}{\mathrm d}

	\newcommand{\lt}{\left}
	\newcommand{\rt}{\right}

	\newcommand{\lag}{\left\langle}
	\newcommand{\rag}{\right\rangle}

	\newcommand{\sk}{\mathscr{K}}
	
	\newcommand{\re}{\mathrm{Re}}

	\newcommand{\tr}{\mathrm{Tr}}

\newcommand{\wt}{\widetilde}

\newcommand{\SU}{\ensuremath{\mathrm{SU}(2)}}

\newcommand{\e}{\mathrm{e}}

\newcommand{\abs}[1]{\ensuremath{\left|#1\right|}}

\begin{document}

\title{Ultraviolet Fixed Point in Covariant Loop Quantum Gravity}

\author{Muxin Han}
\email{hanm(At)fau.edu}
\affiliation{Department of Physics, Florida Atlantic University, 777 Glades Road, Boca Raton, FL 33431, USA}
\affiliation{Institut f\"ur Quantengravitation, Universit\"at Erlangen-N\"urnberg, Staudtstr. 7/B2, 91058 Erlangen, Germany}

\begin{abstract}

We investigate the ultraviolet behavior of 4-dimensional Lorentzian covariant Loop Quantum Gravity (LQG) and address the problem of infinite ambiguities relating to the triangulation dependence of spinfoam amplitudes. 
We consider the complete LQG amplitude that summing spinfoam amplitudes over 2-complexes. By introducing spin-network stacks and their covariant extension, spinfoam stacks, the summation over complexes is partitioned into distinct families. 
We demonstrate that the theory exhibits a condensation phenomenon, where quantum geometry condenses to a dominant small spin configuration. 
We identify a candidate fixed point controlling the ultraviolet (small spin) regime of covariant LQG. 
At this fix point, the complete LQG amplitude dynamically reduces to a topological theory at leading order, and the infinite ambiguities of triangulation dependence reduces to a finite set of boundary coefficients associated with a finite basis of 3-dimensional boundary blocks.
These results provide a definition for the continuum limit of spinfoam theory at the fundamental level.

\end{abstract}

\maketitle



In covariant loop quantum gravity (LQG), the transition amplitudes between spin-network boundary states are defined via a spinfoam path integral, wherein one sums over states that decorate a chosen 2-complex~\cite{EPRL,KKL,Rovelli:2010vv}. The 2-complex serves as a discretization of spacetime, dual to a cellular decomposition; for any fixed choice of 2-complex, the amplitude is specified in a well-defined manner. 
However, a fundamental challenge arises from the fact that the discretization is not merely a technical device or auxiliary regulator: the transition amplitude depends on the choice of 2-complex, and modifications to the underlying complex, such as refining or altering its combinatorial structure, generally yield inequivalent amplitudes. 
So the theory has infinite ambiguities corresponding to choices of 2-complexes.
A useful approach toward resolving the issue is to sum over possible 2-complexes (see e.g. \cite{Reisenberger:2000fy,Oriti:2009wn,Oriti:2014yla}).
Nevertheless, the ambiguities of the theory persist within this framework: for each 2-complex contributing to the sum, one may assign an independent weight, the most general sum includes an infinite collection of arbitrary coefficients, each corresponding to a distinct 2-complex. The resulting microscopic parameter space is still infinite-dimensional. 
The situation structurally analogous to that of a nonrenormalizable quantum field theory, in which infinitely many couplings must be fixed to define the theory. 

To address this ambiguity in the continuum limit, a variety of renormalization group (RG) approaches have been developed for LQG and related models~\cite{Dittrich:2014mxa,Bahr:2009mc,Rivasseau:2011hm}. The goal of these approaches is to identify the ultraviolet (UV) fixed point, which should ensure the theory's fundamental independence of discretization choice. In spinfoams, coarse-graining via tensor network and geometric techniques aims to identify fixed points where the theory becomes independent of microscopic details~\cite{Dittrich:2013voa,Dittrich:2016tys,Bahr:2010mc,Bahr:2011uj,Steinhaus:2020lgb,Delcamp:2016dqo}. Complementing the covariant approach, a Hamiltonian renormalization program is developed for canonical LQG, toward a lattice-independent fix point \cite{Lang:2017beo,Thiemann:2020cuq,RodriguezZarate:2025ipb,Thiemann:2024vjx}. In Group Field Theory, RG methods have revealed renormalizable models and emergent condensed geometric phases~\cite{BenGeloun:2011rc,Carrozza:2012uv,Carrozza:2013wda,Carrozza:2016tih,Benedetti:2014qsa,Oriti:2013hta,Krajewski:2012aw}. 
The causal dynamical triangulations (CDT) program sums over triangulations to recover smooth spacetime via second-order phase transitions~\cite{Ambjorn:2012jv,Ambjorn:2004qy,Loll:2019rdt}.
These ideas connect with the broader asymptotic safety scenario, which postulates a non-Gaussian UV fixed point underlying quantum gravity~\cite{Weinberg:1980gg,Reuter:1996cp,Percacci:2007sz,Eichhorn:2017egq,Reuter:2012id,Eichhorn:2018phj,Daum:2010qt,Benedetti:2009gn,Ferrero:2024rvi}.

The central result of this article is that, with a nonperturbative summation organized by \emph{spinfoam stacks}, we identify a candidate fixed point controlling the ultraviolet (UV) regime of covariant LQG (Fig.~\ref{fig:fixpt}). At this fix point, the complete bulk LQG amplitude that summing over 2-complexes dynamically reduces to a topological theory at leading order, and the dependence on ambiguous weights of complexes collapses from infinitely many coefficients to a \emph{finite} set of boundary coefficients associated with a finite basis of \emph{boundary blocks}.
These results provide a definition for the continuum limit of spinfoam theory at the fundamental level.

\begin{figure}[t]
\centering
\includegraphics[width=0.3\columnwidth]{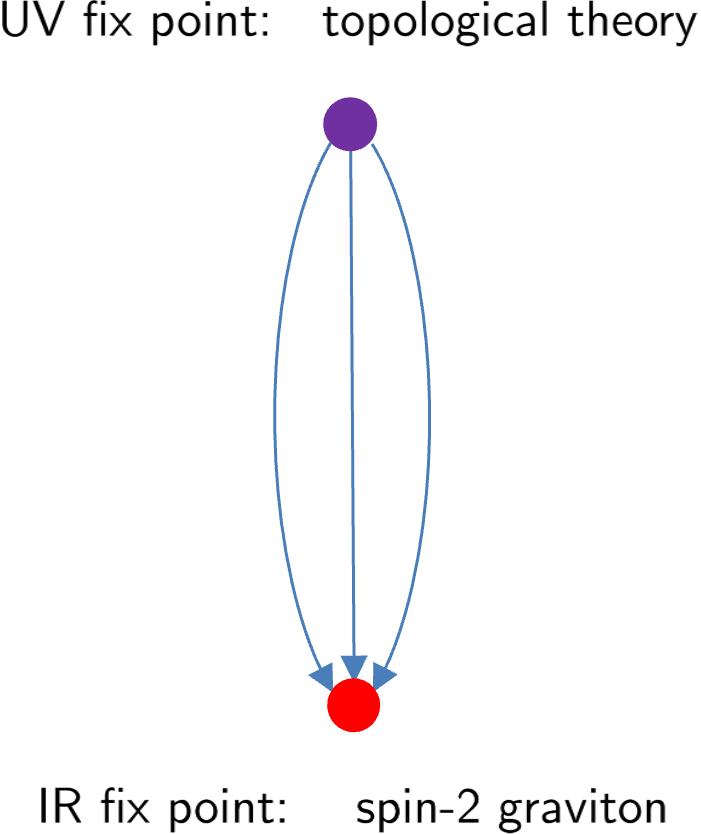}
\caption{Schematic UV fixed-point signal: in a large-cutoff regime the bulk stack amplitude localizes and becomes effectively triangulation independent, leaving finitely many boundary data.}
\label{fig:fixpt}
\end{figure}

\emph{Stacks.--}
The starting point is to build the summing-over-graph structure into the kinematics.
On a spatial slice, a conventional spin-network state lives on a fixed graph.
We instead consider a family $F(\Gamma)$ of graphs generated from a \emph{root graph} $\Gamma$. 
A root graph is a closed oriented graph with nodes $\mathfrak{n}$ and oriented links $\mathfrak{l}$ such that any pair of nodes is connected by at most one link (the pair of nodes can be identical in the case of a loop). 
The family $F(\Gamma)$ contains graphs $\Gamma(\vec p)$ given by stacking $p_{\mathfrak{l}}\in\mathbb{Z}_+$ parallel links on each root link $\mathfrak{l}\subset\Gamma$. A \emph{spin-network stack} is a general superposition of spin-network states on $\Gamma(\vec p)$ by summing over all $\vec{p}=\{p_{\mathfrak{l}}\}_{\mathfrak{l}\subset\Gamma}$, spin assignments $\vec{k}$, and intertwiner labels. 

Geometrically, stacked links represent an arbitrarily fine subdivision of the corresponding face in the dual cellular picture: a single ``flux line'' is replaced by multiple strands carrying quanta of area.
Because stacked strands are physically indistinguishable, we impose permutation invariance by projecting onto the subspace of states that are invariant under permutations of stacked links on each $\mathfrak{l}$.
This \emph{bosonic} projection is motivated by the fact that discrete permutations are part of the gauge redundancy associated with diffeomorphisms \cite{Broukal:2025cgx,Baytas:2026haw}.
The resulting boundary Hilbert space, spanned by permutation-invariant superpositions over all $\vec p$, is the minimal enlargement needed to represent arbitrarily fine subdivisions while keeping a fixed root combinatorics~\cite{Han:2025emp,Han:2019emp}. Permutation invariance implies that states are labelled by unordered tuples of spins; this is why the covariant sums later admit a bosonic occupation-number description.


A spinfoam is understood as a covariant history of a spin-network: starting from a three-dimensional spin-network, its links $\mathfrak{l}$ and nodes $\mathfrak{n}$ evolve into faces $f$ and edges $e$ in four spacetime dimensions. A \emph{root 2-complex} $\mathcal{K}$ is defined such that for every boundary cycle of a face $f$, the face multiplicity is $p_f=1$ and no two faces share the exact same boundary cycle. The covariant history of stacked links is \emph{stacked faces}: for a given root 2-complex $\mathcal{K}$, we generate a family $F(\mathcal{K})$ of complexes $\mathcal{K}(\vec p)$ by replacing each root face $f$ with $p_f$ identical copies, all sharing the same oriented boundary cycle as $f$.
Stacking thus subdivides dual cellular faces while preserving the 1-skeleton. We associate to each root face a complex coupling $\lambda_f\in\C$ and regulate the spin sums by a cutoff $A_f$: for the $p_f$ stacked faces on $f$ carrying SU(2) spins $k_1/2,\dots,k_{p_f}/2$, we impose $\sum_{i=1}^{p_f}E_{k_i}\le A_f$, with a fixed positive and polynomially increasing sequence $E_k$ (e.g. $E_k=k(k+2)$).
These couplings $\lambda_f$ play a role analogous to fugacities: they weight how strongly the sum explores refinements of each face.
With these ingredients we define the \emph{stack amplitude} as a sum over complexes in the family:
\begin{equation}
\mathscr{A}_{\mathcal{K}}
=\sum_{\vec p}\left(\prod_{f\subset\mathcal{K}}\lambda_f^{p_f}\right)\mathscr{A}\bigl(\mathcal{K}, \vec p\bigr),
\label{eq:stackamp}
\end{equation}
where $\mathscr{A}(\mathcal{K},\vec p)$ is the Lorentzian EPRL/KKL amplitude on the complex $\mathcal{K}(\vec p)$, adapted to the permutation-invariant boundary states. 
In the EPRL-KKL amplitude, the sum over spins for each stack of faces is regularized by imposing the cutoff $A_f$. 
The boundary state of stack amplitude is a spin-network stack. 

In the stack amplitude, the summation over complexes within each family $F(\mathcal{K})$ is governed solely by the coupling constants $\{\lambda_f\}$.
This constraint enables us to pinpoint the UV fixed point: the fix point is discovered exactly on the submanifold of the infinite parameter space where these weights are given by the couplings $\{\lambda_f\}$. Allowing for more general weighting schemes amounts to perturbing away from this fixed-point.

Crucially, the UV regime under consideration corresponds to the small-spin regime , which is opposite to the large-spin (semiclassical) limit:
the condensation mechanism to be discussed below keeps the dominant spins at a small nonzero spin $k_0/2$ even for large $A_f$. 
The cutoff $A_f$ provides the large parameter that leads to an expansion of the spinfoam path integral.
The controlled expansion parameter is $A_h/E_{k_0}$, so one can take the cutoff large while holding fixed the typical microscopic area scale through $k_0$.
In the full spinfoam theory, the cutoffs $A_f$ should be understood as large but finite constants, intuitively related to the bare cosmological constant \cite{Han:2021tzw}.
Consequently, a significant suppression of fluctuations in the expansion only occurs when $A_h/E_{k_0}$ is large, which for any finite $A_h$ corresponds to small $k_0$.
Therefore, the vicinity of the fixed point is the regime of small $k_0$, where the topological theory is at the leading order.
The smallness of $k_0$ identifies this fixed point as ultraviolet, since $k_0$ is proportional to the expectation value of the quantum area.

A precise characterization of the notion of scale is essential in background-independent quantum gravity. Unlike conventional quantum field theory, where length or energy scales are referenced with respect to a fixed background geometry, in quantum gravity the geometry itself is dynamical, necessitating a revised definition of scale. In the framework of covariant LQG, spin variables are naturally related to geometric scales, as they correspond to the quantization of area. However, since spins are themselves dynamical variables in LQG, it is appropriate to identify the expectation value of the spin with the physical scale characterizing the regime of the theory. In this context, large expectation values of spin correspond to the IR regime, while small expectation values defines the UV regime. Specifically, the condensation spin $k_0/2$ discussed above is the expectation value of the spin in the path integral. Thus, the UV regime is characterized by a small $k_0$, while the IR regime corresponds to a large $k_0$.

To address the original triangulation dependence problem, we still need to sum over inequivalent root 2-complexes with the same boundary graph $\Gamma=\partial\mathcal{K}$.
This introduces the microscopic ambiguity coefficients $c_{\mathcal{K}}$, one for each root complex,
\begin{equation}
\mathscr{A}=\sum_{\mathcal{K}:\,\partial\mathcal{K}=\Gamma}c_{\mathcal{K}}\,\mathscr{A}_{\mathcal{K}}.
\label{eq:complete}
\end{equation}
Without further input, Eq.~\eqref{eq:complete} appears to be an intractable definition of the dynamics, since the space of ambiguities $\{c_{\mathcal{K}}\}$ and $\{\lambda_f\}$ is infinite dimensional.
The purpose of our analysis is to show these infinite ambiguities are removed at the UV fix point.

\emph{A boson gas of faces.--}
The key mechanism is already visible at the level of a single stack of internal faces $h$.
Because the stacked faces are indistinguishable, the sum over spins on a stack can be reorganized in terms of occupation numbers, as for a noninteracting Bose gas.
Technically, the stacked face amplitude $\omega_h^{\rm bos}$ in the stack amplitude can be computed by Laplace transform method: The Laplace transform of $\omega_h^{\rm bos}$ is a bosonic grand-canonical partition function $\Xi_h(s)$ (where $s$ plays the role of an inverse temperature) of the form
\begin{equation}
\Xi_h(s)=\prod_{k=1}^{\infty}\frac{1}{1-\lambda_h\,\tau^{(h)}_{k}(g_h)\,e^{-sE_k}}-1,
\label{eq:Xi}
\end{equation}
where $k/2$ labels the SU(2) spin, $g_h=\{g_{ve}\}_{v\subset\partial h}$ are half-edge $\Slc$ holonomies around $\partial h$, and $\tau^{(h)}_{k}(g_h)$ is $d_k$ times a simplicity-projected $\Slc$ character and satisfies a uniform bound $|\tau^{(h)}_{k}(g_h)|\le d_k^{2}$, with $d_k=k+1$. Expanding Eq.~\eqref{eq:Xi} yields a sum over occupation numbers $\{n_k\}_{k\ge1}$ with weights $\prod_k[\lambda_h\,\tau_k^{(h)}(g_h) e^{-sE_k}]^{n_k}$, making the Bose-gas structure explicit.

The corresponding microcanonical $\omega_h^{\rm bos}$ is the inverse Laplace transform
\begin{equation}
\omega_h^{\mathrm{bos}}
=\mathscr{P}\!\!\int_{T-\I\infty}^{T+\I\infty}\frac{\dd s}{2\pi \I s}\,\e^{A_h s}\,\Xi_h(s),
\label{eq:omegah}
\end{equation}
where $T$ lies to the right of all poles of $\Xi_h$ and $\mathscr{P}$ denotes principal value.

As a side remark: Equation~\eqref{eq:Xi} has a direct physical interpretation.
For fixed $k$, the single-face factor $\lambda_h\,\tau^{(h)}_{k}(g_h)\e^{-sE_k}$ plays the role of a single-particle partition function.
More concretely, for the choice $E_k=k(k+2)$, which matches the energy spectrum of 2d SU(2) Yang--Mills theory on a disk, the single-face factor can be identified with the partition function of the Yang--Mills theory on a disk with a boundary condition determined by $k$ and the boundary holonomy $g_h$.
Thus $\Xi_h$ describes a grand-canonical ensemble of indistinguishable world-sheets carrying 2d gauge theory, stacked over the root face.
This statistical-mechanics viewpoint is an efficient way to extract the large-cutoff behavior of the face factor from the analytic structure of $\Xi_h$.

The key point is that the same condition maximizing statistical weight enforces a geometric constraint.
The bound $|{\tau^{(h)}_k(g_h)}|\le d_k^2$ is saturated only when the wedge holonomies reduce to SU(2) elements and the face holonomy is $\pm\mathbb{I}$:
\begin{equation}
    g_{ve}^{-1}g_{ve'}\in\SU,
    \qquad
    \overrightarrow{\prod_{(v;e,e')\subset\partial h}} g_{ve}^{-1}g_{ve'}=\pm\mathbb{I},
    \label{eq:suflat}
    \end{equation}
As $A_h$ grows, the inverse Laplace transform is dominated by residues at the poles of $\Xi_h(s)$ with largest real parts. These poles occur precisely on locus of $g_h$ satisfying \eqref{eq:suflat}.
In this way, the state sum for internal faces act as nonperturbative filters: generic $\Slc$ holonomies are suppressed, while the SU(2) subsector is enhanced.
This provides the origin of the path integral localization derived below.

\emph{Condensation spin and the UV scale.--} The poles of $\Xi_h(s)$ are given by
\begin{equation}
s_h(k,m,g_h)=\frac{\ln\!\bigl[\lambda_h\,\tau^{(h)}_{k}(g_h)\bigr]+2\pi \I m}{E_k},
\label{eq:poles}
\end{equation}
where $\ln(\cdot)$ is the principal logarithm.
Here $k\in\mathbb{Z}_{+}$ and $m\in\mathbb{Z}$.
The leading contribution to $\omega_h^{\rm bos}$ is governed by the global maximum of $\mathrm{Re}(s_h(k,m,g_h))$ and thus relates to a \emph{condensation spin} $k_0/2$ satisfying
\begin{equation}
\beta_{k_0}(\lambda_h)=\sup_{k\ge1}\beta_k(\lambda_h),\qquad \beta_k(\lambda_h)=\frac{\ln\!\bigl(\abs{\lambda_h}\,d_k^2\bigr)}{E_k}.
\label{eq:k0}
\end{equation}
The situation is directly analogous to Bose--Einstein condensation: beyond a critical fugacity the ensemble accumulates in the lowest available energy level, here labeled by $k_0$.
In geometric terms, the spin $k_0/2$ sets the typical quantum area scale selected by the state sum.

Indeed, evaluating Eq.~\eqref{eq:omegah} by residues at the dominant pole yields the leading behavior $\omega_h^{\rm bos}(A_h;g_h,\lambda_h)\propto [\lambda_h\,\tau^{(h)}_{k_0}(g_h)]^{A_h/E_{k_0}}$, up to a subexponential prefactor and corrections suppressed faster than any power of $A_h^{-1}$.
The exponent $p_h:=A_h/E_{k_0}\gg1$ can be read as the occupation number of faces in the condensate, in direct analogy with Bose--Einstein condensation: the cutoff $A_h$ plays the role of a total ``energy'' stored in $p_h$ quanta each of ``energy'' $E_{k_0}$. 
A large ratio $A_h/E_{k_0}$ plays a role analogous to the thermodynamic limit: it suppresses fluctuations away from the dominant contribution, and it is in favor of small $k_0$. In this sense the condensation mechanism selects a UV regime (with small typical area). The condensation spin $k_0/2$ controls the dominant microscopic scale.
Since in LQG spin labels quanta of area, condensation to small $k_0$ implies that a macroscopic area is obtained as a superposition of many microscopic quantum areas rather than by exciting a single large spin.
As $\abs{\lambda_h}$ decreases, the maximizer $k_0$ shifts to larger $k$ and eventually all $\beta_k(\lambda_h)\le0$ or close to zero, so the condensed regime disappears, and the stack sum reduces effectively to the spinfoam amplitude on the root 2-complex, corresponding to an infrared description.

\emph{Localization of stack amplitude.--}
The stack amplitude~\eqref{eq:stackamp} generated from a root complex $\mathcal{K}$ with internal faces $h$ and boundary faces $b$ can be written as
\begin{equation}
\mathscr{A}_{\mathcal{K}}
=\int \dd\Omega(g)\,
\prod_{h}\omega_h^{\mathrm{bos}}(g)\,
\prod_{b}\omega_b^{\mathrm{bos}}(g),
\label{eq:stackint}
\end{equation}
which is an integral over the bulk holonomies $g_{ve}$ with product $\Slc$ Haar measure $\dd\Omega(g)$, modulo gauge freedom. The factor $\omega_b^{\mathrm{bos}}$ is the amplitude associated to the stack of boundary faces.
By the leading behavior of $\omega_h^{\mathrm{bos}}$, the stack amplitude becomes a stationary-phase form, with the effective action $S(g)=\sum_h \frac{A_h}{E_{k_0}}\,\ln \left[{\tau^{(h)}_{k_0}(g_h)}/{d_{k_0}^2} \right]$ plus a positive constant.
The maximum $\mathrm{Re}(S)=0$ is reached if and only if the conditions \eqref{eq:suflat} hold on every internal face.
Moreover, the same condition implies $\delta_g S=0$. 
We define the submanifold $\mathcal{C}_{\mathrm{int}}$ as the locus in the space of $\mathrm{SL}(2,\mathbb{C})$ holonomies where the conditions \eqref{eq:suflat} are satisfied for all internal faces. The submanifold $\mathcal{C}_{\mathrm{int}}$ is called the critical manifold, since it contain all critical points of $S$ making dominant contribution to $\mathscr{A}_{\cal K}$. When $A_h/E_{k_0}$ are uniformly large, the stack amplitude localizes on $\mathcal{C}_{\mathrm{int}}$.

The structure of $\mathcal{C}_{\mathrm{int}}$ is transparent.
On each internal face Eq.~\eqref{eq:suflat} restricts all holonomies to SU(2) and fixes the face holonomy to $s_h\mathbb{I}$ with $s_h=\pm1$.
Collecting the signs $\{s_h\}$ labels disjoint sectors $\mathcal{C}_{\mathrm{int}}=\bigcup_{\{s_h\}}\mathcal{C}_{\mathrm{int}}^{\{s_h\}}$. 
The critical manifold $\mathcal{C}_{\mathrm{int}}$ contains the gauge freedom of the SU(2) holonomies.
After removing these gauge freedom, the remaining degrees of freedom are precisely the moduli of flat SU(2) connections (up to sign) on the internal 2-skeleton; their dimension is controlled by the fundamental group of the bulk subcomplex.
In this paper, we focus on the simply connected $\mathcal{K}$ relevant to the local UV continuum limit, then the moduli reduce to only discrete sign sectors.

\emph{Boundary blocks and triangulation independence.--} The boundary face amplitudes $\omega_b^{\rm bos}$ do not depend on gauge freedom, when they are evaluated on $\mathcal{C}_{\mathrm{int}}$, so their dependence on the bulk continuous degrees of freedom completely drops out: the contribution from each boundary face is constant on $\mathcal{C}_{\rm int}^{\{s_h\}}$, and they only depend on discrete sign data induced from the bulk.

Concretely, each stack of boundary faces $b$ acquires a label $\varsigma_{b}\in\{0,\pm1\}$ determined by the induced signs of holonomies along the internal edges in $\partial b$. 
Here, internal edges are edges that do not connect to the boundary.
The value $0$ corresponds to the special case that $b$ does not have any internal edge, then no sign dependence remains.
By the one-to-one correspondence between $b$ and boundary link $\mathfrak{l}$ of the root graph $\Gamma=\partial \mathcal{K}$, we identify $\varsigma_{\mathfrak{l}}=\varsigma_{b}$.
Collecting these labels gives $\bm{\varsigma}\in\{0,\pm1\}^{|\mathcal L|}$, where $|\mathcal L|$ is the number of links in $\Gamma$. The asymptotics of the stack amplitude becomes a finite sum
\begin{equation}
\mathscr{A}_{\mathcal{K}}=\sum_{\bm\varsigma} b_{\bm\varsigma}(\mathcal{K})\,B_{\bm\varsigma}\left[1+O(A^{-1})\right].
\label{eq:stackblocks}
\end{equation}
The \emph{boundary blocks} $B_{\bm\varsigma}$ equals $\int \prod_{(v,e_b)} \dd g_{v e_b}\prod_b\omega_b^{\rm bos}(g)$ restricted onto $\mathcal{C}_{\rm int}$ and depend only on boundary data ($\bm \varsigma$, boundary holonomies and the boundary couplings). All bulk degrees of freedom are absorbed into finitely many coefficients $b_{\bm\varsigma}(\mathcal{K})$.

Equation~\eqref{eq:stackblocks} already exhibits the UV reduction: for fixed boundary graph, the set $\{B_{\bm\varsigma}\}$ spans a finite-dimensional vector space of dimension $3^{|\mathcal{L}|}$.
When neglecting $O(A^{-1})$, every root complex $\mathcal{K}$ with boundary $\Gamma$ defines a vector $\mathscr{A}_{\mathcal{K}}$ in this space.
The crucial step is to return to the complete amplitude~\eqref{eq:complete}.
Substituting Eq.~\eqref{eq:stackblocks} yields that the complete amplitude is also a vector in the vector space
\begin{equation}
\mathscr{A}=\sum_{\bm\varsigma} b_{\bm\varsigma}\,B_{\bm\varsigma}\left[1+O(A^{-1})\right],
\label{eq:boundaryblocks}
\end{equation}
where the \emph{renormalized} coefficients
\begin{equation}
b_{\bm\varsigma}=\sum_{\mathcal{K}:\,\partial\mathcal{K}=\Gamma}c_{\mathcal{K}}\,b_{\bm\varsigma}(\mathcal{K})
\label{eq:bvarsigma}
\end{equation}
package all bulk degrees of freedom including microscopic triangulation dependence into finitely many numbers.
This is the \emph{nut} of the article:
\emph{at the leading order, the infinite ambiguity of summing over 2-complexes collapses to finitely many boundary coefficients $b_{\bm\varsigma}$.} 
Triangulation independence here plays the role of scale invariance in a background-independent setting, indicating that at leading order the theory sits at a fixed point of the full covariant LQG.

Several consequences follow immediately.
(i) The predictive content of the theory at the fixed point is controlled by a \emph{finite} set of parameters $\{b_{\bm\varsigma}\}$, despite the infinite microscopic parameter space spanned by $\{c_{\mathcal{K}},\lambda_f\}$.
(ii) The theory at the fix point is a topological theory without propagating bulk degree of freedom. The surviving degrees of freedom are edge-mode-like and encoded by $\bm\varsigma$ sectors.
The reduction of the bulk partition function to a linear combination of boundary blocks is characteristic of topological quantum field theories. A standard example is the relation between Chern-Simons theory and two-dimensional conformal blocks.
(iii) There are two standard continuum limit strategies in spinfoams: refining a complex versus summing over complexes \cite{Rovelli:2010qx,Asante:2022dnj,Oriti:2009wn}. Here they become equivalent, because at the level of stack amplitude, the leading order depends only on boundary blocks and not on bulk refinement details.

From an effective-field-theory perspective, the coefficients $b_{\bm\varsigma}$ in Eq.~\eqref{eq:boundaryblocks} play the role of renormalized couplings that parametrize the UV completion on the fixed boundary root graph.
Different microscopic choices $\{c_{\mathcal{K}}\}$ that yield the same finite set $\{b_{\bm\varsigma}\}$ are indistinguishable at short distances: bulk refinements change only subleading $O(A^{-1})$ terms.
Equivalently, specifying the UV theory reduces to choosing a vector in the finite-dimensional boundary-block space, rather than an infinite list of triangulation weights.
This turns the continuum problem into a finite matching problem: one can aim to determine $\{b_{\bm\varsigma}\}$ from a finite set of boundary observables and then predict the response to arbitrary bulk refinements in the UV regime.

\emph{Interpretation and scope.--}
In our analysis, the UV fixed point is identified operationally, instead of being derived from any RG equation.
In the large internal cutoff limit, the complete LQG amplitude becomes topological and triangulation independent. 
It provides a concrete notion of UV fixed-point behavior in a background-independent theory.
Understanding the RG behavior then becomes a question of determining which deformations are relevant and induce flows away from the fixed point toward the infrared (IR) regime.
This is in analogy with perturbative quantum field theories that expands around a UV fixed point, such as perturbative QCD at short distances, although the fix point here is topological, in contrast to the free-theory fix point in QCD.

From the broader physics perspective, the result provides a concrete mechanism of fixed-point in a background-independent setting.
The bosonic condensation to a small spin provides a dynamical selection of a microscopic area scale, while the large-cutoff limit produces a topological bulk that is insensitive to short-distance discretization.
Together they realize, in the spinfoam context, the RG intuition that a UV completion can be predictive even when the microscopic parameter space is infinite dimensional.
The boundary blocks provide the finite set of UV data that may be fixed by matching to physical observables.

From the viewpoint of asymptotic-safety ideas, the UV fixed point suggested here is not a perturbative Gaussian point but an emergent topological phase: local bulk propagation is absent at leading order, yet nontrivial boundary data and controlled deformations remain.
This suggests that RG universality in quantum gravity may be realized through phases, captured here by the boundary-block space and by the finite set of coefficients $\{b_{\bm\varsigma}\}$.

\emph{Conclusion.--}
We introduced permutation-invariant spin-network and spinfoam stacks to organize the sum over 2-complexes in Lorentzian covariant LQG.
The state-sum of bosonic stacked faces admit a reformulation as grand-canonical partition-function, which yields a condensation phenomena. The condensation spin $k_0/2$ controls the dominant microscopic area scale.
For uniformly large cutoff and small condensation spin, covariant LQG reduce to a topological and triangulation independent fix-point theory.
As a result, the complete amplitude at the fix point becomes a finite linear combination of boundary blocks, so infinite ambiguities of microscopic triangulations reduce to finitely many boundary coefficients.
A natural next step is to turn on the $O(A^{-1})$ perturbations that generate propagating bulk degrees of freedom and to classify which deformations are relevant and drive the RG flow toward the infrared. 
Another is to study the boundary blocks and connect the resulting coefficients to physical observables.

\vspace{1em}
\textbf{Acknowledgements:} This work was made possible through the support of the WOST, WithOut SpaceTime project (https://withoutspacetime.org), supported by Grant ID\# 63683 from the John Templeton Foundation (JTF). The opinions expressed in this work are those of the author(s) and do not necessarily reflect the views of the John Templeton Foundation. The author receives supports from Center for SpaceTime and the Quantum (CSTQ) and the National Science Foundation through grants PHY-2207763 and PHY-2512890.

\vspace{2em}

\textbf{Organization of the paper:} The remainder of this paper is organized as follows. Section~\ref{Permutation-invariant spin-network stacks} introduces the framework of permutation-invariant spin-network stacks. Section \ref{Spinfoam stacks} extends this construction to spacetime by defining spinfoam stack amplitudes and the complete amplitude. 
In Section \ref{State sum and bose gas}, the state sum of a stack of internal face is analyzed with the Laplace transform method and related to a grand canonical partition function of boson system. In Section \ref{2-dimensional gauge theory}, we relate the system to a Bose gas of 2-dimensional Yang-Mills theories. Finally, in Section \ref{Bose-Einstein condensate of quantum geometry}, we demonstrate the condensation of quantum geometry and derive the asymptotic behavior of the state sum for large cut-off.
Section \ref{Localization of stack amplitude} derives the localization of the stack amplitude onto a critical manifold in the large-cutoff limit. The parametrization of this critical manifold and the non-degeneracy is detailed in Sections \ref{Some explicit computations}. The associated Hessian matrix is discussed in \ref{Nondegenerate Hessian matrix}. Section \ref{Boundary blocks} establishes the main result: the reduction of the complete amplitude to a finite set of boundary blocks. Section \ref{Finiteness of boundary block} proves that the boundary block is a linear functional over the vector space spanned by spin-networks.

\newpage

\tableofcontents

\section{Permutation-invariant spin-network stacks}\label{Permutation-invariant spin-network stacks}

A \emph{root graph} is a closed, oriented graph $\G$ where any pair of nodes are connected by at most one link (the pair of nodes can be identical in the case of loop). We define the link multiplicity to be the number of links connecting a given pair of nodes. The link multiplicities of a root graph equal to one. The root graph serves as the basic combinatorial structure underlying our constructions. On such a graph, a spin-network state is specified by assigning to each oriented link $\fl$ a spin $j = k/2$ ($k \in \mathbb{Z}_+$) and to each node $\fn$ a normalized intertwiner $I_\fn$. 

Moving beyond a single root graph $\G$, one can generate an entire family $F(\G)$ of graphs by allowing the link multiplicity between pairs of neighboring nodes to increase---that is, by stacking multiple links between the same pair of nodes and aligning their orientations (see FIG.\ref{sn_stack}). Each graph in the family, denoted by $\G(\vec p)$, has the link multiplicities $\vec p=\{p_\fl\}_{\fl\subset \G}$ (the root graph is $\G=\G(\vec{1})$). For any spin-network state on $\G(\vec{p})$, each of the $p_\fl$ links between a pair of nodes carries its own spin, and the intertwiners become correspondingly higher-valent. A general superposition of these states, referred to as a \emph{spin-network stack}, is constructed by summing over all combinations of link multiplicities $\{p_\fl\}_{\fl\subset \G}$, spin assignments $\vec{k}$, and intertwiner labels. A conventional spin-network state is just a special case of a stack in which all superposition coefficients but one vanish. We restrict attention to abstract graphs with ordered links, without specifying any embedding in a manifold. Consequently, issues related to knotting or topological embedding are excluded from our considerations.

The graphs discussed above can be described formally as labelled multigraphs 

\begin{definition}
A labelled multigraph $\cg$ is a tuple $(\cn, \cl, \Sigma_\cl, \ell_\cl)$, where: 
\begin{itemize}

\item $\cn$ denotes the set of nodes. 

\item $\cl$ is a multiset consisting of ordered pairs of nodes, each corresponding to an oriented link. A pair of nodes may coincide, in which case the link is a loop. For each pair $\fl \in \cl$, the number of occurrences is the link multiplicity $p_\fl$. All stacked links between the same pair of nodes are oriented identically.

\item $\Sigma_\cl$ is a finite sets serving as the alphabets of link labels, satisfying $|\Sigma_\cl|=|\cl|$.

\item $\ell_\cl: \cl \rightarrow \Sigma_\cl$ is a bijective map that assigns a unique label to each link, with the convention that stacked links corresponding to the same pair $\fl\in\cl$ are arranged consecutively in the ordering.

\end{itemize}
A root graph $\G$ is a labelled multigraph with $\cl$ being a set, i.e. $p_\fl=1$.
\end{definition}

Here, multiset allows for multiple instances for each of its elements, e.g. $\{\fl,\fl,\fl,\fl',\fl',\cdots\}$, and the multiset is well-ordered by the map $\ell_\cl$. We use the pair $(\fl, i)$ to label the $i$-th stacked link on the root link $\fl\subset \G$. The label $i=1,\cdots,p_\fl$ corresponds to the ordering to the stacked links. The set of multiplicities $\vec{p}=\{p_\fl\}_{\fl}$ is well-ordered according to the ordering of $\fl$. 


We define an map $\Phi$ from a labelled multigraph $\cg$ to a pair $(\G,\vec{p})$: its root graph $\G$ is obtained by projecting the multiset of links $\cl(\cg)$ onto a set $\cl(\G)$ comprising a single representative for each group of stacked links. Specifically, for each class of stacked links in $\cl$, the projection selects the first ordered link as the root link. The corresponding labelling map $\ell_{\cl(\G)}$ are inherited so that the ordering match those of the original $\cg$. For example, this projection maps $\{\fl,\fl,\fl,\fl',\fl',\cdots\}$ to the well-ordered set $\{\fl,\fl',\cdots\}$. 

Conversely, given any pair $(\G,\,\vec{p})$, the preimage $\mathfrak{I}_{\G,\vec p} = \Phi^{-1}(\G, \vec{p})$ consists of all labelled multigraphs $\cg$ that differ only in their choice of link labels---that is, all possible orderings of the stacked links. To construct an element $\cg \in \mathfrak{I}_{\G,\vec p}$, start from the root graph $\G$ and replace each root link $\fl$ with $p_\fl$ copies in the multiset $\cl(\cg)$; the ordering of these copies within the multiset is ambiguous. Thus, $\mathfrak{I}_{\G,\vec p}$ enumerates all distinct ways of ordering the stacked links. By regarding $\mathfrak{I}_{\G,\vec p}$ as an equivalence class and selecting a distinguished representative $\cg = \G(\vec{p})$ via a fixed ordering, $\Phi$ establishes a bijection between the set of pairs $(\G, \vec{p})$ and the set of these ordered representatives $\cg = \G(\vec{p})$.

We define the kinematical Hilbert space $\ch_{\rm Kin}$ as the direct sum over all labelled multigraphs:
\be
\ch_{\rm Kin} = \bigoplus_{\cg} \ch_{\cg},\label{Hkin}
\ee
where $\ch_{\cg}$ denotes the Hilbert space of spin-network states (with nonzero spins) defined on a labelled multigraph $\cg$. Each multigraph possesses a choice of labelling of its links, and the labelling is not summed over in \eqref{Hkin}. By the map $\Phi$, each $\cg$ is identified to the pair $(\G,\vec p)$, i.e. $\cg = \G(\vec{p})$.
Therefore, $\ch_{\rm Kin}$ admits the following decomposition
\be
\ch_{\rm Kin} = \bigoplus_{\G} \ch_{\G,{\rm st}},\qquad \ch_{\G,{\rm st}}=\bigoplus_{\vec p} \ch_{\G(\vec p)}.
\ee
Here $\ch_{\G,{\rm st}}$ denote the Hilbert space of all spin-network stack states on a given root graph $\G$.  

Geometrically, a spin-network state on the root graph represents the quantum geometry of a cellular decomposition of a spatial slice $\Sigma$, with each intertwiner encoding a flat-faced polyhedron whose number of faces is given by the node's valence. Stack states generalize this picture: stacked links divide a flat face into smaller faces, each with (in general) non-parallel three-dimensional normals, so that the stack encodes a quantum superposition of cellular geometries with arbitrarily discretized faces \cite{Han:2019emp}.

\begin{figure}[t]
\centering
\includegraphics[width=0.7\textwidth]{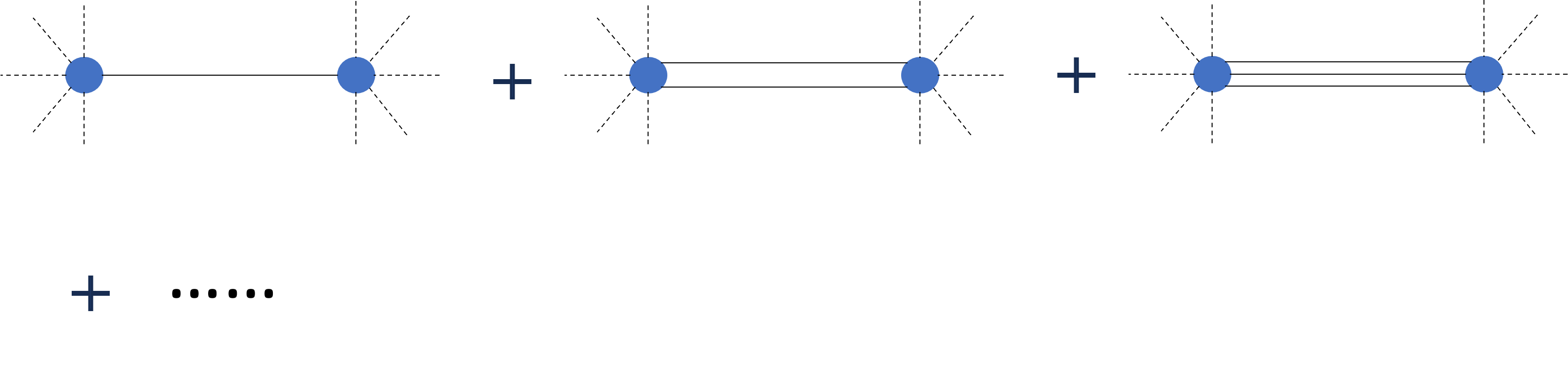}
\caption{The spin-network stack.}
\label{sn_stack}
\end{figure}

The graph $\G(\vec p)$ possesses a symmetry corresponding to permutations of the stacked links at each root link $\fl\subset\G$. Specifically, at every root link $\fl$, the permutation group $S_{p_\fl}$ acts by exchanging the $p_\fl$ identical links: given any spin-network state $\psi_{\G(\vec{p})}$ on the graph, we denote by $\sig_{ij}(\fl)\in S_{p_\fl}$ the permutation exchanging a pair of stacked links $i$ and $j$ at $\fl$, the unitary representation $U_{\sig_{ij}(\fl)}$ of $\sig_{ij}(\fl)$ acting on $\psi_{\G(\vec{p})}$ is given by interchanging the holonomies of the links in the entries
\be
U_{\sig_{ij}(\fl)}\psi_{\G(\vec{p})}\lt(\cdots,H_{\fl}^{(i)},\cdots, H_{\fl}^{(j)},\cdots\rt):=\psi_{\G(\vec{p})}\lt(\cdots,H_{\fl}^{(j)},\cdots, H_{\fl}^{(i)},\cdots\rt).\label{defUsigmaPermut}
\ee
Equivalently, the action exchanges the spins associated to the holonomies $H_{\fl}^{(i)}$ and $H_{\fl}^{(j)}$ and transposes each connected intertwiner. Indeed, if we consider a state where the $i$-th and $j$-th stacked links are colored by spins $k$ and $k'$ respectively, ignoring irrelevant ingredients in the spin-network and following \eqref{defUsigmaPermut},
\be
U_{\sig_{ij}}:  (I_{1})_{ab} D^{k}(H_\fl^{(i)})^{a}_{\ \a } D^{k'}(H_\fl^{(j)})^{b}_{\ \b} (I_{2})^{\a\b} &\mapsto& (I_{1})_{ab}  D^{k}(H_\fl^{(j)})^{a}_{\ \a } D^{k'}(H_\fl^{(i)})^{b}_{\ \b} (I_{2})^{\a\b} = (I_{1})_{ba} D^{k'}(H_\fl^{(i)})^{a}_{\ \a}D^{k}(H_\fl^{(j)})^{b}_{\ \b } (I_{2})^{\b\a}\nonumber\\
&=&(I_{1}^T)_{ab} D^{k'}(H_\fl^{(i)})^{a}_{\ \a}D^{k}(H_\fl^{(j)})^{b}_{\ \b } (I_{2}^T)^{\a\b}.\nonumber
\ee
In the result, the $i$-th and $j$-th stacked links are colored by spins $k'$ and $k$ respectively, while the intertwiners are transposed with respect to the corresponding slots.
In terms of Dirac-ket notation, the action can be written as
\be
U_{\sig_{ij}}\Big|\psi_{\G(p);k,k';I_1,I_2}\Big\rangle=\lt|\psi_{\G(p);k',k;I_1^T,I_2^T}\rt\rangle.
\ee
	 
The permutations $U_\sigma$ extend unitarily to entire $\ch_{\rm Kin}$. Intuitively, this permutation symmetry would be a discrete subgroup within the spatial diffeomorphism group if these graphs were embedded in a 3d spatial manifold\footnote{When the spatial manifold is 2-dimensional (in (2+1)-dimensional gravity), the permutation symmetry does not relate to spatial diffeomorphisms. Therefore, if we apply the formalism in this paper to (2+1)-dimensional gravity, the stacked links and faces are distinguishable}. Viewing the permutations as a part of gauge redundancy motivates us to perform a group average and project the Hilbert space ${\ch}_{\G(\vec p)}$ onto its subspace $\wt{\ch}_{\G(\vec p)}$ of permutation-invariant states\footnote{
	See also \cite{Baytas:2026haw}. A similar idea is in \cite{Broukal:2025cgx} for permutations of nodes.
}. We define the permutation-invariant projection $\mathbb{P}$ on any spin-network state $\psi_{\G(\vec{p})}\in \ch_{\G(\vec p)}$ by average over all permutations of stacked links (at each root link)
\be
\mathbb{P}\psi_{\G(\vec{p})}:= \frac{1}{|G|}\sum_{\sig\in G}U_\sig\psi_{\G(\vec{p})}\equiv \Psi_{\G(\vec{p})},\qquad  G= \times_\fl S_{p_\fl},\qquad U_\sig=\prod_\fl U_{\sig(\fl)}.
\ee
where $|G|$ denotes the number of elements in $G$. The resulting permutation invariant spin-network state $\Psi_{\G(\vec{p})}$ spans the subspace $\wt{\ch}_{\G(\vec{p})}=\mathbb{P}{\ch}_{\G(\vec{p})}$. A linear combination of $\Psi_{\G(\vec{p})}$ over $\vec{p}$ is referred to as permutation-invariant spin-network stacks. Summing over all root graphs, we construct the total Hilbert space of permutation-invariant spin-network stacks as
\be
\ch_{\rm Sym} =  \bigoplus_{\G} \wt{\ch}_{\G,\rm{st}},\qquad \wt{\ch}_{\G,\rm{st}}=\bigoplus_{\vec{p}}\wt{\ch}_{\G(\vec{p})}.
\ee

For the permutation-invariant spin-networks $\wt\Psi_{\G(\vec{p})}$ (normalization of $\Psi_{\G(\vec{p})}$) forming an orthonormal basis, spin labels on stacked links at each $\fl$ are unordered--the basis states are specified by tuples of spins $(k_1,\dots,k_{p_\fl})$, with all permutations identified. To designate each basis vector uniquely, one may adopt the convention $k_1\leq k_2 \leq \cdots \leq k_{p_\fl}$ for every $\fl$ (here $k_i$ colors the $i$-th stacked link), so that every configuration appears once in the basis.

\begin{definition}
Let $\ch_{\G(\vec p)}$ be the Hilbert space of spin-networks on the graph $\G(\vec p)$ with link multiplicities $\{p_\fl\}$. We define the subspace of {\bf ordered seed states}, denoted $\ch_{\G(\vec p),{\rm sd}} \subset \ch_{\G(\vec p)}$, as the span of spin-network states $\psi$ satisfying two conditions:

1.  Spin Ordering: For every root link $\fl$, the assigned spins satisfy $1\leq k_1 \leq k_2 \leq \cdots \leq k_{p_\fl}$.

2.  Stabilizer Symmetry: If $G_{\vec{k}} \subset \times_\fl S_{p_\fl}$ denote the stabilizer subgroup leaving the tuple of spin labels $\vec{k}$ invariant, then $U_\sigma \psi = \psi$ for all $\sigma \in G_{\vec{k}}$.

\end{definition}

We denote by $K^0$ the set of spin configurations satisfying the spin ordering $1\leq k_1 \leq k_2 \leq \cdots \leq k_{p_\fl}$ on all $\fl\subset\G$. The Hilbert space $\mathcal{H}_{\G(\vec p),{\rm sd}}$ is a direct sum over $\vec{k} \in K^0$: $\mathcal{H}_{\G(\vec p),{\rm sd}}=\bigoplus_{\vec{k}\in K^0 }\mathcal{H}_{\vec{k}}^{G_{\vec{k}}}$

\begin{lemma}

For any $\psi_{\vec{k}}\in \ch_{\vec k}^{G_{\vec k}}$, we define the linear map $\mathbb{U}: \ch_{\G(\vec p),{\rm sd}}\to \wt{\ch}_{\G(\vec p)}$ by
\be
\mathbb{U}\psi_{\vec k}=\sqrt{\frac{|G|}{|G_{\vec k}|}}\mathbb{P}\psi_{\vec k} .
\ee
The map $\mathbb{U}$ is unitary. These two Hilbert spaces are unitarily equivalent.


\end{lemma}

\begin{proof} It suffices to prove $\mathbb{U}$ is surjective and preserves the inner product. We denote by $K^0$ the set of spin configurations satisfying the spin ordering $k_1 \leq k_2 \leq \cdots \leq k_{p_\fl}$ on all $\fl\subset\G$. 


Surjectivity: For any $\Psi \in \widetilde{\mathcal{H}}_{\G(\vec p)}$, we want to construct a preimage $\psi \in \mathcal{H}_{\G(\vec p),{\rm sd}}$ such that $\Psi= \mathbb{U}\psi$. The spin-network decomposition of $\Psi \in \widetilde{\mathcal{H}}_{\G(\vec p)}\subset {\mathcal{H}}_{\G(\vec p)}$ gives $ \Psi = \sum_{\vec{k}'} \Psi_{\vec{k}'} $ with $\Psi_{\vec{k}'} \in \mathcal{H}_{\vec{k}'}$, where $\mathcal{H}_{\vec{k}'}$ is the subspace of $\ch_{\G(\vec{p})}=\oplus_{\vec k'}\mathcal{H}_{\vec{k}'}$. Due to invariance $U_\sigma \Psi = \Psi$, we have $\sum_{\vec{k}'} U_\sigma \Psi_{\vec{k}'} = \sum_{\vec{k}'} \Psi_{\vec{k}'} $. Comparing components in the subspace $\mathcal{H}_{\sigma \cdot \vec{k}'}$, we find $ U_\sigma \Psi_{\vec{k}'} = \Psi_{\sigma \cdot \vec{k}'} $. Therefore, the set of $\vec{k}'$ in the sum $ \Psi = \sum_{\vec{k}'} \Psi_{\vec{k}'} $ is invariant under $G$ and decomposed into orbits of $G$, so $\Psi$ decomposes into invariant vectors within each orbit of spin configurations. It suffices to consider $\Psi=\sum_{\vec{k}'\in \mathcal{O}_{\vec{k}^0}} \Psi_{\vec{k}'} $ associated with a single orbit $\mathcal{O}_{\vec{k}^0}$ generated by an ordered configuration $\vec{k}^0 \in {K}^0$. We define the candidate seed state $\psi$ as the projection of $\Psi$ onto the ordered sector followed by a rescaling:
\be 
\psi \equiv \sqrt{\frac{|G|}{|G_{\vec k^0}|}} \Psi_{\vec{k}^0} 
\ee
First, $\psi$ is in the seed space $\mathcal{H}_{\G(\vec p),\mathrm{sd}}$, because $\psi \in \mathcal{H}_{\vec{k}^0}$ has ordered spins and satisfies $ U_\tau \Psi_{\vec{k}^0} = \Psi_{\tau \cdot \vec{k}^0} = \Psi_{\vec{k}^0} $ for $\tau \in G_{\vec{k}^0}$. Second, we verify $\mathbb{U}\psi = \Psi$: Let $R$ be a set of representatives for the left cosets $G / G_{\vec{k}^0}$. Then any $\sigma \in G$ can be written uniquely as $\sigma = \rho \tau$ with $\rho \in R$ and $\tau \in G_{\vec{k}^0}$. 
\be
\mathbb{U}\psi = \sqrt{\frac{|G|}{|G_{\vec k^0}|}} \frac{1}{|G|}\sum_{\sigma \in G} U_\sigma \psi = \frac{1}{|G_{\vec{k}^0}|} \sum_{\sigma \in G} U_\sigma \Psi_{\vec{k}^0} =\frac{1}{|G_{\vec{k}^0}|} \sum_{\rho \in R} \sum_{\tau \in G_{\vec{k}^0}} U_\rho U_\tau \Psi_{\vec{k}^0}=\sum_{\rho \in R} U_\rho \Psi_{\vec{k}^0}=\sum_{\rho \in R} \Psi_{\rho \cdot \vec{k}^0}
\ee
As $\rho \cdot \vec{k}^0$ runs through all distinct spin configurations in the orbit $\mathcal{O}_{\vec{k}^0}$ exactly once, $ \mathbb{U}\psi = \sum_{\vec{k}' \in \mathcal{O}_{\vec{k}^0}} \Psi_{\vec{k}'} = \Psi $. Thus, $\mathbb{U}$ is surjective.


Preserving inner product: For any $\psi=\sum_{\vec k\in K^0}\psi_{\vec{k}}$ and $\phi=\sum_{\vec k\in K^0}\phi_{\vec{k}}$ in $\ch_{\G(\vec p),{\rm sd}}$
\be
\langle \mathbb{U} \psi,\mathbb{U}\phi\rangle=\sum_{\vec k\in K^0}\frac{|G|}{|G_{\vec k}|}\langle \mathbb{P}\psi_{\vec k} ,\mathbb{P}\phi_{\vec k} \rangle=\sum_{\vec k\in K^0}\frac{|G|}{|G_{\vec k}|}\langle \psi_{\vec k} ,\mathbb{P}\phi_{\vec k} \rangle.
\ee
For each term in the sum,
\begin{equation}
\langle \psi_{\vec k}, \mathbb{P} \phi_{\vec k} \rangle = \frac{1}{|G|} \sum_{\sigma \in G} \langle \psi_{\vec k}, U_\sigma \phi_{\vec k} \rangle.
\end{equation}
Since $\psi_{\vec k}$ and $\phi_{\vec k}$ belong to the ordered sector with spins $\vec{k}$, and $U_\sigma$ maps spins $\vec{k}$ to $\sigma \cdot \vec{k}$, the inner product $\langle \psi_{\vec k}, U_\sigma \phi_{\vec k} \rangle$ vanishes unless $\sigma \cdot \vec{k} = \vec{k}$ (i.e., $\sigma \in G_{\vec k}$).
Thus, the sum restricts to the stabilizer $G_{\vec k}$. Since $\phi_{\vec k}$ is a seed state (invariant under $G_{\vec k}$), $U_\sigma \phi_{\vec k} = \phi_{\vec k}$ for $\sigma \in G_{\vec k}$.
 \begin{equation}
    \frac{1}{|G|} \sum_{\sigma \in G_{\vec k}} \langle \psi_{\vec k}, \phi_{\vec k} \rangle = \frac{|G_{\vec k}|}{|G|} \langle \psi_{\vec k}, \phi_{\vec k} \rangle.
\end{equation}
As a result,
\be
\langle \mathbb{U} \psi,\mathbb{U}\phi\rangle=\sum_{\vec k\in K^0}\langle \psi_{\vec k} ,\phi_{\vec k} \rangle=\langle \psi,\phi\rangle.
\ee
Therefore $\mathbb{U}$ is unitary.
\end{proof}

We have established the identification $\ch_{\G(\vec p),{\rm sd}} \cong \wt{\ch}_{\G(\vec p)}$ via the isomorphism $\mathbb{U}$. This allows any spin-network state $\psi_{\G(\vec{p})}\in \ch_{\G(\vec p),{\rm sd}}$ to serve as a representative for the corresponding permutation-invariant state $\wt\Psi_{\G(\vec{p})} = \mathbb{U}\psi_{\G(\vec{p})} \in \wt{\ch}_{\G(\vec p)}$. By construction, $\psi_{\G(\vec{p})}$ is supported only on spin assignments $k_1\leq k_2\leq \cdots \leq k_{p_\fl}$ for each root link $\fl$, and can be represented a spin-network function of holonomies:
\be
\langle \vec{H}\mid \psi_{\G(\vec{p})}\rangle = \psi_{\G(\vec{p})}(\vec{H}),\qquad \vec{H} = \{H_\fl^{(i)}\}_{\fl\subset\G,\, i=1,\ldots,p_\fl},
\ee
where $H_\fl^{(i)}\in\mathrm{SU}(2)$ denotes the holonomy assigned to the $i$-th stacked link along $\fl$. Extending the isomorphism $\mathbb{U}$ to entire $\ch_{\rm Sym}$, we obtain the following representation of $\ch_{\rm Sym}$ in terms of ordered seed states:
\be
\ch_{\rm Sym} \cong  \bigoplus_{\G} \wt{\ch}_{\G,\rm{st}},\qquad \wt{\ch}_{\G,\rm{st}}\cong \bigoplus_{\vec{p}}{\ch}_{\G(\vec{p}),\mathrm{sd}}.\label{symRepSd}
\ee

\section{Spinfoam stacks}\label{Spinfoam stacks}

In what follows, we will develop a spinfoam formalism naturally associated with $\ch_{\rm Sym}$, as distinct from the spinfoam models of \cite{Han:2025emp,spinfoamstack} which naturally act on $\ch_{\rm Kin}$.

Given that spin-network stacks are well-defined in the Hilbert space and have interesting geometric interpretations, it is natural to expect that their covariant dynamics should be properly incorporated within spinfoam theory. In the spinfoam framework, a spinfoam is understood as a covariant history of a spin-network: starting from a three-dimensional spin-network, its links $\fl$ and nodes $\fn$ evolve into faces $f$ and edges $e$ in four spacetime dimensions. These faces and edges are respectively assigned spins $j_f = k_f/2$ and intertwiners $I_e$, mirroring the assignments on the initial spin-network. Conversely, slicing a spinfoam produces a spin-network state on the corresponding boundary.

In general, the faces and edges form a 2-complex, which underlies the definition of the spinfoam amplitude. Most of the studies on spinfoams focus on a fixed 2-complex, resulting in amplitudes and predictions that are sensitive to this choice. To develop a more complete and robust spinfoam formulation, however, one should seek amplitudes that are independent of the particular 2-complex used. This suggests summing over all possible 2-complexes, in line with the structure of LQG's Hilbert space, where generic states are superpositions over spin-networks on different graphs.

This motivation leads to extending the idea of stacks to the spacetime (covariant) setting. A spinfoam stack is defined as a sum over spinfoams corresponding to a family of 2-complexes, with the property that intersecting any spatial slice yields a spin-network stack. The amplitude for such a stack--a \emph{stack amplitude}--is the sum of the spinfoam amplitudes over all 2-complexes in this family. More concretely, just as a spin-network stack is assembled by stacking links onto a root graph $\G$, a spinfoam stack arises by stacking faces onto a root 2-complex $\ck$.

\begin{definition}
A labelled multi-2-complex is a tuple $(\mathcal{G}, \mathcal{F}, \Sigma_{\mathcal{F}}, \ell_{\mathcal{F}})$, where:

\begin{itemize}

\item  $\mathcal{G}=(V,E)$ is a graph that serves as the 1-skeleton of the complex: $V$ is the set of vertices (0-cells). $E$ is the set of edges (1-cells).

\item  $\mathcal{F}$ is a multiset of oriented boundary semi-cycles representing the faces (2-cells) $f$. The semi-circle becomes a full circle $\partial f$ if $f$ is an internal face. An oriented boundary cycle along $\partial f$ is a finite sequence of edges $(e_1, e_2, \dots, e_k)$ from $E$. Each $e_i$ is endowed an orientation by $f$ such that the target vertex of $e_i$ is the source vertex of $e_{i+1}$, forming an oriented walk in $\mathcal{G}$. For a given boundary semi-cycle, the multiplicity $p_f$ is the number of distinct faces in $\mathcal{F}$ that share this exact oriented boundary cycle $\partial f$. These faces are the stacked faces. The stacked faces sharing the same boundary circle have the same orientation.

\item  $\Sigma_{\mathcal{F}}$ is a finite set serving as the alphabet of face labels, satisfying $|\Sigma_{\mathcal{F}}| = |\mathcal{F}|$.

\item   $\ell_{\mathcal{F}}: \mathcal{F} \rightarrow \Sigma_{\mathcal{F}}$ is a bijective map that assigns a unique label to each face, with the convention that the faces are ordered by their labels. The stacked faces sharing the same boundary circle are arranged consecutively in the ordering.

\end{itemize}

A root 2-complex $\mathcal{K}$ is a labelled multi-2-complex where $\mathcal{F}$ is a set rather than a multiset. This implies that for every boundary cycle $\partial f$, the multiplicity is $p_f = 1$. No two faces share the exact same boundary cycle.

\end{definition}

Given a root 2-complex $\ck$, the family $F(\ck)$ is generated by increasing the face multiplicities arbitrarily (see FIG.\ref{sf_stack}). A labelled multi-2-complex in $F(\ck)$ with face multiplicities $\vec{p}=\{p_f\}_f$ is denoted by $\ck(\vec{p})$. It is obtained by replacing each root face $f\subset \ck$ with $p_f$ copies in the multiset $\mathcal{F}$ of $\ck(\vec{p})$, followed by a choice of ordering. 

In the following, we use $v, e,$ and $f$ to denote vertices, edges, and faces of the root complex $\ck$. We distinguish internal faces $h$ from boundary faces $b$. The internal faces $h$ do not connect to $\partial \ck$. In $\ck(\vec p)\in F(\ck)$, the stacked faces are labelled by the pair $(f,i)$, $f=h$ or $b$, where $i = 1,\cdots,p_{f}$ corresponds to the ordering to the stacked faces. The 2-complex $\ck(\vec p)\in F(\ck)$ has the same set of edges and vertices as the root complex $\ck$.

The stack amplitude $\sa_\ck$ sums the spinfoam amplitudes on all complexes $\ck(\vec p)\in F(\ck)$. To organize the sum, each root face $f$ is assigned a coupling constant $\lambda_f\in\C$. The amplitude on a given complex $\ck(\vec p)\in F(\ck)$ is weighted by $\prod_{f\subset\ck} \lambda_f^{p_f}$. 
\be
\sa_\ck=\sum_{\vec p}\prod_{f\subset \ck}\l_f^{p_f}\sa \lt(\ck,\vec p\rt).\label{StackAmplitude}
\ee
The amplitude $\sa \left(\ck, \vec{p}\right)$ generalizes the EPRL-KKL spinfoam amplitude to ensure compatibility with permutation invariance.

\begin{figure}[t]
	\centering
	\includegraphics[width=0.6\textwidth]{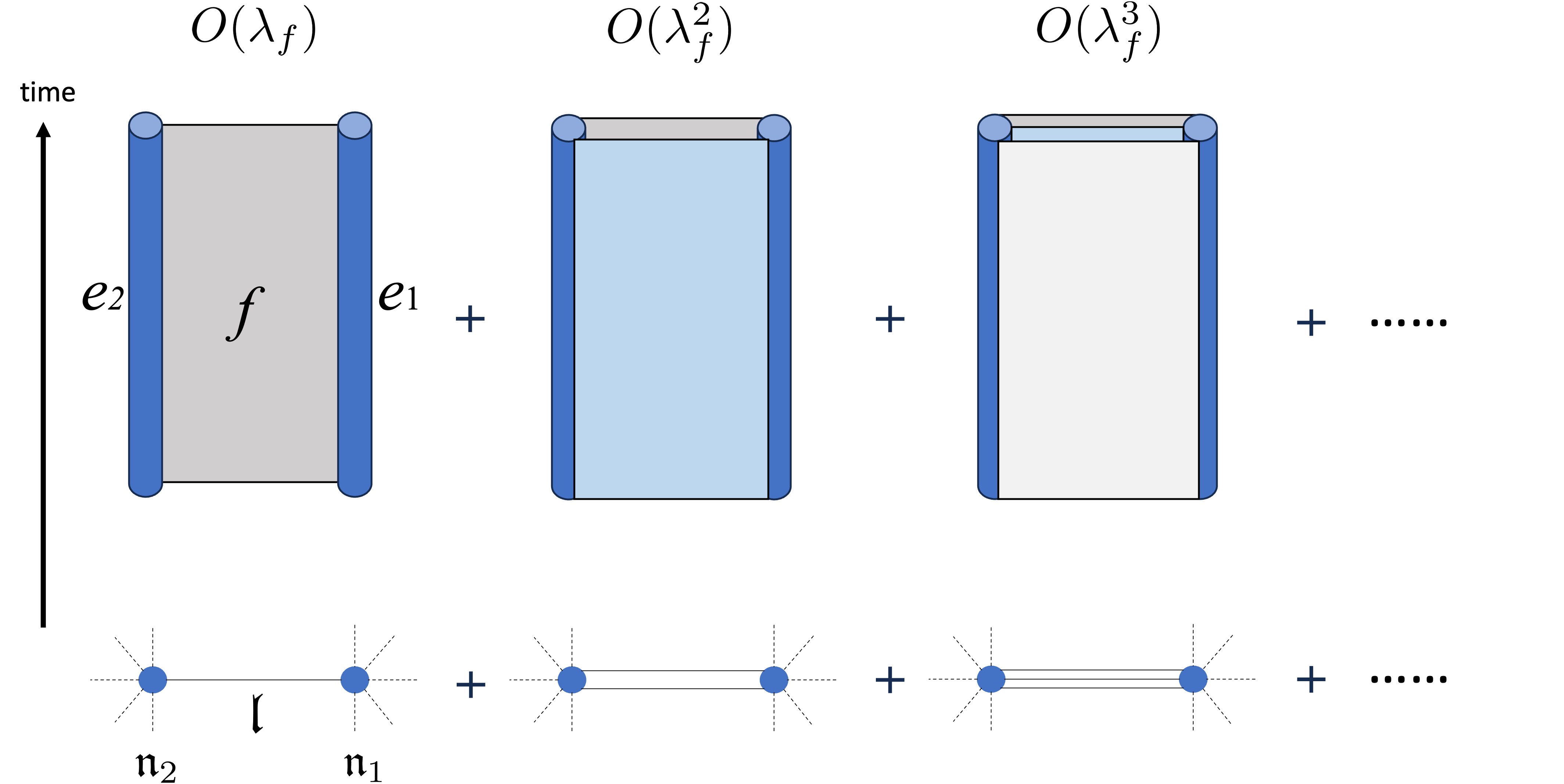}
	\caption{The spin-network stack evolves to the spinfoam stack: The spin-network link $\fl$ evolves to the spinfoam face $f$. the spin-network nodes $\fn_1,\fn_2$ evolve to the spinfoam edges $e_1,e_2$. The faces evolving from the dashed links are not shown on this figure. The leftmost complex is the root complex. The power of coupling constant $\l_f$ counts the number of stacked faces.}
	\label{sf_stack}
\end{figure}

In the EPRL-KKL formalism \cite{EPRL,KKL,Rovelli:2010vv,HanPI}, for any 2-complex such as $\ck(\vec p)$, each edge $e$ between a pair of vertices $v,v'$ is associated with a pair $g_{ve},g_{v'e}\in\Slc$, and each edge $e_b$ between a vertex $v$ and a boundary node $\fn$ is associated with $g_{ve_b}\in\Slc$. For any face, we introduced the short-hand notation $g_h=\{g_{ve}\}_{e\subset\partial h}$ and $g_b=\{g_{ve}\}_{e\subset\partial b}$ being the collection of group variables along the boundary of the face. Given any $f=h$ or $b$, all $p_f$ stacked faces share the same boundary edges and depend on the same set of group variables $g_{ve}\in\Slc$. The amplitude associated to a face $(f,i)\subset \ck(\vec p)$ colored by spin $k_i/2$ is 

\begin{itemize}
	\item Internal face $f=h$:
	\be
	\tau_{k_i}^{(h)}(g_h)= d_{k_i}\tr_{(k_i,\rho_i)}\lt[\overrightarrow{\prod_{v\in\partial h}}P_{k_i}g_{ve}^{-1}g_{ve'}P_{k_i}\rt],\qquad  \rho_i=\g (k_i+2),\label{zetakh}
	\ee
	where $d_k=k+1$ is the dimension of the spin-$k/2$ representation of SU(2). At any vertex $v$, the edges $e$ and $e'$ are, respectively, the incoming and outgoing edges according to the face's orientation. The trace $\tr_{(k_i,\rho_i)}$ is over the Hilbert space $\ch_{(k_i,\rho_i)}$ carrying the principal series $\Slc$ irrep, which can be decomposed into SU(2) irreps by $\ch_{(k_i,\rho_i)}=\oplus_{l=0}^\infty \ch_{k_i+2l}$. The orthogonal projection $P_{k_i}$ is from $\ch_{(k_i,\rho_i)}$ to the lowest SU(2) irrep subspace $\ch_{k_i}$ containing solutions of simplicity constraint. $\rho_i=\g (k_i+2)=2\g (j_i+1)$ follows from the convention in \cite{generalize} and is a convenient choice for the computation in Section \ref{Nondegenerate Hessian matrix}, although the result is not affected by a different choice such as $\rho_i=\g k_i$.
	
	\item Boundary face $f=b$:
	\be
	\tau_{k_i}^{(b)}(g_b)= d_{k_i}\tr_{(k_i,\rho_i)}\lt[ \lt(\overrightarrow{\prod_{v\in\partial b}}P_kg_{ve}^{-1}g_{ve'}P_k\rt)H_{\fl(b)}^{(i)}\rt]. \label{zetakb}
	\ee
	where the SU(2) holonomy $H_{\fl(b)}^{(i)}$ is along the boundary link $(\fl(b),i)$ of the face $(b,i)$. The orientation of the holonomy matches the orientation of the boundary link induced by the face.
	
\end{itemize}

\noindent	
The EPRL-KKL amplitude is defined by a product of $\tau_{k_i}^{(h)}(g_h)$ and $\tau_{k_i}^{(b)}(g_b)$ over all faces, followed by integrating over $g_{ve}\in\Slc$ for all pairs $(v,e)$ and summing over all spins coloring faces. In particular, at each root face $f$, we sum over $k_1, \cdots, k_{p_f}$ for $p_f$ stacked faces on $f$.

To ensure compatibility with $\ch_{\rm Sym}$, we require that for each face $f$, the $p_f$-tuple of spins $(k_1, \ldots, k_{p_f})$ assigned to the stacked faces satisfy the ordered condition $k_1 \leq \cdots \leq k_{p_f}$. This ordering mirrors the constraint imposed on ordered seed states in $\ch_{\G(\vec p),{\rm sd}}$. It turns out to be the only restriction needed for spinfoams. Consequently, the boundary spin data resulting from slicing a spinfoam stack correspond to states in $\ch_{\rm Sym}$ as given by the representation \eqref{symRepSd}. The summation over spins in the spinfoam amplitude $\sa \lt(\ck,\vec p\rt)$ runs over all assignments of spins to the stacked faces that obey the ordering condition for each root face $f$.

However, summing over an infinite range of spins leads to so-called bubble divergences in the spinfoam amplitudes; see, for example, \cite{Frisoni:2021vkv,Bonzom:2010zh,Bonzom:2013vha,Perez:2000fs}. In particular, these divergences generally occur in $\sa \lt(\ck,\vec p\rt)$ as stacked internal faces share boundaries, forming bubble-like closed surfaces where the corresponding spins become unbounded. To regularize the divergence, we introduce a cut-off scheme as follows. 

Let $\{E_k\}_{k\in\mathbb{Z}_+}$ be a sequence of positive integers, with $E_k$ defined to be monotonically increasing and to grow asymptotically as $k^N$ for some positive integer $N$ as $k\to\infty$. For each root face $f$, we define a \emph{cut-off function}
\be
\a_{p_f,\vec{k}} = \sum_{i=1}^{p_f} E_{k_i},\qquad \vec{k}=(k_1, \ldots, k_{p_f})
\ee
and impose that $\a_{p_f,\vec{k}}\leq A_f$ for some finite cut-off $A_f\gg1$. This requirement ensures that both the face multiplicities $p_f$ and the spins are bounded. In practice, we frequently choose $E_k = k(k+2)$, proportional to the quadratic Casimir, for explicit computations.

Using the formalism developed above, the stack amplitude $\sa_\ck$ can be written as a function of the boundary SU(2) holonomies $\vec{H} = \{H_{\fl(b)}^{(i)}\}_{b,i}$:
\be
\sa_{{\cal K}}\left(\vec{A},\vec{H},\vec{\l}\right)&=&\int\rmd\O(g)\prod_h\o_h^{\rm bos}\left(A_h;g_h,\l_h\right)\prod_{b}\o_b^{\rm bos}\left(A_b;g_b,\vec{H}_{\fl(b)},\l_b\right)\label{saKampli}\\
\o_h^{\rm bos}&=&\sum_{p_h=1}^{\infty}\l_h^{p_h}\sum_{1\leq k_1\leq \cdots\leq k_{p_h}}^\infty\prod_{i=1}^{p_h}\tau^{(h)}_{k_{i}}\left(g_h\right)\Theta\left(A_{h}-\alpha_{p_h,\vec{k}}\right),\label{omegaf20}\\
\o_b^{\rm bos}&=&\sum_{p_b=1}^{\infty}\l_b^{p_b}\sum_{1\leq k_1\leq \cdots\leq k_{p_b}}^\infty\prod_{i=1}^{p_b}\tau^{(b)}_{k_{i}}\left(g_b,{H}^{(i)}_{\fl(b)}\right)\Theta\left(A_{b}-\alpha_{p_b,\vec{k}}\right).\label{omegab20}
\ee
The function $\Theta(A_f-\alpha_{p_f,\vec k})$ imposes the cut-off on the summations over $p_f$ and $\vec k$, where $\Theta(x)$ is defined as $\Theta(x)=1$ for $x>0$, $\Theta(x)=0$ for $x<0$, and $\Theta(x)=1/2$ for $x=0$. 

The sum over 2-complexes $\ck(\vec{p})\in F(\ck)$ in $\sa_\ck$ is encoded in Eqs.\eqref{omegaf20} and \eqref{omegab20} through summations over $p_f$. The spinfoam amplitude $\sa(\ck, \vec{p})$ is obtained by extracting the coefficient of $\prod_f \lambda_f^{p_f}$ in the power series expansion of $\sa_\ck$.

We make several remarks about the stack amplitude $\sa_\ck$:

\begin{itemize}

\item If the coupling constants $\l_f,\l_b$ satisfy some consistent relations between different complexes, the stack amplitude is consistent under cut and gluing by the same computation as in \cite{face,Han:2025emp}: If we slice $\ck$ into two root complexes $\ck_1$ and $\ck_2$, the consistency of stack amplitudes under cut and gluing is 
\be
\int \rmd \vec H \sa_{\ck_1}(\vec H)\sa_{\ck_2}(\vec H)=\sa_{\ck},\qquad \ck=\ck_1\cup \ck_2.\label{cylconsist}
\ee
The SU(2) holonomies $\vec H$ are along the stacked links where $\ck_1,\ck_2$ are glued back to $\ck$. We denote by $\rmd \vec H $ the product Haar measure. For any faces $f\subset\ck$ being sliced into faces $b_1\subset\ck_1$ and $b_2\subset\ck_2$, the consistency relation between coupling constants $\l_f$ on $\ck$ and $\l_{b_1},\l_{b_2}$ on $\ck_1,\ck_2$ is $\l_f=\l_{b_1}\l_{b_2}$. 
When this consistency relation for the coupling constants is satisfied, the stack amplitudes obey the consistency condition \eqref{cylconsist}. While these relations help clarify the interpretation of spinfoams as the histories of spin-network stacks through slicing, they are not required for the fixed point analysis in this work.

\item For any root face $f=h$ or $b$, consider a group of stacked faces that share the same spin value. The functions $\tau_{k_i}^{(b)}$ depends on $i$ only via $k_i$ and $H^{(i)}$, exchanging $H^{(i)}$ and $H^{(j)}$ with $k_i=k_j$ leaves the product $\prod_{i=1}^{p_b}\tau_{k_i}^{(b)}$ invariant. More generally, for any assignment of spins $\vec{k}$ to the stacked boundary faces, there exists a stabilizer subgroup $G_{\vec{k}}$ consisting of all permutations that preserve the spin profile $\vec{k}$. The partial amplitudes with fixed boundary $\vec{k}$ in $\sa_{\mathcal{K}}$ are invariant under the action of $G_{\vec{k}}$. This property reflects the second condition $U_\sigma \psi = \psi$, $\sig\in G_{\vec{k}}$, in the definition of ordered seed states. As a result, the boundary data of $\sa_{\mathcal{K}}$, or the boundary data obtained by slicing $\sa_{\mathcal{K}}$, match with the spin-network stack states in $\ch_{\rm Sym}$ specified by the representation \eqref{symRepSd}.

\item The stack amplitude $\sa_\ck$ constructed above is in the holonomy representation, where the boundary state is the generalized eigenstate of holonomy operators. For a generic boundary state $\psi\in\ch_{\rm Sym}$, the associated stack amplitude is given by
\be
\sa_\ck[\psi]:=\int \rmd \vec{H} \sa_\ck (\vec{H})^* \psi (\vec{H}),
\ee
where $\rmd \vec{H}$ is the product SU(2) Haar measure, and $\psi (\vec{H})$ is a linear combination of ordered seed states, according to the representation \eqref{symRepSd}.

\item The integrand $\prod_h\o_h^{\rm bos}\prod_b\o_b^{\rm bos}$ in $\sa_\ck$ is invariant under a set of continuous gauge transformations:
\be
g_{ve}\to x_v g_{ve} u_e,\qquad x_v\in\Slc,\quad u_e\in\Su.
\ee
At each vertex, the $\Slc$ gauge freedom leads to a divergence. We address this by fixing the gauge: choose a specific edge $e_0(v)$ at each vertex $v$, $e_0(v)\neq e_0(v')$ for $v\neq v'$ and fix $g_{v,e_0(v)}=\mathbb{I}$. The integration measure $\rmd\O(g)$ reads
\be
\rmd\O(g)=\prod_{(v,e)}\rmd g_{ve}\prod_v\delta\lt(g_{v,e_0(v)}\rt)
\ee
where $\rmd g_{ve}$ denotes the Haar measure on $\Slc$. Under the gauge fixing, the absolute convergence of $\sa_\ck$ rely on a mild restriction on the root complexes under consideration, namely, when a sphere is used to slice the neighborhood containing a single vertex $v$, the resulting intersection graph $\G_v$ on the sphere is required to be 3-connected \cite{Kaminski:2010qb,finite}.

For the analysis in Section \ref{Nondegenerate Hessian matrix}, we impose the condition that $e_0(v)$ does not intersect the boundary $\partial \ck$, when $\ck$ has at least one internal face, i.e. we require $e_0(v)\in E_{\rm int}$ (defined below) for all vertex $v$. This condition does not impose any restriction on the topology of root complex (see Appendix \ref{vertexEdgeInjection} for an explanation).


\item For our subsequent discussion, it is convenient to define the subcomplex $\ck_-$ (the interior of $\ck$) as the 2-complex formed exclusively by internal faces $h$, with associated internal edges $e\in E_{\rm int}$ and their endpoint vertices, where the set $E_{\rm int}$ collects all edges on the boundary of internal faces. The subcomplex $\ck_-$ is constructed from $\ck$ by removing all boundary faces $b$ and any edges $e_b$ attached to the boundary. In this work, we restrict attention to root complexes $\ck$ such that, whenever a small sphere is used to cut out a neighborhood of any vertex $v\in \ck_-$ (the vertex connecting to a boundary face), the intersection subgraph $\G_{v,-}\subseteq \G_v$ formed between the sphere and the subcomplex $\ck_-$ is always connected. This requirement will be important for the technical arguments in Section \ref{Localization of stack amplitude} and subsequent sections. Note that this requirement is nontrivial only for $v\in\partial \ck_-$, since for $v\not\in\partial \ck_-$, $\G_{v,-}=\G_v$ has been assumed to be 3-connected. Any 2-complex dual to a simplicial complex satisfies this requirement. 

\item In this work, we only consider the root complexes $\ck$ having simply connected interior, i.e. $\pi_1(\ck_-)$ is trivial, for the technical reason to be discussed in Section \ref{Some explicit computations}. Physically, a trivial fundamental group is valid for all local patches of spacetime. 


\end{itemize}

The complete spinfoam amplitude $\sa$ is obtained by summing these stack amplitudes $\sa_\ck$ over root complexes with compatible boundaries, each possibly further weighted by a complex coefficient $c_\ck$.
\be
\sa=\sum_{\ck}c_\ck\sa_\ck,\label{sa6}
\ee
The root complexes included in the sum satisfy the following conditions: (1) they are simply connected, (2) for every vertex $v\in\ck$, the graph $\G_v$ is 3-connected, (3) for every vertex $v$ lying on $\partial\ck_-$, the graph $\G_{v,-}$ is connected, and (4) all have identical boundaries, i.e., $\partial\ck = \G$ for every $\ck$.

In the sum over complexes \eqref{sa6}, the infinitely many undetermined coefficients $c_\ck$ encode the triangulation dependence of spinfoam theory. In what follows, our main result is to show that, at the fix point that we discover in this paper, these infinitely many ambiguities are effectively reduced to only a finite number of degrees of freedom.

We note that, in the most general case, the weights assigned to each $\sa(\ck,\vec p)$ in \eqref{StackAmplitude} could be chosen arbitrarily, rather than being determined specifically by the coupling constants as in our construction. However, we interpret the specific assignment of weights via the coupling constants in \eqref{StackAmplitude} as intrinsic to the definition of the fixed point; any deviation from this prescription would drive the theory away from the fixed point. Fixing these weight by coupling constants may also be viewed as a consistent truncation of the infinite-dimensional parameter space: the theory is defined within the truncated parameter space. The parameter space after the truncation remains infinite-dimensional.





\section{State sum}\label{State sum and bose gas}

We assume the root complex $\ck$ to have one or more internal faces. In order to compute $\o_h^{\rm bos }$ of an internal root face $h$, we apply the following inverse Laplace transform formula known as Perron's formula \cite{Hardy1915General,BarberoG:2008dwr,lqgee1}: Given two sequence $\{\a_n\}_{n=1}^\infty $ and $\{\b_n\}_{n=1} ^\infty $ where $\b_n\in\C$ and $\a_n>0$, if $\sum_{n=1}^\infty\left|\beta_{n}\right|e^{-\alpha_{n} c}<\infty$ for some $c>0$, we have
\be
\sum_{n=1}^{\infty}\beta_{n}\Theta\left(A-\alpha_{n}\right) =\mathscr{P}\int_{T-i\infty}^{T+i\infty}\frac{\rmd s}{2\pi i s}e^{As}\Xi(s),\qquad \Xi(s)=\sum_{n=1}^{\infty}\beta_{n}e^{-\alpha_{n}s},\qquad T>c\label{invLaplace}
\ee
In particular, if $\Xi(s)$ is meromorphic for $\re(s)>0$, the parameter $T>0$ is greater than the real part of all singularities given by the integrand. $\mathscr{P}\int $ denotes the principal value of the integral.

Apply this formula to the state-sum in $\o_h^{\rm bos}$: The sum over $n$ in \eqref{invLaplace} corresponds to the sum over $p_h$ and $k_1,\cdots,k_{p_h}$, and $\a_n$ corresponds to $\a_{p_h,\vec{k}}$. The summand $\b_n$ corresponds to $\l^{p_h}_h\prod_{i=1}^{p_h} \t_{k_i}^{(h)}(g_h)$ with a fixed $g_h$. The permutation symmetric tuple $(k_1,\cdots,k_{p_h})$ indicates that the stacked faces are indistinguishable and obey bosonic statistics. The sum over symmetric tuples can be recast as a sum over non-negative occupation numbers $n_k$ for each spin $k$, such that $p_h=\sum_{k\in\Z_+} n_k$:
\be
\Xi(s_h) =\sum_{p_h=1}^{\infty}\l_h^{p_h}\sum_{1\leq k_1\leq \cdots\leq k_{p_h}}^\infty\lt[\prod_{i=1}^{p_h}\tau^{(h)}_{k_{i}}\left(g_h\right)\rt]\,e^{-s_h \a_{p_h,\vec{k}}} = {\sum_{\{n_k\}_{k\in\Z_+}}}
\prod_{k=1}^\infty
\left[ \lambda_h \tau_k^{(h)}(g_h) e^{-s_h E_k} \right]^{n_k}-1,
\ee
where the subtraction of $1$ accounts for the vacuum contribution with $p_h=0$, corresponding to all $n_k=0$. The series is absolutely convergent for sufficiently large $\re(s_h)>0$. The proof is given in Appendix \ref{meromorphic function}.

The state-sum $\o^{\rm bos}_h$ is written as
\be
\omega_h^{\mathrm{bos}} =\mathscr{P}
 \int_{T-i\infty}^{T+i\infty}
\frac{\rmd s_h}{2\pi i s_h}\, e^{A_h s_h}\Xi[s_h].
\ee
The function $\Xi[s_h]$ is absolutely convergent at every point on the integration contour, then we can carry out the sum over $\{n_k\}$ in $\Xi[s_h]$ (see Appendix \ref{meromorphic function} for details)
\be
\Xi(s_h) = \prod_{k=1}^\infty \frac{1}{1-\lambda_h \tau_k^{(h)}(g_h) e^{-s_hE_k}} - 1,\label{Xishprod}
\ee
The integration contour of $\omega_h^{\mathrm{bos}}$ lies to the right of all singularities of \eqref{Xishprod}. To compute the integral of $\o_h^{\rm bos}$, we analytically continue $\Xi(s_h)$ in \eqref{Xishprod} beyond the contour and obtain $\Xi(s_h)$ as a meromorphic function for $\re(s_h)>0$ (see Appendix \ref{meromorphic function}).

\section{2-dimensional gauge theory}\label{2-dimensional gauge theory}

The Laplace transform $\Xi(s_h)$ of $\omega^{\rm bos}_h$ is an analog of grand canonical partition function:
\be
\Xi(s_h)=\wt{\sum_{\{n_k\}_{k\in\Z_+}}}
\prod_{k=1}^\infty
 Z_{k}(g_h,s_h) ^{n_k},\qquad Z_{k}(g_h,s_h)=\lambda_h \tau_k^{(h)}(g_h) e^{-s_h E_k}
\ee
The quantity $Z_{k}(g_h,s_h)$ is a partition function on a single face with spin $k/2$. In the following, we show that $Z_{k}(g_h,s_h)$ is a partition function of 2-dimensional Yang-Mills (YM) theory with the boundary condition that relates to $k$, $g_h$ and the projection $P_k$ of simplicity constraint.

The classical action for 2-dimensional Yang-Mills (YM) theory with a compact gauge group $G$ is given by
\be
S_{\text{YM}_2} = \frac{1}{4e^2} \int_{\Sigma} \rmd^2 x\, \sqrt{|\det \mathrm{g}|}\, \mathrm{Tr}(\cf_{\alpha\beta} \cf^{\alpha\beta}),
\ee
where $\cf$ is the curvature of a gauge connection $\ca$, $e$ is the coupling constant, and $\mathrm{g}_{\alpha\beta}$ is an arbitrary metric on the 2d surface $\Sigma$. In two dimensions, any 2-form is proportional to the area element so we may write $\cf_{\alpha\beta} = f\, \eps_{\alpha\beta}$, where $\eps_{\alpha\beta}$ is the area element normalized by $\epsilon_{\alpha\beta} \epsilon^{\alpha\beta} = -2$ and $f$ is an adjoint-valued scalar. Plugging this into the action, we obtain
\[
S_{\text{YM}_2} = -\frac{1}{2e^2} \int_{\Sigma} \rmd^2 x\, \sqrt{|\det G|}\, \mathrm{Tr}(f^2).
\]
Thus, the action depends on the metric only via the area element $\eps=\rmd^2 x\, \sqrt{\det \mathrm{g}}$. Since $f$ is a scalar, the action is invariant under all area-preserving diffeomorphisms, and the only invariants of the geometry on which the partition function can depend are the genus of $\Sigma$ and its total area.

We quantize the theory on a cylinder $\Sigma = S^1 \times [0,t]$, $t>0$. This is equivalent to the time evolution of the state on a circle. We take the temporal gauge $\ca_0=0$. At a constant time slice $x^0=0$, the gauge invariant data is the holonomy along the circle $U=\calp \exp\oint_{S^1}\ca_1\rmd x^1$. The Hilbert space $\ch_{\rm YM_2}$ is spanned by $L^2$ (w.r.t. Haar measure) functions $\psi(U)$ satisfying the gauge invariance $\psi(U)=\psi(gUg^{-1})$, for all $g,U\in G$. The Hamiltonian operator is given by $H\propto\oint_S \rmd x^1 E^2_1\propto L\Delta_U$ where $L$ is the circumference of the circle and $\Delta_U$ is the Laplace operator on $G$, whose eigenvalues are quadratic Casimirs $C_2(R)$ of irreps $R$. Given two states $\psi_1,\psi_2\in\ch_{\rm YM_2}$ that define the boundary condition, the partition function of 2d YM theory on the cylinder equals to the transition amplitude
\be
Z_{\rm YM_2}(S^1 \times [0,t])=\langle \psi_2\mid e^{-{s} \Delta_U} \mid\psi_1\rangle,\qquad s=\frac{i}{2} {tL}{e^2} ,
\ee
here $s$ is proportional to the total area $tL$ of the cylinder.

The partition function on a disk $D_2$ can be obtained from the cylinder partition function by collapsing one boundary circle to a single point. This imposes the boundary condition $U=1$ on that end, corresponding to the choice $\psi_1(U) = \delta(U)$. By $\delta(U)=\sum_R \dim(R)\chi_R(U)$ where $R$ labels the irrep of $G$,
\be
Z_{\rm YM_2}(D_2)=\langle \psi_2\mid e^{-{s} \Delta_U} \mid 1\rangle=\sum_R\langle \psi_2\mid\chi_R\rangle e^{-{s} C_2(R)} ,
\ee
where $C_2(R)$ is the quadratic Casimir of $R$.

To relate to our context, we consider $G=SU(2)$ and thus $C_2(R)=C_2(k)=k(k+2)/4$, and we let $D_2$ be the internal face $h$. The circle holonomy $U$ of 2d $\mathrm{YM}$ theory is along $\partial h$, we split $U$ into segment holonomies $H_{ee'}$ connecting the middle points of two neighboring edges $e,e'$, following the orientation of $\partial h$:
\be
U=\overrightarrow{\prod_{(e,e^{\prime})\subset \partial h}}H_{ee^{\prime}}
\ee
We define the following boundary state as a function of $H_{ee^{\prime}}$
\begin{equation}
	\Psi_{(k,\rho)}\left[g_{h}\right]=\l_h\bigotimes_{v\in\partial h}\psi_{(k,\rho)}\left[g_{ve}^{-1}g_{ve'}\right],\qquad \psi_{(k,\rho)}\left[g_{ve}^{-1}g_{ve'}\right]\left(H_{e^{\prime}e}\right)=d_{k}\mathrm{Tr}_{(k,\rho)}\left[P_{k}g_{ve}^{-1}g_{ve^{\prime}}P_{k}H_{e^{\prime}e}\right]^{*},
\end{equation}
then we have the following relation when we set $E_k=k(k+2)$:
\be
Z_{k}(g_h,s_h)&=&\lambda_{h}\tau_{k}^{(h)}(g_{h})e^{-s_{h}E_{k}}=\langle\Psi_{(k,\rho)}\left[g_{h}\right]\mid e^{-4s_{h}\Delta_{U_{h}}}\mid1\rangle\equiv Z_{\rm YM_2}\lt(D_2,\Psi_{(k,\rho)}\left[g_{h}\right]\rt)\\
&=&\lambda_{h}\int\prod_{(e,e^{\prime})}dH_{ee^{\prime}}e^{-4s_{h}\Delta_{U_{h}}}\delta\left(U_{h}\right)\prod_{v\in\partial h}\left(d_{k}\mathrm{Tr}_{(k,\rho)}\left[P_{k}g_{ve}^{-1}g_{ve^{\prime}}P_{k}H_{e^{\prime}e}\right]\right)
\ee
This is the partition function of a 2d YM theory on the face $h$ with a boundary condition defined by $\Psi_{(k,\rho)}[g_h]$. The partition function is analytically continued, so it has a complex area parameter $s_h\in\C$. The state $\Psi_{(k,\rho)}[g_h]$ couples the 2d gauge field $\ca$ to the external $\Slc$ gauge field $g_{ve}^{-1}g_{ve'}$. The state $\Psi_{(k,\rho)}[g_h]$ is not invariant under gauge transformation of $H_{ee'}$, and the partition function implicitly involves a projection of this state onto the gauge-invariant subspace.

\begin{table}[h!]
    \centering
    \begin{tabular}{|l|l|l|}
    \hline
    \textbf{Statistical ensemble} &\textbf{Stacked spinfoam faces}   & \textbf{Ensemble of YM world-sheets} \\
    \hline
    A gas of bosonic particles  & A stack of spinfoam faces over $h$      & A gas of world-sheets carrying YM theories. \\
	\hline
    A bosonic particle                           & A single stacked face  & A single world-sheet\\
	\hline
    Quantum state $k$              & The spin $k/2$ carried by the face & State $\Psi_{(k,\rho)}$ at the world-sheet boundary\\
	\hline
	Particle energy level $\eps_k$ & $E_k$ for defining cut-off & YM energy (quadratic Casimir) $C_2(k)$\\
	\hline
    Occupation number        & The number $n_k$ of faces with spin $k/2$ & The number of world-sheets having boundary $\Psi_{(k,\rho)}$ \\
	\hline
    Inverse temperature $\beta$   & $s_h$   & The complexified area parameter $s\propto tLe^2$\\
	\hline
	Single particle partition function $z_k e^{-\b \eps_k}$ & $Z_k(g_h,s_h)=\l_h\t_k^{(h)}(g_h) e^{-sE_k}$ & YM partition function $Z_{\rm YM_2}\lt(D_2,\Psi_{(k,\rho)}\left[g_{h}\right]\rt)$\\
	\hline
	Grand canonical partition function & $\Xi(s)=\wt{\sum}_{\{n_k\}_{k\in\Z_+}}
	\prod_{k=1}^\infty Z_{k}(g_h,s_h) ^{n_k}$ & $\wt{\sum}_{\{n_k\}_{k\in\Z_+}}
	\prod_{k=1}^\infty Z_{\rm YM_2}\lt(D_2,\Psi_{(k,\rho)}\left[g_{h}\right]\rt)^{n_k}$\\
	 \hline
    \end{tabular}
    \caption{}
	\label{table_ensemble}
\end{table}

Table \ref{table_ensemble} establishes a dictionary between the statistical ensemble of boson gas, stacked spinfoam faces, and ensemble of world-sheets carrying 2d YM theories. The comparison suggests that we can interpret $\Xi(s_h)$ as a grand canonical partition function for a gas of faces over $h$ carrying YM theories (with complexified area). Moreover, the stacked face amplitude $\o_h^{\rm bos}$, as the inverse Laplace transform of $\Xi(s_h)$, is the partition function of the corresponding micro-canonical ensemble.

In this statistical ensemble, the world-sheets are indistinguishable and satisfy bosonic statistics. As we see below, for $\lambda_h$ is not too small, the system can accumulate all world-sheets in the lowest energy state (lowest spin $k=1$), analogous to Bose-Einstein condensation.

\section{Condensation of quantum geometry}\label{Bose-Einstein condensate of quantum geometry}

The poles of $\Xi(s_h)$ are denoted by $s_h(k,m,g_h,\l_h)$ with $k\in\Z_+,\  m\in\Z,\ \l_h\in\C,\ g_h=\{g_{ve}\}_{e\subset\partial h}$:
\be
s_h(k,m,g_h,\l_h)=r_h(k,g_h,\l_h)+im\, \nu_h(k),\qquad r_h(k,g_h,\l_h)=\frac{\ln\lt[\l_h\t_k^{(h)}(g_h)\rt]}{E_k},\qquad \nu_h(k)=\frac{2\pi}{E_k}.
\ee
Here $r_h(k,g_h,\l_h)$ is complex in general. $\ln[\cdots]$ takes the principal value of the logarithm. The real part of $s_h(k,m,g_h,\l_h)$ does not depend on $m$:
\be
R_h(k,g_h,\l_h):=\re[s_h(k,m,g_h,\l_h)]=\re[r_h(k,g_h,\l_h)]=\frac{\ln\lt[\lt|\l_h\t_k^{(h)}(g_h)\rt|\rt]}{E_k}.
\ee 
Given $g_h$ and $\l_h$ such that $\Xi(s_h)$ has poles on the right-half plane, $\o_h^{\rm bos}$ can be computed by using the residue theorem, and the dominant contribution comes from the poles $s_*$ with largest real parts
\be
\o_h^{\rm bos}&=&\lt[\sum_{s_*}\underset{s\to s_*}{\mathrm{Res}}\frac{e^{A_h s}}{s}\Xi\lt(s\rt)\rt]\lt[1+O(A_h^{-\infty})\rt],\qquad \re\lt[s_*\rt]=\sup_{k}\frac{\ln\lt[\lt|\l_h\t_k^{(h)}(g_h)\rt|\rt]}{E_k},\label{sumsstar}
\ee
where $\mathrm{Res}_{s\to s_*}$ denotes the residue at the pole $s_*$. The error $O(A^{-\infty})$ term collects the exponentially suppressed contribution from non-dominant poles. The proof of this formula is given in Appendix \ref{Inverse Laplace transform}.

We denote by $\mathscr{G}_h$ the space of $g_h=\{g_{ve}\}_{v\in\partial h}$, with the gauge fixing $g_{v,e_0(v)}=1$ if $e_0(v)\subset \partial h$. The function $\t_k^{(h)}(g_h)$ satisfies the uniform bound $|\t_k^{(h)}(g_h)|\leq d_k^{2}$ on $\mathscr{G}_h$ for each $k$, and the bound is saturated if and only if $g_h$ satisfies (see Appendix \ref{Bound of tau k} or \cite{spinfoamstack})
\be
g_{ve}^{-1}g_{ve'}\in \Su,\quad \forall e,e'\subset\partial h,\quad e\cap e'=v, \qquad \overrightarrow{\prod_{v\in\partial h}} g_{ve}^{-1}g_{ve'} =\pm \mathbb{I}\ .\label{g0solution} 
\ee  
We denote the subspace of $g_h$ satisfying this condition by $\cc_h^\pm$, where $\cc_h^+$ corresponds to the product being $+\mathbb{I}$ and $\cc_h^+$ to $-\mathbb{I}$. We define $\cc_{h}=\cc_h^+\cup\cc_h^-$, where $\cc_h^+$ and $\cc_h^-$ are disjoint. For any $k\in\Z_+$ and $\l_h$, $R_h(k,g_h,\l_h)$ reaches the maximum at $g_h\in \cc_{h}$:
\be
\b_k(\l_h) \equiv E_k^{-1}\ln\lt(|\l_h|d_k^2\rt)=\sup_{g_h}R_h(k,g_h,\l_h).
\ee
Fixing the value of $\l_h$, since $\b_k(\l_h)\to0$ as $k\to\infty$, for any infinitesimal $\eps>0$, there exists $k_m\in\Z$ such that there are only finitely many $k<k_m$ such that $\b_k(\l_h)>\eps$. The maximum of $\b_k(\l_h)$ among $k<k_m$ gives the supremum for the real part of poles from $\Xi(s_h)$:
\be
\sup_{k\in\Z_+}\b_k(\l_h)=\sup_{k<k_m}\b_k(\l_h)=\sup_{k,g_h}R_h(k,g_h,\l_h).
\ee
This is the global maximum of the poles' real parts. As we see in the following analysis. This global maximum $\sup_{k}\b_k(\l_h)$ relates to the leading asymptotic behavior of $\o_h^{\rm bos}$ as $A\to\infty$.

Focus on $g_h$ satisfying \eqref{g0solution} and compare $\b_k(\l_h)$ among $k$'s. We define the \emph{condensation spin} $k_0/2$ to be the location of the global maximum, i.e. the condensation spin is $k_0/2$ such that
\be
\b_{k_0}(\l_h)=\sup_k\b_k(\l_h)=\sup_k\lt[E_k^{-1}\ln\lt(|\l_h|d_k^2\rt)\rt].
\ee
The value of $k_0$ clearly relates to the value of $|\l_h|$. Qualitatively, as $|\lambda_h|$ decreases, the value $k_0$ that maximizes $\beta_k(\lambda_h)$ increases. This is due to the fact that simultaneously small values of $k_0$ and $\lambda_h$ would result in $\beta_k$ being negative, which is not admissible for the real part of $s_*$. As $\l_h\to 0$, the real part of $s_*$ approaches to zero, and correspondingly, only the root face survives in $\sa_\ck$.

For $E_k=k(k+2)$ and $|\l_h|$ not too small, $|\l_h|\geq \frac{3 }{8}\sqrt[5]{\frac{3}{2}}\simeq 0.40668$, the maximum of $\b_k$ is at the lowest nonzero spin $k=1$. For a generic value of $|\l_h|$, the maximum of $\b_k$ picks up a single $k=k_0$. However, for some special values of $|\l_h|$, the maximum can locate at two different neighboring $k$'s, i.e. in this case $k_0$ is not unique. See FIG.\ref{plotEx1} for illustrations.

\begin{figure}[h]
	\centering
	\includegraphics[width=1\textwidth]{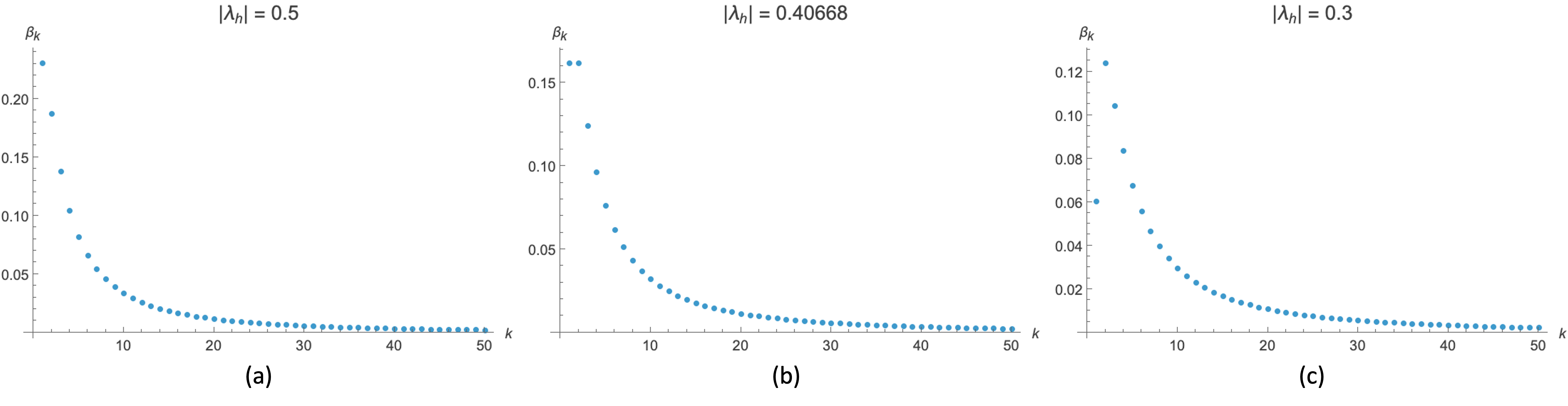}
	\caption{$E_k=k(k+2)$: (a) For $|\l_h|=0.5$, the maximum of $\b_k$ is at $k_0=1$. (b) For $|\l_h|=\frac{3 }{8}\sqrt[5]{\frac{3}{2}}\simeq 0.40668$, the maximum of $\b_k$ is attained at both $k_0=1$ and $k_0=2$. (c) For $|\l_h|=0.3$, the maximum of $\b_k$ is at $k_0=2$.}
	\label{plotEx1}
\end{figure}





From now on we fix a generic $\l_h$ such that $k_0$ is unique. At any $\mathring g_h\in\cc_h$, it is clear that $s_*$ in \eqref{sumsstar} only corresponds to $k=k_0$, i.e. $s_* = s_h(k_0,m,\mathring g_h,\l_h)$ is the pole having the largest real part that equals $\b_{k_0}(\l_h)$.  The value $\b_{k_0}(\l_h)$ stands out as an isolated maximum among the sequence $\{\b_{k}(\l_h)\}_k$. Define $\b^* = \sup_{k \neq k_0} \b_k(\l_h) > 0$, which satisfies $\b^* < \b_{k_0}(\l_h)$. The function $R_h(k_0, g_h, \l_h)$ depends continuously on $g_h$ within the domain $\cv\subset\mathscr{G}_h$ where $\t_{k_0}^{(h)}(g_h) \neq 0$. When  generalizing away from $\cc_h$, we introduce the following open neighborhood of $\cc_h$
\be
\mathscr{U}_h = \big\{ g_h \in \cv \mid R_h(k_0, g_h, \l_h) > \b^* \big\}.\label{scrUhDef}
\ee
Clearly, $\mathscr{U}_h$ is open and contains $\cc_h$.

\begin{lemma}
$s_* = s_h(k_0,m, g_h,\l_h)$ still holds on the open neighborhood $\mathscr{U}_h$, i.e. $R_h(k_0,g_h,\l_h)$ is still the largest among all poles for all $g_h\in\mathscr{U}_h$. 
\end{lemma}

\begin{proof}
Since $\b^* = \sup_{k \neq k_0} \b_k(\l_h) = \sup_{k \neq k_0,\, g_h} R_h(k, g_h, \l_h)$, it follows that $R_h(k_0, g_h, \l_h) > \b^* \geq R_h(k \neq k_0, g_h, \l_h)$, for all $ g_h \in \mathscr{U}_h$.

\end{proof}

As a consequence, for any $g_h\in\mathscr{U}_h$, the poles $s_* = s_h(k_0,m,\mathring g_h,\l_h)$ are simple for all $m\in\Z$. By computing the residues at the poles,
\be
\o_h^{\rm bos}(A_h;g_h,\l_h)&=&{e^{A_h r_h(k_0,g_h,\l_h)}}\Fs_h\lt(A_h,g_h,\l_h\rt)\lt[1+O(A_h^{-\infty})\rt],\label{leadingpolecontri}\\
\Fs_h\lt(A_h,g_h,\l_h\rt)&=&\frac{1}{E_{k_0}}\sum_{m\in\Z}\frac{e^{iA_h m \nu_h(k_0)}}{r_h\lt(k_0,g_h,\l_h\rt)+im\, \nu_h\lt(k_0\rt)}  Q_h\lt(k_0,m,g_h,\l_h\rt)\\
Q_h\lt(k_0,m,g_h,\l_h\rt)&=&\prod_{k \neq k_0}\frac{1} {1 - \lambda_h \tau_k^{(h)}(g_h) e^{-s_h(k_0,m,g_h,\l_h) E_k}}.
\ee
Here is sum over $m$ is a symmetric limit $\lim_{M\to\infty}\sum_{m=-M}^M\cdots$ due to the principal value integral. See Appendix \ref{sum over m} for more discussion about this sum. Both functions $Q_h\lt(k_0,m,g_h,\l_h\rt)$ and $\Fs_h\lt(A_h,g_h,\l_h\rt)$ are smooth on $\mathscr{U}_h$, see Lemmas \ref{lemmaSmoothFs} and \ref{lemmaSmoothQ} for proofs.

Recall that $\tau_{k}^{(h)}(g_h)$ is the amplitude associated to a face carrying the spin $k/2$. The fact that the amplitude $\o_h^{\rm bos}$ is proportional to $e^{A_h r_h(k_0,g_h,\l_h)} = [\lambda_h \tau_{k_0}^{(h)}(g_h)]^{A_h/E_{k_0}}$ indicates that the dominant contribution comes from configurations where all faces are concentrated at the state $k_0$, in analogy with the Bose-Einstein condensation\footnote{Large $A_h$ corresponds to a large number of identical bosons and thus reduces the fluctuations in the grand canonical ensemble.}. Here, $A_h$ represents the total energy and $E_{k_0}$ is the energy of a face at spin $k_0$, so the total number of such faces is $p_h = A_h/E_{k_0}$. When $\l_h$ is such that the condensation spin $k_0/2$ is small, the condensation phenomenon occurs in the UV (i.e., small-$j$) regime. Given the LQG relation between spin and quanta of area, when this condensation occurs, the macroscopic area is given by the superposition of many microscopic quantum areas. As the coupling constant $\lambda$ decreases, the condensation spin increases, causing the theory to move away from the UV regime. In the limit $\lambda \to 0$, the spinfoam stack amplitude coincides with the standard spinfoam amplitude, which is supported on the root complex.

\begin{lemma}\label{lemmaBoundUc}

(1) For all $g_h$ in the complement ${\mathscr{U}}^c_h=\mathscr{G}_h\setminus \mathscr{U}_h$, there exists a constant $\Delta > 0$ and a function $C_h(g_h, \Delta) > 0$, which is continuous in $g_h$ and independent of $A_h$, such that
\be
\lt|\o_h^{\rm bos}(A_h;g_h,\l_h)\rt|\leq C_h(g_h,\Delta) e^{A_h\lt[\b_{k_0}(\l_h)-\Delta\rt]},\qquad \forall g_h\in {\mathscr{U}}_h^c.\label{outsidebound}
\ee

(2) For all $g_h\in \mathscr{G}_h$, and for any (infinitesimal) $\eps>0$, there exists a function $C_h(g_h,-\eps) > 0$, which is continuous in $g_h$ and independent of $A_h$, such that
\be
\lt|\o_h^{\rm bos}(A_h;g_h,\l_h)\rt|\leq C_h(g_h,-\eps) e^{A_h\lt[\b_{k_0}(\l_h)+\eps\rt]},\qquad \forall g_h\in \mathscr{G}_h.\label{outsidebound11}
\ee

\end{lemma}

\begin{proof} 
(1) Recall the definition of $\mathscr{U}_h$ in \eqref{scrUhDef}. We consider $g_h\not\in \mathscr{U}_h$. There are two cases for $k=k_0$: (1) If $g_h \notin \mathcal{V}$, i.e. $\tau_{k_0}^{(h)}(g_h) = 0$, there is no pole with positive real part associated with $k_0$. (2)  If $g_h \in \mathcal{V}$ but $g_h \notin \mathscr{U}_h$, by the definition of $\mathscr{U}_h$, we have
\be 
R_h({k_0}, g_h, \lambda_h) \le \beta^*. 
\ee
For $k \neq k_0$, for any $g_h$ and any $k \neq k_0$, the real part of the poles satisfies
\be 
R_h(k,g_h, \lambda_h) \le \beta_k(\lambda_h) \le \sup_{k \neq k_0} \beta_k(\lambda_h) = \beta^*.
\ee
Combining all these cases, for any $g_h \notin \mathscr{U}_h$, the real part of any pole $s$ of $\Xi(s_h)$ satisfies $\mathrm{Re}(s) \le \beta^*$, and in particular the maximum real part of the poles satisfies $\mathrm{Re}(s_*(g_h))\le \beta^*$. 

We define the constant $\Delta>0$ by
\be
\Delta =  \beta_{k_0}(\lambda_h) - \beta^* -\eps, 
\ee
for any (infinitesimal) $0<\eps<\beta_{k_0}(\lambda_h) - \beta^*$. Choose the contour parameter $T = \beta^*+\eps =\beta_{k_0}(\lambda_h)-\Delta$ with $\Delta>0$ for $\o_h^{\rm bos}$, we have the strict inequality: $ T > \beta^* \ge \mathrm{Re}(s_*(g_h)) $. The hypothesis in Lemma \ref{smoothboundB3} is satisfied for all $g_h\in{\mathscr{U}}^c_h$. As a result, there exists a function $C_h(g_h, \Delta)$ continuous in $g_h$ such that:
\be
\left|\omega_h^{\mathrm{bos}}(A_h;g_h,\lambda_h)\right| \le C_h(g_h, \Delta) \, e^{A_h T}. 
\ee

(2) For all $g\in\mathscr{G}_h$, the maximum of the real part satisfy $\re(s_*(g_h))\leq \b_{k_0}(\l_h)$. The bound \eqref{outsidebound11} is proven by choosing the contour parameter $T=\b_{k_0}(\l_h)+\eps$ for any infinitesimal $\eps>0$ and use Lemma \ref{smoothboundB3}.

\end{proof}

\section{Localization of stack amplitude}\label{Localization of stack amplitude}

Let us apply these results to the stack amplitude $\sa_\ck=\int \rmd\O\prod_h\o_h\prod_b\o_b$. We denote by $E_{\rm int}$ the set of edges not incident on the boundary, and define the space $\mathscr{G}_{\rm int}$ to be the space of $\{g_{ve}\}_{e\in E_{\rm int}}$ with the gauge fixing $g_{v, e_0(v)}=1$ implemented. For any internal face $h$, we define the continuous projection map 
\be
\fp_h:\ \mathscr{G}_{\rm int}\to \mathscr{G}_{h},\qquad \{g_{ve}\}_{e\in E_{\rm int}}\mapsto g_h, 
\ee
which is given by restricting $\{g_{ve}\}_{e\in E_{\rm int}}$ to $ g_h$. We define the subspaces
\be
\mathscr{U}_{\rm int}=\cap_h \fp^{-1}_h(\mathscr{U}_h),\qquad \cc^{\cs}_{\rm int}=\cap_h \fp^{-1}_h(\cc_h^{\varsigma_h}),\quad\text{where}\quad \cs=\{\varsigma_h\}_h,\ \varsigma_h=\pm1,\qquad \cc_{\rm int}=\cup_\cs \cc_{\rm int}^{\cs}.
\ee
The subspace $\cc_{\rm int}$ collects $\{g_{ve}\}_{e\in E_{\rm int}}$ satisfying \eqref{g0solution} for all $h$. The subspace $\mathscr{U}_{\rm int}\subset \mathscr{G}_{\rm int}$ is an open neighborhood containing $\cc_{\rm int}$, given by
\be
\mathscr{U}_{\rm int}=\lt\{\{g_{ve}\}_{e\in E_{\rm int}}\,\Big|\,  R_h(k_0, g_h, \l_h) > \b^*_h,\ \forall h\rt\},\qquad \b^*_h=\sup_{k \neq k_0} \b_k(\l_h)>0.\label{scrUintBetaStar}
\ee

\begin{lemma}
Given the root complex $\ck$ such that the intersection graph $\Gamma_{v,-}$ is connected for all vertex $v$ (recall the discussion in Section \ref{Spinfoam stacks}),

(1) We denote by $U_{vh}=g_{ve}^{-1}g_{ve'}$ for $v,e,e'$ belong to the boundary of an internal face $h$. Any sequence of configurations in $\mathscr{G}_{\rm int}$ that goes to infinity results in $\Vert U_{vh}\Vert\to \infty $ for at least one $h$ and one $v\in\partial h$, where $\Vert U \Vert = \sqrt{\tr(U U^\dagger)}$ for $2 \times 2$ matrix $U$.

(2) Any sequence in $\mathscr{G}_{\rm int}$ tending to infinity results in $\tau_k^{(h)}(g_h) \to 0$ for at least one internal face $h$.

(3) The spaces $\overline{\mathscr{U}}_{\rm int}$ (the closure of ${\mathscr{U}}_{\rm int}$), $\cc_{\rm int}^{\cs}$, and $\cc_{\rm int}$ are all compact.
\end{lemma}

\begin{proof}
(1) Let $\{x^{(n)}\}_{n\in\mathbb{N}}$ be a sequence of configurations in $\mathscr{G}_{\rm int}$. The space $\mathscr{G}_{\rm int}$ is a direct product of many copies of $\Slc$. We say the sequence goes to infinity if it leaves every compact subset of $\mathscr{G}_{\rm int}$. This is equivalent to the condition that the norm diverges \footnote{
	In terms of coordinate, any element $g \in \Slc$ can be written as $g = u_1 e^{r \sigma_3/2} u_2$, where $u_1, u_2 \in \Su$ and $r \in \mathbb{R}$. The "infinity" in $\Slc$ corresponds to taking $|r| \to \infty$. Equivalently, this asymptotic region can be characterized by $\Vert g \Vert =\sqrt{2\cosh(r)} \to \infty$.
}: 
\be
\lim_{n \to \infty} \max_{v, e} \| g_{ve}^{(n)} \| = \infty .
\ee
We proceed by contradiction. Assume that the sequence goes to infinity, but $U_{vh}$ remain uniformly bounded. That is, assume there exists a constant $M > 0$ such that for all $n$, all vertices $v$, and all faces $h$ incident to $v$:  
\be
\| U_{vh}^{(n)} \| \le M .
\ee
Consider an arbitrary vertex $v$. Let $\Gamma_{v,-}$ be the intersection graph when a sphere is used to cut out a small neighborhood of $v$ in the subcomplex $\mathcal{K}_-$. The nodes of $\Gamma_{v,-}$ are the internal edges $e \in E_{\rm int}$ incident to $v$, including the gauge-fixed edge $e_0(v)$. Two nodes $e_a, e_b$ are connected by a link in $\Gamma_{v,-}$ if they are the boundary edges of an internal face $h$ at $v$. We associate to this link the variable $U_{vh} = g_{v e_a}^{-1} g_{v e_b}$ (or its inverse).
The graph $\Gamma_{v,-}$ is connected. For any internal edge $e$ incident to $v$, there exists a path in $\Gamma_{v,-}$ connecting $e_0(v)$ to $e$. Let this path be $e_0(v)=\epsilon_0, \epsilon_1, \dots, \epsilon_m=e$ where $\epsilon_i$ are the nodes of the graph. We can express $g_{ve}$ as
\be
g_{ve} = g_{v \epsilon_0} (g_{v \epsilon_0}^{-1} g_{v \epsilon_1}) (g_{v \epsilon_1}^{-1} g_{v \epsilon_2}) \cdots (g_{v \epsilon_{m-1}}^{-1} g_{v \epsilon_m})=\mathbb{I} \cdot U_{v, h_0}^{\sigma_0} \cdot U_{v, h_1}^{\sigma_1} \cdots U_{v, h_{m-1}}^{\sigma_{m-1}},
\ee
using the gauge condition $g_{v \epsilon_0} = \mathbb{I}$ and identifying the terms in parentheses as link variables $U_{v, h_j}^{\pm 1}$. Since $\|U^{-1}\| = \|U\|$ and the norm is sub-multiplicative ($\|AB\| \le \|A\|\|B\|$). Thus, by the assumption that all $\| U_{vh}^{(n)} \| \le M$.
\be 
\| g_{ve}^{(n)} \| \le \|\mathbb{I}\| \prod_{j=0}^{m-1} \| U_{v, h_j}^{(n)} \| \le \sqrt{2} M^m
\ee
Since the complex is finite, the path length $m$ is bounded. Thus, $\| g_{ve}^{(n)} \|$ is bounded uniformly for all $v, e$. This contradicts the hypothesis that the sequence goes to infinity. Therefore, the assumption of boundedness must be false. There exists at least one pair $(v, h)$ such that $\| U_{vh}^{(n)} \|$ is unbounded. Specifically, $\lim_{n \to \infty} \max_{v,h} \| U_{vh}^{(n)} \| = \infty$.

(2) Recall the definition of the function $\tau_k^{(h)}$ in \eqref{zetakh}. We define $A(U) \equiv P_k D^{(k,\rho)}(U) P_k$, $U\in\Slc$, as an operator on the finite-dimensional subspace $\mathcal{H}_k$. The matrix elements of $A(U)$ are matrix coefficients of the unitary principal series representation of $\mathrm{SL}(2,\mathbb{C})$. For any non-trivial unitary irreducible representation of a non-compact semisimple Lie group, the matrix coefficients vanish at infinity \cite{2014arXiv1403.0223C} (see Lemma \ref{LemmaFgFxi} for the case of $\Slc$). Therefore
\be
\lim_{\|U\| \to \infty} \| A(U) \|_{k} = 0 
\ee
where $\| A(U) \|_k$ is the operator norm. Let $\{x^{(n)}\}$ be a sequence in $\mathscr{G}_{\rm int}$ tending to infinity. We have shown in the proof of (1) that there exists at least one pair $(v, h)$ such that $\| U_{vh}^{(n)} \|$ is unbounded in the sequence i.e. $\| U_{v h}^{(n)} \| \to \infty$ as $n \to \infty$. Along this subsequence, $\lim_{n \to \infty} \| A(U_{v h}^{(n)}) \|_k = 0 $. Using the sub-multiplicativity of the trace and operator norm, we obtain
\be
\lt|\tau_k^{(h)}(g_{h}^{(n)})\rt| \le d_k^2 \left( \prod_{v \in \partial h, v \neq v} \lt\| A(U_{v h}^{(n)}) \rt\|_k \right) \lt\| A(U_{v h}^{(n)}) \rt\|_k \to0
\ee
Thus, any sequence tending to infinity results in $\tau_k^{(h)} \to 0$ for at least one face.

(3) In the formula \eqref{scrUintBetaStar} for $\mathscr{U}_{\rm int}$, the condition $R_h > \beta^*_h$ is equivalent to:
\be
|\tau_{k_0}^{(h)}(g_h)| > \frac{1}{|\lambda_h|} e^{E_{k_0} \beta^*_h} \equiv \epsilon_h >0
\ee
The closure $\overline{\mathscr{U}}_{\rm int}$ is contained in the set defined by the non-strict inequalities\footnote{
Let $X$ be a Hausdorff topological space and $f: X \to \mathbb{R}$ be a continuous function. Let $\epsilon > 0$. Define:
$
U = \{ x \in X \mid f(x) > \epsilon \}
$ and 
$
K_{set} = \{ x \in X \mid f(x) \ge \epsilon \}.
$
We always have the inclusion:
$
\overline{U} \subseteq K_{set}
$
Therefore, the equality holds if and only if for every $x$ such that $f(x) = \epsilon$, every neighborhood $V$ of $x$ intersects $U$ (i.e., contains a point $y$ where $f(y) > \epsilon$). However, if there exists a point $x_0$ such that $f(x_0) = \epsilon$ and $x_0$ is a local maximum of $f$, then there exists a neighborhood $V$ of $x_0$ such that for all $y \in V$, $f(y) \le f(x_0) = \epsilon$. This implies $V \cap U = \emptyset$. 
}:
\be
K = \bigcap_h \lt\{ g \in \mathscr{G}_{\rm int} \Big|\, \lt|\tau_{k_0}^{(h)}(g_h)\rt| \ge \epsilon_h \rt\}
\ee
The set $K$ is closed because $\tau_{k_0}^{(h)}$ is continuous. We now show that $K$ is also bounded. Assume for contradiction that $K$ is unbounded. Then there must exist a sequence $\{x^{(n)}\} \subset K$ with $x^{(n)} \to \infty$. From part (2) above, for such a sequence, there is at least one internal face $h$ such that $\tau_{k_0}^{(h)}(\fp_h(x^{(n)})) \to 0$ as $n \to \infty$. However, this contradicts the defining property $|\tau_{k_0}^{(h)}(\fp_h(x^{(n)}))| \geq \epsilon_h > 0$ for all $n$, since $x^{(n)} \in K$. Thus, $K$ must in fact be bounded. As $\overline{\mathscr{U}}_{\rm int}$ is a closed subset of the bounded set $K$ in the finite-dimensional manifold $\mathscr{G}_{\rm int}$, it is therefore compact. The compactness of $\cc_{\rm int}$ and $\cc_{\rm int}^{\cs}$ can be proven similarly.

\end{proof}

We denote by $\sa_\ck|_{\cu}$ the integral restricting on any the compact neighborhood $\cu$ satisfying $\cc_{\rm int}\subset\cu\subset{\mathscr{U}}_{\rm int}$ \footnote{
	For example, choose $\epsilon'_h$ such that $\epsilon_h<\epsilon'_h<d_{k_0}^2$ and define $\cu=\bigcap_h \{ g \in \mathscr{G}_{\rm int} \,|\, |\tau_{k_0}^{(h)}(g_h)| \ge \epsilon_h' \}$, then $\cc_{\rm int}\subset\cu\subset{\mathscr{U}}_{\rm int}$. Moreover $\cu$ is a closed subset in the compact set $K$, so $\cu$ compact.
}.
\be
\sa_\ck\Big|_{\cu}&=&e^{\sum_h\b_{k_0}(\l_h)A_h}\int_{\cu}\rmd\O(g)\, e^{S(\vec{A},g,\vec{m},\vec{\l})}\prod_h \Fs_h(A_h,g_h,\l_h)\prod_b\o_b^{\rm bos}(A_b;g_b,\vec{H}_{\fl(b)},\l_b)\lt[1+O(A^{-\infty})\rt],\label{AKexponen0}
\ee
Both $\prod_h\Fs_h$ and $\prod_b\o_b^{\rm bos}$ are smooth on $\cu$. The effective action $S$ is expressed as
\be
S&=&\sum_h A_h \lt[r_h(k_0,g_h,\l_h)-\b_{k_0}(\l_h)\rt]=\sum_h \frac{A_h}{E_{k_0(h)}} \ln \lt[\frac{\t_{k_0}^{(h)}(g_h)}{d_{k_0(h)}^2}\rt].\label{actionS=}
\ee
 The notation $k_0(h)$ emphasizes that the condensation spin $k_0/2$ can be different for different $h$.



We consider the asymptotic regime that $A_h/E_{k_0(h)}$ are uniformly large among internal faces and use the stationary phase approximation to compute the integral. Since we have fixed $k_0$, we write $A_h\equiv A a_h$ for all $h$ and scale $A\to\infty$. The boundary data (namely, the function $\o_b$ along with $A_b$, $\vec{H}$, and $\l_b$) remain fixed. We employ the stationary phase approximation for the integration over those variables $\{g_{ve}\}_{e\in E_{\rm int}}$, with $E_{\rm int}$ denoting the set of edges not incident on the boundary: Due to $|\tau_k^{(h)}(g_h)|\leq d_k^2$, it is clear that for all $\{g_{ve}\}_{e\in E_{\rm int}}$
\be
\re(S)=\sum_h\frac{A_h}{E_h}\ln\lt(\frac{\lt|\tau_{k_0}^{(h)}(g_h)\rt|}{d_{k_0}^2}\rt)\leq0\ ,
\ee 
and $\re(S)=0$ if and only if $\{g_{ve}\}_{e\in E_{\rm int}}\in \cc_{\rm int}$. Moreover, to compute $\partial_g S$, we deform $g_{ve}\to g_{ve}(1+it^{(v,e)}_{IJ}J^{IJ})$, where $J^{IJ}$ is the Lie algebra generator. For any pair $(v_*,e_*)$ with $e^*\subset\partial h$, the derivative of $\t^{(h)}_k(g_h)$ with respect to $t^{(v_*,e_*)}_{IJ}$ on $\cc_{\rm int}$ gives
\be
\frac{\partial\tau_{k}^{(h)}(g_h)}{\partial t^{(v_*,e_*)}_{IJ}}\Bigg|_{\cc_{\rm int}}=\begin{cases}
	- i d_{k}\tr_{(k,\rho)}\lt[P_{k}J^{IJ} \mathring g_{v_*e_*}^{-1}\mathring g_{v_*e'_*}P_{k}\overrightarrow{\prod}_{v\neq v_*}P_{k} \mathring g_{ve}^{-1}\mathring g_{ve'}P_{k}\rt] \\
	i d_{k_i}\tr_{(k,\rho)}\lt[P_{k} \mathring g_{v_*e_*'}^{-1}\mathring g_{v_*e_*}J^{IJ}P_{k}\overrightarrow{\prod}_{v\neq v_*}P_{k} \mathring g_{ve}^{-1}\mathring g_{ve'}P_{k}\rt]
\end{cases}	
=\pm d_{k}\tr_{(k,\rho)}\lt[P_k J^{IJ}P_k\rt]=0,
\ee
where $\mathring g_{ve}\in \cc_{\rm int}$. The distinction between these two cases depends on whether $e_*$ is oriented as incoming or outgoing at $v$, according to the orientation of $\partial h$. In the second step, the defining condition \eqref{g0solution} characterizing $\cc_{\rm int}$ is applied. Therefore, 
\be
\partial_g S\Big|_{\cc_{\rm int}} = \sum_h \frac{A_h}{E_{k_0}} \frac{\partial_g \tau_{k_0}^{(h)}}{\tau_{k_0}^{(h)}} \Bigg|_{\cc_{\rm int}}=0
\ee
We conclude that $\cc_{\rm int}$ is the critical manifold that collect all $\{g_{ve}\}_{e\in E_{\rm int}}$ satisfying $\re(S)=\partial_g S=0$.

Asymptotically, the integral in $\sa_{\ck}|_\cu$ localizes onto the critical manifold $\cc_{\rm int}$. Performing the stationary phase expansion for each $\vec{m}$, we obtain
\be
\sa_{\ck}\big|_{\cu}
&=&\frac{e^{\sum_h\b_{k_0}(\l_h)A_h}}{A^{D_{\rm int}/2}}\sum_\cs\cf_\cs\lt(\vec A,\vec \l\rt)\int\limits_{\cc_{\rm int}^\cs}\frac{\rmd \mu(\vec \sig)}{\sqrt{\det\lt(-\mathbb{H}_\cs(\vec\sig)/(2\pi)\rt)}}\int \prod_{(v,e_b)} \rmd g_{ve_b}\prod_b\o_b^{\rm bos}\Big|_{\cs,\vec \sig} \lt[1+O(A^{-1})\rt],\label{AKexponen}
\ee
where $\mathbb{H}_\cs(\vec\sig)=A^{-1}\partial^2_g S\big|_{\cs,\vec \sig}$ is the Hessian matrix for the transverse directions to $\cc_{\rm int}^\cs$. The notation $|_{\cs,\vec \sig}$ indicates the restriction onto $\cc_{\rm int}^\cs$, and $\vec\sig$ denotes the coordinates on $\cc_{\rm int}^\cs$. The function $\cf_\cs(\vec A,\vec \l)$ is defined as the restriction of $e^{S(\vec{A},g,\vec{m},\vec{\lambda})} \prod_{h} F_{h}(g_{h},m_{h},\lambda_{h})$ onto $\cc_{\rm int}^\cs$. Its dependence on $g$ appears only through $\tau_k^{(h)}$, which is constant over $\cc_{\rm int}^\cs$. Therefore, $\cf_\cs(\vec A,\vec \l)$ is constant on $\cc_{\rm int}^\cs$.
\be
\cf_\cs\lt(\vec A,\vec \l\rt)&=&e^{S(\vec{A},g,\vec{m},\vec{\lambda})}\prod_h\Fs_{h}(A_h,g_{h},\lambda_{h})\Big|_{\cs,\vec \sig}
=\prod_{h}\sum_{m_h}\frac{e^{\frac{iA_{h}}{E_{k_{0}}}\left[\varphi_{h}+\pi\Theta(-\varsigma_{h})+2\pi im_{h}\right]}\mathring Q_h\lt(m_h,\lambda_h,\varsigma_h\rt)}{\beta_{k_{0}}(\lambda_{h})E_{k_{0}}+i\left[\varphi_{h}+\pi\Theta(-\varsigma_{h})\right]+2\pi im_{h}}\\
\mathring Q_h\lt(m_h,\lambda_h,\varsigma_h\rt)&=&\prod_{k\neq k_{0}}\frac{1}{1-\lambda_{h}\varsigma_{h}d_{k}^{2}e^{-\beta_{k_{0}}(\lambda_{h})E_{k}-\frac{iE_{k}}{E_{k_{0}}}\left[\varphi_{h}+\pi\Theta(-\varsigma_{h})+2\pi m_{h}\right]}} ,\qquad \varphi_h=\mathrm{Arg}(\l_h).
\ee
In the formula \eqref{AKexponen}, $e_b$ denotes the edges connecting to the boundary and $g_{ve_b}\in\Slc$. The exponent $D_{\rm int}/2>0$ is half of the dimension of the Hessian matrix $\mathbb{H}_\cs(\vec\sig)$. The Hessian matrix $\mathbb{H}_\cs(\vec\sig)$ does not depend on $g_{ve_b}$ because $S$ only depends on $g_h$. The Hessian matrix also does not depend on $\vec{m}$. The nondegeneracy of $\mathbb{H}_\cs(\vec\sig)$ is discussed in Section \ref{Nondegenerate Hessian matrix}. The induced measure on $\cc_{\rm int}$ from the Haar measure $\prod_{v,e\neq e_b}\rmd g_{ve}$ is denoted by $\rmd\mu(\vec\sig)$, whose explicit expression is derived in Section \ref{Some explicit computations}. 

Observe that the smooth function $u = \prod_h \Fs_h \prod_b \omega_b^{\rm bos}$ appearing in the stationary phase integral $\int d\Omega\, e^S u$ of \eqref{AKexponen0} depends on the scaling parameter $A$, since $\Fs_h$ depends on $A$. It is important to verify that the error term in the stationary phase expansion for \eqref{AKexponen} remains of order $O(A^{-1})$. The stationary phase error is controlled by $A^{-1} \sum_{|\alpha| \leq 2} \sup_{\cu} | D_g^{\alpha} u |$~\cite{stationaryphase}, where $D_g^\alpha$ denotes a multi-index derivative as defined in Appendix~\ref{Bound the derivatives of tau}. By the continuity of $\fp_h$, the compactness of $\cu$ implies the compactness of $\fp_h(\cu)\in \mathscr{U}_h$, then Lemma~\ref{lemmaSmoothFs} establishes that $D_g^\alpha \Fs_h$ is uniformly bounded by a constant that does not depend on $A$. As a result, $\sum_{|\alpha|\leq 2}\sup_\cu | D_g^{\alpha} u |$ is independent of $A$, ensuring that the error term is indeed of order $O(A^{-1})$.

The function $|\cf_\cs|$ is bounded and periodic in $A$ (see Appendix \ref{sum over m}). Therefore, the leading asymptotics of $\sa_{\ck}|_{\cu}$ is exponentially growing as $A\to\infty$ \footnote{The asymptotics depends on the choice of $E_k$ and the cut-off function $\a_{p_h,k}$, because the cut-off $A_h$ and the uniform limit $A_h\to\infty$ may be different for different choices. But our final results \eqref{AKexponen41} and \eqref{completeAmpliResult} does not depend on the choice.}:
\be
\sa_{\ck}\big|_{\cu}\sim e^{A\sum_h\b_{k_0}(\l_h)a_h}A^{-D_{\rm int}/2}.
\ee

\begin{lemma}
For any compact neighborhood $\wt{\cu}\subset \mathscr{G}_{\rm int}$ that have no intersection with $\mathscr{U}_{\rm int}$, the integral of $\sa_\ck$ on $\wt{\cu}$ exponentially suppresses relative to $\sa_{\ck}|_{\cu}$, i.e. there exists $B>0$ such that
\be
\lt|\sa_{\ck}\big|_{\wt{\cu}}\rt|\leq B e^{\b' A},\qquad \b' < \sum_h\b_{k_0}(\l_h)a_h.
\ee

\end{lemma}

\begin{proof}
The complement $\mathscr{G}_{\rm int}\setminus \mathscr{U}_{\rm int}$ is given by the union of regions where at least at one face $\fp_h(g)\not\in \mathscr{U}_h$, $g=\{g_{ve}\}_{e\in E_{\rm int}}$. For any compact neighborhood $\wt{\cu}\subset \mathscr{G}_{\rm int}\setminus \mathscr{U}_{\rm int}$, we define
\be
W_h\equiv \lt\{ g \in \mathscr{G}_{\rm int}\, \big|\, \fp_h\lt(g\rt) \notin{\mathscr{U}}_h \rt\}.
\ee
The set $W_h$ is closed in $\mathscr{G}_{\rm int}$ because $\mathscr{G}_h\setminus \mathscr{U}_h$ is closed and $\fp_h$ is continuous. The set $V_h$ is identified as the intersection $V_h = \wt{\cu} \cap W_h$, so $V_h$ is compact in $\mathscr{G}_{\rm int}$, and $\wt{\cu} = \bigcup_{h} V_h$. Then we have:
\be
\left| \mathscr{A}_{\mathcal{K}}\big|_{\wt{\cu}} \right| \leq \sum_{h} \int_{V_h} \rmd\Omega(g) \left| \prod_{h} \omega_h^{\rm bos}(g) \prod_{b} \omega_b^{\rm bos}(g) \right|. 
\ee
It suffices to bound each term in this sum individually. Consider a fixed internal face $h'$ and the integral over the region $V_{h'}$. By definition, for any configuration in $V_{h'}$, the group variables $g_{h'}$ satisfy $g_{h'} \notin {\mathscr{U}}_{h'}$. We apply Lemma \ref{lemmaBoundUc}. There exists a constant $\Delta > 0$ and a function $C_{h'}(g_{h'},\Delta)$ continuous in $g_{h'}$ such that $\left| \omega_{h'}^{\rm bos} \right| \leq C_{h'}(g_{h'},\Delta) e^{A_{h'} [\beta_{k_0}(\lambda_{h'}) - \Delta]}$, while for other $h\neq h'$, there exists a function $C_h(g_h,-\eps) > 0$ continuous in $g_h$ such that $\lt|\o_h^{\rm bos}\rt|\leq C_h(g_h,-\eps) e^{A_h\lt[\b_{k_0}(\l_h)+\eps\rt]},\ \forall  g_h\in \mathscr{G}_h,\  \forall \eps>0$. Combining the bounds, we obtain
\be
\int_{V_h} \rmd\Omega(g) \left| \prod_{h} \omega_h^{\rm bos}(g) \prod_{b} \omega_b^{\rm bos}(g) \right|\leq e^{\b'A}\int_{V_h} \rmd\Omega(g) C_{h'}(g_{h'},\Delta) \prod_{h\neq h'} C_h(g_h,-\eps) \prod_{b} \lt|\omega_b^{\rm bos}(g) \right|
\ee
The integrand is a continuous function, so the integral gives a finite value, whose sum over $h$ gives the constant $B$. Moreover,
\be
\b'=\sum_h a_h\b_{k_0}(\l_h)-\Delta a_{h'}+\eps \sum_{h\neq h'}a_h<\sum_h a_h\b_{k_0}(\l_h),
\ee
because $\Delta>0$ and $\eps$ is arbitrarily small. 

\end{proof}

Let $D$ be any compact integration domain containing ${\mathscr{U}}_{\rm int}$. The domain $D$ is the disjoint union of compact sets $\cu \subset {\mathscr{U}}_{\rm int}$ and $\wt{\cu} = D \setminus {\mathscr{U}}_{\rm int}$ and an open set ${\mathscr{U}}_{\rm int} \setminus \cu$. The above shows that the spinfoam integral of $\sa_\ck$ over $\wt{\cu}$ exponentially suppresses relative to $\sa_{\ck}|_{\cu}$. By the regularity of the Haar measure, for any $\epsilon > 0$, we can choose the compact neighborhood $\cu$ so that the measure of ${\mathscr{U}}_{\rm int} \setminus \cu$ is less than $\epsilon$, which implies that the contribution of the spinfoam integral over ${\mathscr{U}}_{\rm int} \setminus \cu$ can be made arbitrarily small. Furthermore, since the stack amplitude $\sa_\ck$ is absolutely convergent, its integral over the complement of a sufficiently large compact set $D$ is arbitrarily small. Therefore, for large $A$, the leading order asymptotics of $\sa_\ck$ is determined by the behavior described in \eqref{AKexponen}, with any other contributions being exponentially suppressed. That is,
\be
\sa_\ck = \sa_\ck\big|_{\cu} \left[ 1 + O(A^{-\infty}) \right].
\ee

\section{Parametrization of critical manifold}\label{Some explicit computations}

In this section, we provide explicit parametrizations of the group variables $g_{ve}$ and derive an explicit expression of the measure $\rmd\mu(\vec{\sigma})$ in \eqref{AKexponen}.

Given the root complex $\ck$, we number the vertices by $v=v_i$, $i=1,\cdots, n$. For every edge $e=(i,j)$ for certain $i,j=1,\cdots,n$ ($i\neq j$), the pair of group variables $g_{v_ie} $ and $g_{v_j e}$ are re-labelled as $g_{ij}$ and $g_{ji}$. We use the following decomposition to parametrize each $g_{ij}\in\Slc$
\be
g_{ij}=u_{ij}e^{-ir_{ij}K^{3}}v_{ij},\qquad u_{ij}=e^{-i\psi_{ij}^{\prime}L^{3}}e^{-i\theta_{ij}^{\prime}L^{2}}\in\Su,\qquad v_{ij}=e^{-i\psi_{ij}L^{3}}e^{-i\theta_{ij}L^{2}}e^{-i\phi_{ij}L^{3}}\in\Su,
\ee
where $r_{ij}\in\R$ and $\psi_{ij}^{\prime},\theta_{ij}^{\prime},\psi_{ij},\theta_{ij},\phi_{ij}$ are Euler angles. To fix the $\Slc$ gauge freedom, we set $r_{ij}=0$ and $u_{ij}=v_{ij}=1$ for $g_{ij}=g_{v e_0(v)}$ and still allow SU(2) gauge transformation at the vertex. The Lie algebra generators relate to the Pauli matrices $\bm\sig^{i=1,2,3}$ by $L^i=\frac{1}{2}\bm\sigma^i,\ K^i=\frac{i}{2}\bm\sigma^i$ in the fundamental representation. The Haar measure $\rmd g_{ij}$ reads
\begin{eqnarray}
\rmd g_{ij}=\frac{1}{4\pi}\sinh^2\lt(r_{ij}\rt)\rmd r_{ij}\rmd u_{ij}\rmd v_{ij}.\label{gruv}
\end{eqnarray}
where $\rmd u_{ij},\rmd v_{ij}$ are Haar measures on SU(2).

For any internal face $h$, we label the vertices of $h$ by $1,\cdots,m$,
\be
\t_{k}^{(h)}&=&d_{k}\mathrm{Tr}_{(k,\rho)}\left[g_{12}P_{k}g_{21}^{-1}g_{23}P_{k}g_{32}^{-1}\cdots g_{ij}P_{k}g_{ij}^{-1}\cdots g_{m1}P_{k}g_{1m}^{-1}\right]
\ee
where
\be
g_{ij}P_{k}g_{ji}^{-1}=\left(u_{ij}e^{-ir_{ij}K^{3}}v_{ij}\right)P_{k}\left(v_{ji}^{-1}e^{ir_{ji}K^{3}}u_{ji}^{-1}\right).
\ee
is the ``holonomy'' along the edge $(i,j)$. The SU(2) variables $v_{ij}$ and $v_{ji}^{-1}$ commute with $P_k$, and one of them is associated to the SU(2) gauge redundancy. The gauge-invariant degrees of freedom is contained in the combination $v_{ij}v_{ji}^{-1}$. We introduce new variables to parametrize this combination:
\be
v_{ij}v_{ji}^{-1} = u_{ij}^{-1} H_{ij} u_{ji},\label{vvuHu}
\ee
where we define $H_{ij}\equiv H_e$ as an SU(2) holonomy along the internal edge $e \in E_{\rm int}$, and $H_{ji} \equiv H_{ij}^{-1}$. In the integration measure, one of $\rmd v_{ij}$ or $\rmd v_{ji}$ is redundant and can be dropped; we identify the remaining measure with $\rmd H_{ij}$.

Notice that when the gauge fixing sets $g_{ij}=g_{v,e_0(v)}=\mathbb{I}$, this corresponds to $u_{ij}=v_{ij}=1$ and $r_{ij}=0$ in the change of variables \eqref{vvuHu}. This gauge fixing prescription only sets at most one of $v_{ij}$ or $v_{ji}$ to the identity for each edge $(i,j)$. The SU(2) holonomy $H_{ij}$ is not fixed by the gauge, which follows from the injectivity condition $e_0(v)\neq e_0(v')$ for $v\neq v'$ (see Lemma \ref{lemmaInjectiveVE}).


Recall that the critical manifold is expressed as $\cc_{\rm int} = \bigcup_{\cs} \cc_{\rm int}^\cs$, where each sector $\cc_{\rm int}^\cs$ is labeled by $\cs = \{\varsigma_h\}_h$, with $\varsigma_h = \pm 1$ assigned to every internal face $h$. The space $\cc_{\rm int}^\cs$ consists of collections $\{g_{ve}\}_{e \in E_{\rm int}}$ such that
\be
g_{ve}^{-1}g_{ve'} \in \mathrm{SU}(2), \quad \text{for all } e, e' \subset \partial h, \; e \cap e' = v,
\qquad
\overrightarrow{\prod_{v \in \partial h}} g_{ve}^{-1} g_{ve'} = \varsigma_h \, \mathbb{I}, \qquad \varsigma_h = \pm 1.
\label{g0solution22}
\ee
Let us consider the first condition. At any vertex $v$ and any internal face $h$ bounded by $v$, $e_0$ (where the gauge fixing $g_{ve_0}=1$ is imposed) and another edge $e_1$ connecting $v$, $g_{ve_0}^{-1} g_{ve_1}\in\Su$ restricts $g_{ve_1}\in\Su$. For any other internal face $h'$ bounded by $v$ and $e_1$ and another edge $e_2$ connecting $v$, $g_{ve_1}^{-1} g_{ve_2}\in\Su$ restricts $g_{ve_2}\in\Su$. The restriction can propagate to all $e$ connecting to $v$ and thus gives $g_{ve}\in\Su$ for all $e\in E_{\rm int}$, because of the connectivity of $\G_{v,-}$. So we have for all edges in $E_{\rm int}$
\be
r_{ij}\approx 0.\label{H=1}
\ee
Here and in the following, we use $\approx$ for the equality that holds only on the critical manifold $\cc_{\rm int}$. It implies
\be
g_{ij}P_{k}g_{ji}^{-1}\approx  H_{ij}.
\ee
Then the second critical point condition in \eqref{g0solution} further restricts $H_{ij}$ to satisfy the $\cs$-flatness condition, i.e.
\be
H_{12}H_{23}\cdots H_{m1}\approx\varsigma_h\mathbb{I}.\label{HHH}
\ee
for any internal face $h$. Recall that the subcomplex $\ck_-$ (the interior of $\ck$) is the 2-complex consisting of only internal faces $h$, edges $e\in E_{\rm int}$ (on the boundary of $h$'s) and vertices as the end points of edges. 


The signs $\cs = \{\varsigma_h\}_h$ assigned to each internal face $h$ must satisfy a compatibility condition with respect to the 2-cycles of $\ck_-$: Let $C = \sum n_h h$ be any 2-cycle in $Z_2(\ck_-, \mathbb{Z})$. The product of holonomies around the faces in $C$ corresponds to the holonomy of a trivial loop in the 1-skeleton, which must be $\mathbb{I}$. This imposes the following necessary condition for the critical manifold $\cc^\cs_{\rm int}$ to be non-empty:
\be
\prod_{h} \varsigma_h^{n_h} = 1 \qquad \text{for every generator of } H_2(\ck_-).\label{closedCycle=1}
\ee
If this condition fails, then $\cc^\cs_{\rm int} = \emptyset$; that is, there are no critical points for such an assignment. Therefore, in the sum over $\cs$ appearing in the stationary phase approximation \eqref{AKexponen}, we restrict to those assignments of $\{\varsigma_h\}_h$ that satisfy this compatibility condition. In what follows, we assume that the assignment of $\varsigma_h$ is admissible\footnote{For simply-connected $\ck_-$, Eq.\eqref{closedCycle=1} implies $\varsigma_h=\prod_{e\subset\partial h} \varsigma_e$, where $\varsigma_e=\pm1$ is defined on each edge.}.

The on-shell gauge freedom $\cg_{\rm int}$ include both $u_{ij}$ and
\begin{eqnarray}
H_{ij}\to x_i H_{ij} x_j^{-1},\qquad x_i\in\Su,
\end{eqnarray}
They are called on-shell gauge freedom because the integrand $e^S\int\prod_{e_b}\rmd g_{ve_b}\prod \o_b$ does not depend on them on $\cc_{\rm int}$.

Choosing a base vertex $v_*$ and a maximal spanning tree $\ct$ in the 1-skeleton of $\ck_-$. A maximal spanning tree is a subgraph that connects all vertices but contains no loops, and for any vertex $v$, there is a unique path $\calp_{v_*\to v}$ within the tree $\ct$ from $v_*$ to $v$. For a give set of $\{H_{e}\}_{e\in E_{\rm int}}$, we define the gauge transformation with
\begin{eqnarray}
x_i=\mathfrak{hol}({\calp_{v_*\to i}}),	
\end{eqnarray}
being the SU(2) holonomy made by $H_{e}$ traveling from $v_*$ to the vertex $i$. The gauge transformation set $H_{e}=1$ along all edges in $\ct$. Moreover, we adopt $\{H_{e}\}_{e\in \ct}\to \{x_i\}_{i\neq v_0}$ to be a change of variables\footnote{For any vertex $i$ in $\ct$, there is a unique vertex $j$ such that $(i,j)$ is an edge and the unique path $\calp_{v_*\to i}$ travels through $j$. Then $x_i=H_{ij}x_j$ and $\rmd x_i=\rmd H_{ij}$.}. After the change of variables, $\t_k^{(h)}$ can be generally written as
\be
\t_{k}^{(h)}&=&d_{k}\mathrm{Tr}_{(k,\rho)}\left[\overrightarrow{\prod_{(i,j)\subset \partial h}} \left(u_{ij}e^{-ir_{ij}K^{3}}u_{ij}^{-1}\right)P_{k}\lt(x_i H_{ij}x_j^{-1}\rt)P_k\left(u_{ji}e^{ir_{ji}K^{3}}u_{ji}^{-1}\right)\right],
\ee
where $H_{ij}\neq 1$ only for $(i,j)\not\subset\ct$.


Let us first consider a general topologically non-trivial $\ck_-$. For each edge $e=(u,v) \notin \ct$, there is a unique fundamental loop $\ell_e$ based at $v_*$ formed by the path $\calp_{v_*\to u}$ in $\ct$, the edge $e$, and the path $\calp_{v\to v_*}$ in $\ct$. All other loops in the 1-skeleton of $\ck_-$ are products of fundamental loops, so a minimal set of generator of the fundamental group $\pi_1(\ck_-)$ can be chosen to be fundamental loops. Then along each generator, there is only a single edge that does not belong to $\ct$. We denote by $\Fl$ the set of edges along these generators. We define a local coordinate chart in the space of $\{g_{ve}\}_{e\in E_{\rm int}}$. The coordinates are given by\footnote{For e.g. $H_e\in \Su$, the coordinates can be chosen as Euler angles.}
\be
\{r_{ij}\}_{(i,j)},\ \{u_{ij}\}_{(i,j)},\ \{x_i\}_{i\neq v_0},\ \{H_l\}_{l\in\Fl, l\not\in\ct},\ \{H_e\}_{e\not\in\Fl, e\not\in\ct}.\label{parametri}
\ee
The boost parameters $\{r_{ij}\}_{(i,j)}$ and the holonomies $\{H_e\}_{e\not\in\Fl,\, e\not\in\ct}$ (i.e., those not associated to the spanning tree or cycle generators) describe directions transverse to the critical manifold $\cc_{\rm int}$. The critical point conditions constrain $r_{ij}\approx 0$ and determine all $\{H_e\}_{e\not\in\Fl,\, e\not\in\ct}$ in terms of the remaining variables.

Recall that $\cc_{\rm int}$ is a disjoint union of $\cc_{\rm int}^\cs$, where $\cs$ collects the signs in the $\cs$-flatness condition \eqref{HHH} at all $h$. $\vec\sig$ described above are coordinates on each $\cc_{\rm int}^\cs$ but are independent of $\cs$. Under the gauge that $H_e=1$ for $e\subset\ct$ and fixing $\cs$, the remaining $\{H_e\}_{e\not\in\Fl, e\not\in\ct}$ are uniquely determined by using \eqref{HHH} (see Lemma \ref{HeenotinFlenotinct}). However, different $\cs$ generally leads to different $\{H_e\}_{e\not\in\Fl, e\not\in\ct}$. One might view $\{H_e\}_{e\not\in\Fl, e\not\in\ct}$ are multi-valued functions of $\vec{\sig}$ and become single-valued on the ``cover space'' $\cc_{\rm int}$.

When $\pi_1(\ck_-)$ is nontrivial, the fundamental group $\pi_1(\ck_-)$ has generators that are subject to algebraic relations. These in turn induce nontrivial constraints among the holonomies $\{H_l\}_{l\in\Fl,\,l\notin\ct}$. Such relations can lead to singularities in the space of SU(2) holonomies that satisfy the $\cs$-flatness condition \eqref{HHH}, and hence introduce singularities into $\cc_{\rm int}^\cs$ \footnote{
	For example, on a 2-torus, the relation $H_A H_B = H_B H_A$ of holonomies along $A$ and $B$ cycles produces a singularity when $(H_A, H_B) = (1,1)$.
}. As a result, the critical manifold $\cc_{\rm int}$ may fail to be smooth in general, which complicates the application of the stationary phase approximation. For these reasons, we will henceforth restrict to the case where $\ck_-$ is simply connected, postponing the general analysis to the future investigation. Physically, simple connectivity is valid for all local patches of spacetime. With this restriction, the variables $\{H_l\}_{l\in\Fl,\,l\notin\ct}$ are absent, so the parametrization \eqref{parametri} reduces to
\be
\{r_{ij}\}_{(i,j)},\ \{u_{ij}\}_{(i,j)},\ \{x_i\}_{i\neq v_0},\ \{H_e\}_{e\not\in\ct},\label{parametri2}
\ee
In this parametrization, the variables $\{r_{ij}\}_{(i,j)}$ and $\{H_e\}_{e\not\in\ct}$ span the directions transverse to the critical manifold $\cc_{\rm int}$. Restricting to the critical manifold $\cc_{\rm int}^\cs$, the conditions $r_{ij} \approx 0$ is imposed, and for all $e \notin \ct$, the $\cs$-flatness constraint enforces $H_e = \pm \mathbb{I}$, where the sign is uniquely fixed by $\cs$ for each $e$ (see Lemma \ref{HeenotinFlenotinct}). The Hessian matrix $\mathbb{H}$ arises as the matrix of second derivatives of $S$ with respect to these transverse variables. The variables $\{u_{ij}\}_{(i,j)},\ \{x_i\}_{i\neq v_0},\ \{H_e\}_{e\in\Fl,\, e\not\in\ct}$ provide coordinates $\vec\sig$ on $\cc_{\rm int}$.

To construct the measure on $\cc_{\rm int}$, we replace the non-redundant $\rmd v_{ij}$ in \eqref{gruv} with $\rmd H_{ij}$, and decompose the product of Haar measures over all internal edges as $\prod_{e\in E_{\rm int}}\rmd H_e = ( \prod_{e \in \ct} \rmd H_e ) ( \prod_{e \notin \ct} \rmd H_e )$. Here, $\prod_{e \in \ct} \rmd H_e$ can be rewritten as $\prod_{i \neq v_*}\rmd x_i$ by the change of variables. After removing $\prod_{e \notin \ct} \rmd H_e$ corresponding the transverse directions, the measure $\rmd\mu(\vec{\sigma})$ on $\cc_{\rm int}^\cs$ takes the form
\begin{eqnarray}
\rmd\mu(\vec{\sig})=\prod_{(i,j)}\frac{\rmd u_{ij}}{4\pi}\prod_{i\neq v_*} \rmd x_i,
\end{eqnarray}
which is independent of $\cs$. As a result of the gauge fixing, the integration over $\rmd u_{ij}$ associated with the edge $e_0(v)$ is omitted. Correspondingly, the critical manifold has the structure
\be
\cc_{\rm int}^\cs \cong \Su^{|E_{\rm int}|-|V|} \times \Su^{|V|-1} = \Su^{|E_{\rm int}|-1}~.
\ee
where $|V|$ denotes the number of vertices in $\ck$. 

Note that the SU(2) gauge symmetry associated to the base vertex $v_*$ does not contribute to the degrees of freedom parametrizing $\cc_{\rm int}^\cs$. The reason is: Lattice SU(2) gauge transformations $\{x_i\}$ act transitively on the space of solutions to the $\cs$-flatness condition. All such solutions can be reached via gauge transformations acting on the reference configuration $H_e = \pm\mathbb{I}$ for all $e\in E_{\rm int}$. However, the gauge transformation at $v_*$ corresponds to the global SU(2) transformation at all vertices by the gauge fixing on $\ct$, so it leaves the reference configuration invariant and therefore acts as the stabilizer, rather than introducing an degree of freedom $\cc_{\rm int}^\cs$.

\section{Hessian matrix}\label{Nondegenerate Hessian matrix}

Using the parametrization \eqref{parametri}, we integrate out the transverse coordinates $\{r_{ij}\}_{(i,j)},\{H_e\}_{e\not\in\ct}$ in the stationary phase integral \eqref{AKexponen0}. The Hessian matrix $\mathbb{H}$ is computed with respect to the transverse directions. In the following, the coordinate index for the transverse directions are denoted by $\a,\b$. To simplify the formulae, we assume $A_h=A$ and $k_0$ to be constant over different $h$ ($\l_h$ is constant over $h$), the expression of the Hessian matrix is given by
\be
\mathbb{H}_{\a\b}=\sum_{h}C_0^{-1}\partial_{\a}\partial_{\b}\t_{k_0}^{(h)}\left(g\right)\Big|_{\cc_{\rm int}},\qquad C_0= d_{k_0}^{2}E_{k_0}.
\ee

To compute the second derivatives of $\t_k^{(h)}$, it is convenient to expand
\be
e^{-ir_{ij}K^{3}}&=& 1-ir_{ij}K^{3}-\frac{1}{2}r_{ij}^{2}\left(K^{3}\right)^{2}+O(r^3),\\
H_{e}&=&\mathring H_e(\cs) e^{-i\sum_{a=1}^{3}t_{e}^{a}L^{a}}={\mathring H_e}(\cs)\lt[1-i\sum_{a=1}^{3}t_{e}^{a}L^{a}-\frac{1}{2}\sum_{a,b=1}^{3}t_{e}^{a}t_{e}^{b}L^{a}L^{b}+O(t^3)\rt],
\ee
where $t^a_{e}=-t^a_{e^{-1}}$. The perturbations $r_{ij}$ and $t^a_e$ represent directions transverse to the critical manifold $\cc_{\rm int}$, and are associated respectively with boosts on the edges $(i,j)$ and with SU(2) holonomies $H_e$ for edges $e$ not lying on the maximal spanning tree. On $\cc_{\rm int}^\cs$, each $H_e$ restricts to $\mathring H_e(\cs) = \pm \mathbb{I}$, with the sign determined by the choice of $\cs$.

The resulting Hessian matrix $\mathbb{H}_{\a\b}$ is a polynomial of the Barbero-Immirzi parameter $\g$. But $\mathbb{H}_{\a\b}$ becomes simplified if we only focus on the leading order of small $\gamma$. In particular, due to the simplicity constraint and $\rho=\g (k+2)$, we have \cite{generalize} \footnote{Recall that here $k=2j$ is an integer.}
\be
\langle k,m|P_{k}K^{3}P_{k}|k,n\rangle=-\gamma\langle k,m|L^{3}|k,n\rangle=O\left(\gamma\right).
\ee
The Hessian matrix is a direct sum of $r$-$r$ and $t$-$t$ blocks as $\g\to0$, due to the off-diagonals
\be
\frac{\partial^{2}}{\partial r_{ij}\partial t_{mn}^{a}}\t_{k_0}^{(h)}\Big|_{\cc_{\rm int}^\cs}&=&\pm d_{k_0}\mathrm{Tr}_{(k_0,\rho_0)}\left[\cdots\left(u_{ij}^{\mp 1}K^{3}u_{ij}^{\pm 1}\right)P_{k_0}\left(x_{i}\mathring{H}_{ij}x_{j}^{-1}\right)P_{k_0}\cdots P_{k_0}\left(x_{m}\mathring{H}_{mn}L^{a}x_{n}^{-1}\right)P_{k}\cdots\right]\nonumber\\
&=&\pm d_{k_0}\mathrm{Tr}_{(k_0,\rho_0)}\left[\cdots\left(u^{\mp1}_{ij}P_{k_0}K^{3}P_{k_0}u_{ij}^{\pm1}\right)\left(x_{i}\mathring{H}_{ij}x_{j}^{-1}\right)\cdots \left(x_{m}\mathring{H}_{mn}L^{a}x_{n}^{-1}\right)\cdots\right]\nonumber\\
&=&O\left(\gamma\right).\label{rtoffdiagonal}
\ee
where $\rho_0 = \gamma(k_0 + 2)$, and the $\pm$ sign corresponds to whether $r_{ij}$ associates with $g_{ij}$ or $g_{ij}^{-1}$. In the second step, we use the property that all operators in the trace commute with $P_k$ except for $K^3$.

Furthermore, the $r$-$r$ block is block-diagonal, where each small block associates to $r_{ij}$ at a given vertex $i$. Indeed, at the vertex $i$, we have the diagonal entries
\be
\sum_{h}C_0^{-1} \frac{\partial^{2}}{\partial r_{ij}^{2}}\t_{k_{0}}^{(h)}\Big|_{\cc^\cs_{\rm int}}
&=&-\sum_{h;(i,j)\subset\partial h}C_0^{-1} \varsigma_h^{k_0} d_{k_{0}} \mathrm{Tr}_{(k_{0},\rho_{0})}\left[P_{k_{0}}\left(K^{3}K^{3}\right)P_{k_{0}}\right]\nonumber\\
& =&-\sum_{h;(i,j)\subset\partial h}C_1\varsigma_h^{k_0}+O\left(\gamma^2\right),\qquad \qquad C_1^{(h)}=\frac{1}{6}d_{k_0}^{2}\left(d_{k_0}+1\right)C_0^{-1}.\label{rrdiagonal}
\ee
In the first step, we have used the $\cs$-flatness condition \eqref{HHH} and the relation $D^{k_0}(\varsigma_h \mathbb{I})=\varsigma_h ^{k_0} D^{k_0}( \mathbb{I})$ for Wigner $D$-matrix on $\ch_{k_0}$. The off-diagonal entries are computed
\be
\sum_{h}C_0^{-1}\frac{\partial^{2}}{\partial r_{ij}\partial r_{im}}\t_{k_{0}}^{(h)}\Big|_{\cc_{\rm int}^\cs}
&=&\sum_{h;(i,j),(i,m)\subset\partial h}C_0^{-1}\varsigma_h^{k_0} d_{k_{0}}\mathrm{Tr}_{(k_{0},\rho_{0})}\left[P_{k_{0}}\left(K^{3}U_{jim}K^{3}\right)P_{k_{0}}U_{jim}^{-1}\right]\nonumber\\
&=&\sum_{h;(i,j),(i,m)\subset\partial h}C_1\varsigma_h^{k_0} \cos\lt(\theta_{jim}\rt)+O\left(\gamma\right).\label{rroffdiagonal}
\ee
Here $\theta_{jim}$ is one of the Euler angles of $U_{jim}=u^{-1}_{ij}u_{im}\in\Su$, and $U_{jim}^{-1}$ appears because of \eqref{g0solution22}. For $r_{im}$ and $r_{jn}$ associate to two different vertices $i\neq j$
\be
\frac{\partial^{2}}{\partial r_{im}\partial r_{jn}}\t_{k_{0}}^{(h)}\Big|_{\cc^\cs_{\rm int}}=O(\g).
\ee
for all $m, n$, since every occurrence of $K^3$ is enclosed between projectors $P_{k_0}$. Therefore, in the limit $\g\to0$, the $r$-$r$ block becomes block-diagonal.

If the coupling constant $\lambda_h$ is chosen such that the condensation spin $k_0/2$ is an integer, then $\varsigma_h^{k_0}=1$ becomes constant, which simplifies the Hessian matrix. In this situation, we can prove that the Hessian matrix is non-degenerate in the limit $\g \to 0$.

\begin{lemma}
When $k_0\in2\Z_+$, the $r$-$r$ block is nondegenerate as $\g\to0$ if $\G_{v,-}(v)$ is connected for all vertices $v$.
\end{lemma}

\begin{proof}
The $r$-$r$ block is nondegenerate if and only if every small blocks associated with a vertex is nondegenerate. Focus on a single vertex $v$ and define a quadratic form $Q(r)=\sum_{e,e'}\mathbf{M}_{e,e'}r_e r_{e'}$ with $r_e,r_{e'}\in\R$, where $e,e'$ are edges connecting to $v$ but  not connecting to the boundary. The matrix $\mathbf{M}$ has the diagonals $\mathbf{M}_{e,e} = -\sum_{h; e\subset \partial h} C_1$ and off diagonals $\mathbf{M}_{e,e'} = \sum_{h; e,e'\subset \partial h} C_1\cos(\theta_{e,e'})$. $k_0\in2\Z_+$ implies $\varsigma_h^{k_0}=1$. The quadratic form can be written as
\be
Q(r)=-\sum_{h,v\in \partial h} C_1 T_h(r),\qquad T_h(r)=r_{e_1(h)}^2+r_{e_2(h)}^2-2r_{e_1(h)} r_{e_2(h)}\cos(\theta_{e,e'}),\qquad C_1>0,
\ee
where $e_1(h),e_2(h)\subset\partial h$ are the pair of edges connecting the vertex $v$. For any $r_e\in\R$, we have $T_h(r)\geq 0$, and it implies $Q(r)\leq 0$. Moreover, $Q(r)= 0$ if and only if $T_{h}(r)=0$. By gauge fixing $r_{e_0}=0$ for one edge $e_0$, $T_{h}(r)=0$ implies $r_e=0$ for all $e$ sharing an internal face $h$ with $e_0$. By induction, $r_e=0$ propagates to all edges $e$ by the connectivity of $\G_{v,-}$. The fact that $Q(r)=0$ implies $r=0$ under gauge fixing indicates that the small block associated to $v$ is nondegenerate, so the $r$-$r$ block is nondegenerate.
\end{proof}


For the $t$-$t$ block of the Hessian, we have the diagonal entries:
\be
\sum_{h}C_0^{-1}\frac{\partial^{2}}{\partial t_{e}^{a}\partial t_{e}^{b}}\t_{k_{0}}^{(h)}\Big|_{C_{\rm int}^\cs}
&=&-\sum_{h;e\subset\partial h}C_0^{-1}\varsigma_h^{k_0} d_{k_{0}}\mathrm{Tr}_{k_{0}}\left[L^{a}L^{b}\right]
=-C_2\delta^{ab}\sum_{h;e\subset\partial h}\varsigma_h^{k_0},\nonumber\\
C_2&=& \frac{1}{12}d^2_{k_0}(d_{k_0}-1)\left(d_{k_0}+1\right)C_0^{-1},
\ee
and the off-diagonal entries:
\be
\sum_{h} C_0^{-1}\frac{\partial^{2}}{\partial t_{e}^{a}\partial t_{e'}^{b}}\t_{k_{0}}^{(h)}\Big|_{\cc_{\rm int}^\cs}
&=&-\sum_{h;e,e'\subset\partial h}s_{e}(h)s_{e'}(h) C_0^{-1} \varsigma_h^{k_0} d_{k_{0}}\mathrm{Tr}_{k_{0}}\left[L^{a}L^{b}\right]\nonumber\\
&=&-C_2\sum_{h;e,e'\subset\partial h}s_{e}(h)s_{e'}(h)\varsigma_h^{k_0},
\ee
where both $e$ and $e'$ are not along the maximal spanning tree $\ct$. The sign $s_{e}(h)=1$ if the orientation of the edge $e$ aligns with the orientation of $\partial h$, otherwise $s_{e}(h)=-1$. 


Unlike $r_{ij}$ which only associates to the vertex $i$, $t^a_{e=(i,j)}$ relates to both vertices $i$ and $j$. So, the $t$-$t$ block of the Hessian is not block diagonal.

\begin{lemma}

When $k_0\in2\Z_+$ and $K_-$ is simply connected, the $t$-$t$ block is non-degenerate.

\end{lemma}

\begin{proof}
We denote by $E_{\rm var}$ to be the set of edges not connecting to boundary and not belonging to $\ct$. We define the following quadratic form associated to the \(t\)-\(t\) block
\be
Q(t) = \sum_{e,e' \in E_{\text{var}}} \sum_{a,b} \mathbb{H}_{(e,a),(e',b)} \, t_e^a t_{e'}^b,\qquad t_e^a \in\R,
\ee
and set \(t_e = 0\) except for $e\in E_{\rm var}$. Let us fix, for each face \(h\), a base vertex \(v_h\) on its boundary. For each edge \(e \in \partial h\), let \(G_{h,e}\) be the background holonomy made by $\mathring H_e$ from \(v_h\) to the target of \(e\) along \(\partial h\) (following the face orientation). The holonomy $G_{h,e}$ reduces to $G_{h,e}=\pm\mathbb{I}$ due to $\mathring H_e=\pm\mathbb{I}$ by the gauge fixing on $\ct$. The quadratic form $Q(t)$ can be written as 
\be
Q(t)=\sum_h Q_h(t),\qquad 
Q_h(t) = -C_2 \sum_{a=1}^3 J^a_h J^a_h,\qquad J_h = \sum_{e \in \partial h} s_e(h) \mathrm{Ad}_{G_{h,e}} t_e \in\mathfrak{su}(2).
\ee
The condition $k_0\in2\Z_+$ implies $\varsigma_h^{k_0}=1$. \(Q(t) \leq 0\) by \(C_2 > 0\), so \(Q(t) = 0\) if and only if for every internal face \(h\),
\be
J_h=\sum_{e \in \partial h} s_e(h) \mathrm{Ad}_{G_{h,e}} t_e=\sum_{e \in \partial h} s_e(h)  t_e = 0.\label{Jh=0}
\ee
where $t_e=\sum_a t_e^a(-iL^a)$. This equation means that the linearized holonomy variation $\delta H_{\partial h}$ vanishes at the background $\{\mathring{H}_e\}_e$ for every $h$.
	
Let \(v_*\) be the root vertex of the spanning tree \(\mathcal{T}\). For any edge \(e \in E_{\text{var}}\), define the loop \(\ell_e\) based at \(v_*\) as follows: travel from \(v_*\) to the source of \(e\) along the unique path in \(\mathcal{T}\), traverse \(e\), and return from the target of \(e\) to \(v_0\) along the unique path in \(\mathcal{T}\). The linearized holonomy variation along \(\ell_e\) gives
\be
		H_{\ell_e}^{-1} \delta H_{\ell_e} = \mathrm{Ad}_{G_e} t_e,\label{AdGete}
\ee
where \(G_e \in \mathrm{SU}(2)\) is the background parallel transport from $v_*$ to the target of $e$ along the tree part of \(\ell_e\). The terms from edges in \(\mathcal{T}\) vanish because \(t_{e\in\mathcal{T}} = 0\). On the other hand, the loop \(\ell_e\) equals a product of loops \(\ell_{1}^{\epsilon_1} \cdots \ell_{k}^{\epsilon_k}\) with \(\epsilon_i = \pm 1\), where $\ell_i$ is the boundary of an internal face $h$, due to the trivial $\pi_1(\ck_-)$. Moreover, recall that $Q(t)=0$ implies $\delta H_{\partial h}=0$. Therefore, we have $\delta H_{\ell_e}=0$. Combining this and \eqref{AdGete} yields $t_e=0$, since \(\mathrm{Ad}_{G_e}\) is an isomorphism of the Lie algebra.
	
This argument holds for every \(e \in E_{\text{var}}\). Therefore, \(Q(t)=0\) implies \(t=0\). Consequently, the quadratic form \(Q\) is nondegenerate. This completes the proof that the \(t\)-\(t\) block is non-degenerate.

\end{proof}

Given that the Hessian matrix (for $k_0\in 2\Z$) is non-degenerate in the vanishing Barbero-Immirzi parameter limit: $\g\to0$, it is still non-degenerate for a generic value of $\g$, in particular, it is nondegenerate for small $\g$.

Although the nondegeneracy is only proven for $k_0\in 2\Z$, our main result in the following section is valid in general even for degenerate critical points: It is shown in \cite{arnold2012singularities} that for general oscillatory integral $\cj(A)=\int e^{A S(x)} \varphi(x) d^n x$ (of the same type as \eqref{AKexponen0}) with an analytic phase and a possibly degenerate critical point $x_c$, there is an asymptotic expansion
\be
\cj(A)=e^{A S(x_c)} \sum_\alpha \sum_{k=0}^{n-1} c_{\alpha, k}(\varphi) A^\alpha(\log A)^k,
\ee
and the leading term is equal to the value of the integrand at the critical point, multiplied by a non-zero constant
\be
\cj(A) \sim C e^{A S(x_c)}  \varphi(x_c) A^{-\a_0}(\log A)^{k_0},
\ee
for some $\a_0,k_0\geq 0$. Localizing the integrand on the critical points is the only ingredient that is necessary for deriving our main result.

\section{Boundary blocks}\label{Boundary blocks}

In this section, we discuss the boundary contribution $\int \prod_{(v,e_b)} \rmd g_{ve_b}\prod_b\o_b^{\rm bos}\big|_{\cs,\vec \sig}$ in the asymptotic formula \eqref{AKexponen}. For a boundary face $b$, whose vertices are labelled by $1,\cdots,m$ with $m\geq 2$,
\be
\t^{(b)}_k=d_k\tr_{(k,\rho)}\lt[\lt(P_kg_{1, e_b}^{-1}g_{12}P_{k}g_{21}^{-1}g_{23}P_{k}g_{32}^{-1}\cdots g_{m-1,m}P_{k}g_{m,m-1}^{-1} g_{m,e_b'}P_k\rt)H_{\fl(b)}\rt], 
\ee
where $g_{ij}$ are along edges in $E_{\rm int}$. Use the parametrization discussed in Section \ref{Some explicit computations}
\be
g_{ij}P_{k}g_{ji}^{-1}=\lt(u_{ij}e^{-ir_{ij}K^{3}}u_{ij}^{-1}\right)P_k\lt(x_iH_{ij}x_j^{-1}\rt)P_k\left(u_{ji}e^{ir_{ji}K^{3}}u_{ji}^{-1}\rt).\label{gPgParameter}
\ee
Restrict $\t^{(b)}_k$ on $\cc_{\rm int}^\cs$ and make the following change of variables, which leaves the Haar measure $\rmd g_{ve_b}$ invariant:
\be
g_{1, e_b}=x_1\wt{g}_{1, e_b},\qquad g_{m, e_b'}=x_m\wt{g}_{m, e_b'},
\ee
We obtain $\t^{(b)}_k$ on $\cc_{\rm int}^\cs$ :
\be
\t^{(b)}_k&\approx& d_k\tr_{(k,\rho)}\lt[\lt(P_k\wt{g}_{1, e_b}^{-1}P_{k}\mathring H_{12}\mathring H_{23}\cdots \mathring H_{m-1,m}P_{k} \wt{g}_{m,e_b'}P_k \rt)H_{\fl(b)}\rt]\nonumber\\
&\approx& \varsigma_b^{k} d_k\tr_{(k,\rho)}\lt[\lt(P_k\wt{g}_{1, e_b}^{-1}P_{k}\wt{g}_{m,e_b'}P_k\rt)H_{\fl(b)}\rt],\qquad \varsigma_b=\pm1,\label{plusCase}
\ee
since $H_e = \pm\mathbb{I}$ for every $e \in E_{\rm int}$ and $D^{k}(\varsigma_b \mathbb{I})=\varsigma_b ^{k} D^{k}( \mathbb{I})$ on $\ch_{k}$. Therefore, both $\t^{(b)}_k$ and consequently $\o^{\rm bos}_b$ take constant values on $\cc_{\rm int}^\cs$. 

The sign $\varsigma_b = \varsigma_b(\cs)$, which is assigned to each boundary face $b$, is determined by the bulk sign configuration $\cs = \{\varsigma_h\}$. Indeed, for each sector $\cs$, $\varsigma_b(\cs)$ is the product of the signs of the critical holonomies $H_e \approx \pm \mathbb{I}$ along the internal edge $e\in E_{\rm int}$ of the boundary face $b$, and the sign for each $H_e$ is uniquely determined by $\cs$ (see Lemma \ref{HeenotinFlenotinct}). 

In the case $m=1$, the face $b$ is a triangle whose sides are $e_b,e_b',\fl(b)$. Then both
\be
\t^{(b)}_k= d_k\tr_{(k,\rho)}\lt[\lt(P_kg_{v e_b}^{-1}g_{ve_b'}P_k\rt)H_{\fl(b)}\rt] \label{minusCase}
\ee
and $\o_b$ are independent of $\{g_{ve}\}_{e\in E_{\rm int}}$, so they are constant on $\cc_{\rm int}$. Note that $g_{v e_b}^{-1}g_{ve_b'}=\wt{g}_{v e_b}^{-1}\wt{g}_{ve_b'}$ is invariant under the gauge transformation at $v$.

We conclude that for a simply connected $\ck_-$, $\o^{\rm bos}_b$ on $\cc_{\rm int}$ does not depend on $\vec \sig$. Its dependence on $\cc_{\rm int}$ is only through the signs $\{\varsigma_b\}_b$. In addition, we extend the label $\varsigma_b$ to $\varsigma_b\in\{0,\pm1\}$ and assign $\varsigma_b=0$ to the case of $m=1$. Therefore,
\be
\o_b^{\rm bos}\big|_{\cs,\vec{\sig}}= \mathring\o_{b,\varsigma_b}\lt(\{\wt{g}_{ve_b}\},\vec{H}_{\fl(b)}\rt).
\ee
The explicit expresses are given below,
\be
\mathring{\o}_{b,\varsigma_b}&=&\sum_{p_b=1}^{\infty}\l_b^{p_b}\sum_{1\leq k_1\leq\cdots\leq k_{p_b}}^\infty\prod_{i=1}^{p_b}\mathring{\t}^{(b,\varsigma_b)}_{k_{i}}\Theta\left(A_{b}-\alpha_{p_b,\vec{k}}\right),\\
\mathring{\t}^{(b,\pm)}_{k}&=& (\pm1)^k d_k\tr_{(k,\rho)}\lt[\lt(P_k\wt{g}_{1, e_b}^{-1}P_{k}\wt{g}_{m,e_b'}P_k\rt)H_{\fl(b)}\rt],\qquad \mathring{\t}^{(b,0)}_{k}=d_k\tr_{(k,\rho)}\lt[\lt(P_k\wt{g}_{v e_b}^{-1}\wt{g}_{ve_b'}P_k\rt)H_{\fl(b)}\rt].\label{zetabkob1}
\ee
When collecting all boundary faces $b$, we denote 
\be
\bm{\varsigma} = \{\varsigma_b\}_b\in \Z_3^{|\cl|}
\ee 
where $|\cl|$ is the number of links of the boundary graph $\G = \partial\ck$, with an one-to-one correspondence between boundary faces $b$ and links in $\G$.

We now relabel all quantities, replacing the notation associated with the 2-complex by the corresponding notation for the boundary graph:
\be
\mathring{\o}_{b,\varsigma_b}\equiv \mathring{\o}_{\fl,\varsigma_\fl},\qquad \l_b\equiv \l_\fl,\qquad p_b\equiv p_\fl,\qquad A_b\equiv A_\fl,\qquad \varsigma_b\equiv \varsigma_\fl,\qquad \fs_b\equiv \fs_\fl\qquad \wt{g}_{ve_b}\equiv g_\fn,
\ee
by virtue of the bijection between each boundary face $b$ and a link $\fl$ of the boundary root graph $\G$, as well as the bijection between each $e_b$ and a node $\fn$ in $\G$. It is clear that the Haar measure is invariant $\rmd g_{ve_b}=\rmd \wt g_{ve_b}=\rmd g_{\fn}$.

Given any root graph $\G$ where the links and nodes are denoted by $\fl$ and $\fn$, we consider $\G$ as the boundary of a root 2-complex $\ck$. We define the vector space $\cv_\G$ spanned by the following set of functions $B_{\bm{\varsigma}}$, which equal the boundary contributions $\int \prod_{(v,e_b)} \rmd g_{ve_b}\prod_b\o_b^{\rm bos}\big|_{\cs,\vec \sig}$ in the asymptotic formula \eqref{AKexponen} 
\be
B_{\bm{\varsigma}}(\vec{H})&:=&\int \prod_{\fn\in\G} \rmd g_{\fn}\prod_{\fl\subset\G}\mathring{\o}_{\fl,{\varsigma}_\fl}\lt(\{g_\fn\},\vec{H}_{\fl}\rt),\qquad\vec{H}=\{\vec{H}_{\fl}\}_{\fl},
\ee
These functions depend on SU(2) holonomies $\vec{H}_\fl=\{{H}_{\fl}^{(i)}\}_{i}$ along the stacked links over $\fl$. Note that $B_{\bm{\varsigma} }(\vec{H})$ depends on the coupling constants $\{\l_\fl\}_\fl$, which are understood as fixed parameters.

The function $B_{\bm{\varsigma}}$ depends only on the data associated with $\ck\setminus\ck_-$, i.e. the degrees of freedom associated to $b$ and $e_b$. We can equivalently associate these data to the links and nodes of the boundary root graph $\G$. Accordingly, we refer to them collectively as the boundary data, consisting of:
\begin{itemize}

	\item $\varsigma_\fl\in\{0,\pm1\}$ for each boundary root link $\fl$. 
	

	\item The boundary SU(2) holonomies $\vec{H}=\{{H}_{\fl}^{(i)}\}_{\fl,i}$ along the stacked links on the boundary root graph $\G$.

	\item The boundary area cut-off $A_\fl$ and coupling constant $\l_\fl$.
\end{itemize}

We call the functions $B_{\bm{\varsigma} }$ boundary blocks and call $\cv_\G$ the space of boundary blocks. The boundary blocks are labelled by $\bm{\varsigma} $, so the dimension of $\cv_\G$ is $3^{|\cl|}$. Each of the functions $B_{\bm{\varsigma} }$ defines the following linear functional on the vector space spanned by spin-network states, $\mathbf{Cyl}\subset\ch_{\rm Kin}$, or its permutation invariant projection.
\be
B_{\bm{\varsigma} }[f]=\lim_{A_\fl\to\infty}\int \rmd \vec{H}\,B_{\bm{\varsigma} }(\vec{H})^* f(\vec{H}),
\ee
where $f \in \mathbf{Cyl}$ denotes a finite linear combination of spin-network states, and $\rmd \vec{H}$ represents the Haar measure on the space of SU(2) holonomies on the graph defined by the union of the support graphs of $B_{\bm{\varsigma} }$ and $f$. The finiteness of $B_{\bm{\varsigma} }[f]$ is addressed in Section \ref{Finiteness of boundary block}.

As we send the cut-offs $A_b=A_\fl$ in $\mathring{\o}_{\fl,\fs_\fl}$ to infinity, there are at most finitely many terms in $B_{\bm{\varsigma} }$ that can contribute to $B_{\bm{\varsigma} }[f]$, since $f$ is based on a finite graph (with finitely many links). Therefore, for any $f\in \mathbf{Cyl}$, there exists a sufficiently large (but finite) $A_\fl$ for each $\fl$ such that $B_{\bm{\varsigma} }[f]$ becomes independent of $A_\fl$ for all larger values. Conversely, this functional viewpoint also allows us to keep $A_b$ finite during the spinfoam calculations above, without affecting the result when acting on any state in $\mathbf{Cyl}$.

With the about definition, we find that the leading asymptotics of $\sa_\ck$ is represented by a finite linear combination over $\bm\varsigma$
\be
\sa_\ck
&=&\sum_{\bm\varsigma}b_{\bm\varsigma}\lt(\ck_-\rt) B_{\bm{\varsigma} }(\vec{H}) \lt[1+O(A^{-1})\rt].\label{AKexponen41}
\ee
The key point is that $\o^{\rm bos}_b$'s dependence on $\cc_{\rm int}$ is only through the signs $\bm{\varsigma}=\{\varsigma_b\}_b$. Consequently, all dependence of the amplitude on the interior $\ck_-$ of the root complex $\ck$ is now entirely reduces to a finite set of coefficients $b_{\bm\varsigma}(\ck_-)$. Infinitely many degrees of freedom in the bulk decouples in the limit $A\to\infty$. These coefficients are explicitly defined as sum over sign configurations $\cs$ compatible with the given boundary data $\bm\varsigma$,
\be
b_{\bm\varsigma}\lt(\ck_-\rt)&=&\frac{e^{\sum_h\b_{k_0}(\l_h)A_h}}{A^{D_{\rm int}/2}}\sum_{\cs,\bm\varsigma(\cs)=\bm\varsigma}\cf_\cs\lt(\vec A,\vec \l\rt)\int\limits_{\cc_{\rm int}^\cs}\frac{\rmd \mu(\vec \sig)}{\sqrt{\det\lt(-\mathbb{H}_\cs(\vec\sig)/(2\pi)\rt)}}.
\ee

Summing over infinitely many 2-complexes $\ck$ which has the same boundary $\G=\partial\ck$, the complete amplitude $\sa$ is still a linear combination over both $\bm\varsigma$ and $\bm\fs$:
\be
\sa=\sum_{\ck,\partial\ck=\G}c_\ck{\sa}_\ck=\sum_{\bm\varsigma}b_{\bm\varsigma} B_{\bm\varsigma}\lt[1+O(A^{-1})\rt],\qquad b_{\bm\varsigma}=\sum_{\substack{\ck,\partial\ck=\G,\\ \bm\varsigma(\ck)=\bm\varsigma}}b_{\bm\varsigma}(\ck_-)\, c_\ck \label{completeAmpliResult}
\ee
where $\{c_\ck\}_\ck$ are infinitely many arbitrary coefficients. $\bm\varsigma(\ck)=\bm\varsigma$ means that the boundary faces of $\ck$ have the same type as $\bm\varsigma$. The sum on the left-hand side converges if and only if all $3^{| \cl|}$ coefficients $b_{\bm\varsigma}$ on the right-hand side converge as infinite linear combinations of $c_\ck$.

We make several important observations regarding this result:

\begin{itemize}

\item In the limit $A \to \infty$, the infinitely many ambiguities contained in the coefficients $c_\ck$ and $\l_h$ become only finitely many boundary coefficients $b_{\bm\varsigma}$. The bulk ambiguities that disappear in this limit decouple from the theory. Since these bulk ambiguities are associated with the triangulation dependence of the spinfoam, their decoupling implies that the resulting theory is triangulation independent.

\item For a given boundary root graph $\G$, the boundary blocks $B_{\bm\varsigma}\in\mathbf{Cyl}^*$ span a $3^{|\cl(\G)|}$-dimensional vector space $\cv_\G$. In the limit $A\to\infty$, the complete amplitude $\sa$ is represented by a vector in this space.

\item As $A \to \infty$, the theory reaches a fixed point where all bulk degrees of freedom become either frozen or decoupled. As a signal of fix point, the theory in the limit becomes topological and invariant under changing bulk triangulation, with all physical degrees of freedom localized on the boundary.

\item Two approaches have been proposed for obtaining the continuum limit of spinfoam models: summing over complexes versus refining them \cite{Rovelli:2010qx,Asante:2022dnj,Oriti:2009wn}. In our result, the stack amplitude $\sa_\ck$ is already free from bulk ambiguities. As a result, taking the continuum limit is trivial, regardless of whether it is performed via summing or refining, and the two methods are manifestly equivalent.

\item Our results demonstrate that at the UV fixed point $A\to\infty$, the triangulation-independent theory is well defined. As $A^{-1}$ increases away from zero, the theory flows away from the fixed point: bulk degrees of freedom begin to propagate and the system flows toward the IR regime.

\item As noted above Lemma \ref{lemmaBoundUc} and in the definition of $S$ in \eqref{actionS=}, the limit $A\to\infty$ can equivalently be characterized by $A_h/E_{k_0} \to \infty$. In our opinion, the spin cut-off $A$ should be understood as a large but finite constant in the full spinfoam theory. Intuitively, $A$ is expected to relate to the cosmological constant. Consequently, a significant suppression of fluctuations (as described by the stationary phase expansion) only occurs in the regime where $A_h/E_{k_0}$ is large, which corresponds to small values of the condensation spin $k_0$. Therefore, the regime of finite $\l_h$, where $k_0$ is small, is the vicinity of the fix point. The topological theory is at the leading order in this vicinity. The small $k_0$ indicates that the fixed point is in the UV regime, because $k_0$ is proportional to the expectation value of the quantum area. 

\end{itemize}

\section{Finiteness of boundary block}\label{Finiteness of boundary block}

As a Hilbert space, the carrier space of the $\Slc$ principal series unitary irrep, ${\cal H}_{(k,\rho)}$, is the subspace of $L^{2}(\Su)$ made of functions obeying the covariance condition
\begin{equation}
F(e^{i\phi\bm{\sigma}^3} u)=\mathrm{e}^{i k\phi}F(u),\qquad u\in \Su\label{covariance}
\end{equation}
An orthogonal basis of ${\cal H}_{(k,\rho)}$ is made of a subset of Wigner matrices $D_{k/2,m}^{j}(u)=\langle u\mid(k,\rho),j,m\rangle$. To define the action of $\Slc$, we use the Iwasawa decomposition: any $g\in \Slc$ can be written as $g= b u$ with $u\in SU(2)$ and $b\in B$, where $B$ is the subgroup of upper triangular matrices
\begin{equation}
B=\left\{b=\begin{pmatrix}\lambda ^{-1}&\mu\cr0&\lambda\end{pmatrix}, \lambda\in{\Bbb C}^{\times},\mu\in{\Bbb C}\right\}.
\end{equation}
Moreover, since $u=b^{-1}g$ is unitary, it is useful to notice that $b^{-1}g=b^{\dagger}(g^{-1})^{\dagger}$. To define the $\Slc$ action on ${\cal H}_{(k,\rho)}$, we use this decomposition for $ug$ for any $u\in\Su$ and $g\in \Slc$:
\begin{equation}
ug=b_{g}(u)u_{g}(u),\qquad b_{g}(u)\in B,\quad u_{g}(u)\in \Su
\end{equation}
Then, the action of $g$ on $F\in{\cal H}_{(k,\rho)}$ reads
\begin{equation}
g\cdot F(u)=\lt[\lambda_{g}(u)\rt]^{-k/2+\mathrm{i}\rho/2-1}\lt[\lambda_{g}(u)^*\rt]^{k/2+\mathrm{i}\rho/2-1}F\lt(u_{g}(u)\rt).\label{SL2Caction}
\end{equation}
where $\lambda_{g}(u)$ is the lower right entry of the matrix $b_g(u)$. Note that the Iwasawa decomposition has the gauge freedom $b_{g}(u)\to b_{g}(u)e^{-i\phi\bm{\sigma}^{3}},\ u_{g}(u)\to e^{i\phi\bm{\sigma}^{3}}u_{g}(u)$. The formula \eqref{SL2Caction} is gauge invariant thanks to the covariance condition \eqref{covariance}. We can use this gauge freedom to set $\l_g(u)\in\R$. In this case, $\lambda_{g}^{2}(h)$ is the upper left corner of the matrix  $(b_{g}(u)^{-1})^{\dagger}b_{g}^{-1}(u)=u(g^{-1})^{\dagger}g^{-1}u^{\dagger}$, i.e.
\begin{equation}
\lambda_{g}(u)=\Big[\langle\textstyle{\frac{1}{2},\frac{1}{2}}|u(g^{-1})^{\dagger}g^{-1}u^{\dagger}|\textstyle{\frac{1}{2},\frac{1}{2}}\rangle \Big]^{\frac{1}{2}}
\label{lambda}
\end{equation}
with $|j,m\rangle$ the standard basis of the spin $j$ representation of SU(2).

The above lemma shows explicitly that the matrix coefficients of $g$ suppress as $t e^{-2t}$ as $t\to\infty$.

\begin{lemma}\label{LemmaFgFxi}
For any two states $F_1,F_2\in \ch_{(k,\rho)}$ and any $g\in\Slc$, 
\be
\lt|\lag F_1\mid g\mid F_2\rag\rt|\leq \xi(g)\|F_1\|_\infty \|F_2\|_\infty
\ee
where $\|F\|_\infty=\sup_u |F|$.  , the function $\xi(g)$ is given by
\be
\xi(g)=\xi(g^{-1})=\frac{2 t}{\sinh (2 t)}.
\ee
\end{lemma}

\begin{proof}
The explicit expression of $\l_g(u)$ can be derived by the decomposition $g=R^{-1}e^{-t\bm\sig_3}R'$: 
\be
\lambda_{g}(u)^2=e^{2 t}|\alpha|^2 +e^{-2 t}|\beta|^2,\qquad Ru^\dagger=\left(
	\begin{array}{cc}
	 \alpha  & \beta  \\
	 -\beta ^* & \alpha ^* \\
	\end{array}
	\right)\in\Su.
\ee
For any two states $F_1,F_2\in \ch_{(k,\rho)}$, the matrix coefficient of the unitary representation of $g\in\Slc$ is given by
\be
\lag F_1\mid g\mid F_2\rag=\int_{\Su} \rmd u \, \lambda_{g}(u)^{\mathrm{i}\rho-2}F_1^*(u)\, F_2\lt(u_{g}(u)\rt)
\ee
Then we estimate
\be
\lt|\lag F_1\mid g\mid F_2\rag\rt|\leq \|F_1\|_\infty \|F_2\|_\infty \int_{\Su} \frac{\rmd U}{e^{2 t}|\alpha|^2 +e^{-2 t}|\beta|^2}=\frac{2 t}{\sinh (2 t)}\|F_1\|_\infty \|F_2\|_\infty 
\ee
where $U=Ru^\dagger$ and $\rmd U$ is the Haar measure.

\end{proof}

Let $\mathcal{L}$ and $\mathcal{N}$ denote the sets of links and nodes, respectively, of the root graph $\G$ underlying $B_{\boldsymbol{\varsigma}}$. For any spin-network state $f \in \mathbf{Cyl}$, $f$ is bounded on the space of holonomies; define $C_f := \sup_{\vec{H}} |f(\vec{H})|$. Then,
\be 
|B_{\boldsymbol{\varsigma}}[f]| \le C_f \int\rmd\vec{H}  \int_{\Slc^{|\mathcal{N}|}} \rmd\vec{g} \prod_{\fl \in \mathcal{L}} \lt|\mathring{\omega}_{\fl,{\varsigma}_\fl} \right| .\label{Bbound1}
\ee

Recall that for each $\fl$, there is a sufficiently large (but finite) cutoff $A_\fl$ so that $B_{\bm{\varsigma} }[f]$ becomes independent of $A_\fl$ for any $f \in \mathbf{Cyl}$ once $A_\fl$ exceeds this threshold. The link amplitude $\mathring{\omega}_{\fl, \varsigma_\fl}$ involves a sum over multiplicities $p$ and spins $k_i$ on the stacked links. The Heaviside function $\Theta(A_\fl - \a_{p_\fl, \vec{k}})$, with a finite $A_\fl$, truncate this sum by restricting to a finite range of $p_\fl$ and $\vec{k}$. We denote by $p_{\fl,\mathrm{max}}$ and $\vec{k}_{\mathrm{max}}$ the maximally allowed values.

For the case of $\varsigma_\fl=\pm1$, each term in the sum is proportional to 
\be
\prod_{i=1}^{p_\fl}\mathrm{Tr}_{(k_i,\rho_i)}\lt(P_{k_i} g_{s(\fl)}^{-1} P_{k_i} g_{t(\fl)} P_{k_i} H^{(i)}_\fl\rt)=\prod_{i=1}^{p_\fl}\sum_{m_i,n_i,q_i}D^{(k_i,\rho_i)}_{k_im_i,k_in_i}\lt(g_{s(\fl)}^{-1}\rt)D^{(k_i,\rho_i)}_{k_in_i,k_iq_i}\lt(g_{t(\fl)}\rt)D^{k_i/2}_{q_i,m_i}\lt( H^{(i)}_\fl\rt)
\ee
where the range of $m_i,n_i,q_i$ is finite. Lemma \ref{LemmaFgFxi} implies that the matrix coefficients $D^{(k_i,\rho_i)}_{k_im_i,k_in_i}(g)$ satisfy the bound 
\be
\lt|D^{(k_i,\rho_i)}_{k_im_i,k_in_i}(g)\rt|\leq d_{k_i} c_{k_i}^2 \xi(g),\qquad c_{k_i}=\sup_{\substack{u\in \Su\\ m,n=-k_i/2,\cdots,k_i/2}} \lt|D^{k_i/2}_{m,n}(u)\rt|
\ee
because the orthonormal basis in $\ch_{(k,\rho)}$ used here is $\sqrt{d_{k_i}} D^{k_i/2}_{k_i/2, m_i}(u)$, $m_i=-k_i/2,\cdots,k_i/2$.
Thus, we obtain
\be
|\mathring{\omega}_{\fl,\pm}| \le \sum_{p_\fl=1}^{p_{\fl,\mathrm{max}}} C_{p_\fl} \xi\lt(g_{s(\fl)}\rt) ^{p_\fl}\xi\lt(g_{t(\fl)}\rt) ^{p_\fl}\leq C_\fl \xi\lt(g_{s(\fl)}\rt) \xi\lt(g_{t(\fl)}\rt) ,\qquad C_{p_\fl}= |\l_\fl|^{p_\fl}\sum_{1\leq k_1\leq\cdots\leq k_{p_b}}^{\vec{k}_{\mathrm{max}}}\prod_{i=1}^{p_\fl}d_{k_i}^6 c_{k_i}^5.
\ee
for some $C_\fl>0$. This second inequality holds because $\xi(g)$ decays exponentially as $\|g\| \to \infty$. As a result, the estimate factorizes, and the dependence on the source and target nodes in $|\mathring{\omega}_{\fl,\pm}|$ becomes decoupled.

The estimate for the case of $\varsigma_\fl = 0$ can be done similarly. But since each term in the sum is proportional to $\mathrm{Tr}(P_k g_{s(l)}^{-1} g_{t(l)} P_k H_l)$ which is a matrix element of the group element $g_{s(l)}^{-1} g_{t(l)}$. The result is
\be 
|\mathring{\omega}_{\fl,0}| \le C_\fl \xi\lt(g_{s(\fl)}^{-1} g_{t(\fl)}\rt) .
\ee
for some $C_\fl>0$. 

The graph $\G$ is the boundary of a simply connected root complex $\ck$. Let us first consider the special case that $\ck$ consists of a single vertex. It follows that all links in $\G$ are of type $\varsigma_\fl = 0$. In this scenario, the finiteness of $|B_{\bm{\varsigma} }[f]|$ is equivalent to the finiteness of the vertex amplitude of $v$. By $\Slc$ gauge fixing at the vertex $v$ and the assumption that the intersection graph $\G_v=\G$ is 3-connected, we conclude that $|B_{\bm{\varsigma} }[f]|$ is finite \cite{Kaminski:2010qb}.

Suppose $\ck$ contains more than one vertex. Consider a vertex $v$ that is connected to the boundary graph $\G = \partial\ck$ via an edge $e_b$. This vertex $v$ may be connected to several boundary nodes $\fn \in \G$; let $\cn_v$ denote the set of these nodes. The links $\fl$ connecting nodes in $\cn_v$ are of type $\varsigma_\fl = 0$. Given that $\ck$ is connected, there exists at least one internal edge $e \in E_{\rm int}$ linking $v$ to a different vertex $v'$. This edge $e$ is represented as a node in the intersection graph $\G_v$. Let $\Fn_v$ be the set of nodes in $\G_v$ that correspond to internal edges $e \in E_{\rm int}$. Thus, the set of nodes in $\G_v$ is the disjoint union $\Fn_v \cup \cn_v$. Since $\G_v$ is 3-connected, there must exist at least three links connecting nodes in $\Fn_v$ with nodes in $\cn_v$.

\begin{lemma}\label{LemmaLinkBijection}
There is a bijection from links in $\G_v$ that connect nodes in $\cn_v$ to nodes in $\Fn_v$ and links $\fl \subset \G$ connecting nodes in $\cn_v$ to other boundary nodes in $\cn\setminus \cn_v$. These links has $\varsigma_\fl = \pm1$. 
\end{lemma}

\begin{proof}
The links in $\G_v$ connecting $\cn_v$ to $\Fn_v$ one-to one correspond to boundary faces $b\subset\ck$, each of which is bounded by an internal edge $e \in E_{\rm int}$ and an edge $e_b$ connecting to the boundary graph. There is one-to-one correspondence between such faces and boundary links $\fl = b \cap \G$ attached to the node $\fn = e_b \cap \G$. Since each of these boundary faces $b$ includes at least one internal edge, it has more than one vertex; hence, the corresponding link $\fl$ is of type $\varsigma_\fl = \pm1$.
\end{proof}

This result allows us to partition the set of nodes $\cn$ in $\G$ into disjoint subsets $\cn_v$, such that $\cup_v \cn_v = \cn$. All links connecting nodes within each subset $\cn_v$ have type $\varsigma_\fl = 0$; let $\cl_v$ denote the set of these links. There are at least three links of type $\varsigma_\fl = \pm1$ connecting nodes in $\cn_v$ to nodes outside $\cn_v$, and we denote this set of links by $\Fl'_v$.

Given the bounds for $|\mathring{\omega}_{\fl,\varsigma_\fl}|$, the integral $ \int \rmd\vec{g} \prod_{\fl \in \mathcal{L}} |\mathring{\omega}_{\fl,{\varsigma}_\fl}| $ can be organized as a product over $\cn_v$.
\be
\int_{\Slc^{|\cn|}}\rmd\vec{g} \prod_{\fl \in \mathcal{L}} |\mathring{\omega}_{\fl,\varsigma_\fl} |\leq C \prod_{\cn_v}\ci_{\cn_v},\qquad \ci_{\cn_v}=\int \prod_{\fn\in\cn_v}\rmd g_{\fn}\prod_{\fl\in\cl_v}\xi\lt(g_{s(\fl)}^{-1} g_{t(\fl)}\rt)\prod_{\fl\in\Fl_v'}\xi\lt(g_{s/t(\fl)}\rt).\label{Bbound2}
\ee
where $s/t(\fl)$ indicates either the source or target of $\fl$.

\begin{theorem}
$|B_{\boldsymbol{\varsigma}}[f]| $ is finite. Therefore, $B_{\boldsymbol{\varsigma}}$ is a linear functional on $\mathbf{Cyl}$.
\end{theorem}

\begin{proof}
The finiteness has been proven for the special case that $\ck$ contains only a single vertex. We focus on the general scenario that $\ck$ has more than one vertices. By \eqref{Bbound1} and \eqref{Bbound2}, it suffices to prove the finiteness of $\ci_{\cn_v}$.

Focus on the intersection graph $\G_v$, we denote by $\Fl_v$ the set of links connecting between nodes in $\Fn_v$. We denote by $\Fl_v'$ the set of links connecting $\Fn_v$ and $\cn_v$, due to the bijection in Lemma \ref{LemmaLinkBijection}. The set of links in $\G_v$ is a disjoint union of $\cl_v$, $\Fl_v$ and $\Fl_v'$. 

The convergent integral to bound the vertex amplitude of $v$ in \cite{Kaminski:2010qb} is 
\be
I_v=\int \prod_{\fn\in\cn_v\cup \Fn_v\setminus\{\fn_0\}}\rmd g_{\fn}\,\prod_{\fl\in\cl_v\cup \Fl_v\cup\Fl_v'}\xi\lt(g_{s(\fl)}^{-1} g_{t(\fl)}\rt)\Big|_{g_{\fn_0=1}},
\ee
For convenience, the gauge fixing $g_{\fn_0}=1$ is imposed at a node $\fn_0\in\Fn_v$. We denote by $F$ and $G$ the Integrands of $I_v$ and $\ci_{\cn_v}$ respectively. It is clear that $G$ is the restriction of $F$ by imposing $g_\fn=1$ for all $\fn\in \Fn_v\setminus\{\fn_0\}$. It is not difficult to see that $\xi(g_2)\leq c_\eps \xi(g_1^{-1}g_2) $ for some $c_\eps>0$ and $g_1$ in any closed ball $B_\eps$ centered at $g_1=1$. Thus, $G\leq C_\eps F$ for some $C_\eps$ and $g_\fn$ in any closed ball $B_\eps$ centered at $g_\fn=1$, $ \fn\in \Fn_v\setminus\{\fn_0\}$. Therefore,
\be
\ci_{\cn_v}=\frac{1}{\cv_\eps} \int \prod_{\fn\in\cn_v}\rmd g_{\fn} \prod_{\fn\in\Fn_v\setminus\{\fn_0\}}\int_{B_\eps}\rmd g_{\fn}\, G
\leq \frac{C_\eps}{\cv_\eps} \int \prod_{\fn\in\cn_v}\rmd g_{\fn} \prod_{\fn\in\Fn_v\setminus\{\fn_0\}}\int_{B_\eps}\rmd g_{\fn}\, F\leq \frac{C_\eps}{\cv_\eps} I_v<\infty.
\ee
where $\cv_\eps$ is the volume of $B_\eps^{|\Fn_v\setminus\{\fn_0\}|}$.

\end{proof}



\appendix


\section{Graph of root complex}\label{vertexEdgeInjection}

For any connected root complex $\ck$, its internal 1-skeleton is a connected graph $G$ that contain all vertices and all edges in $E_{\rm int}$ (the set of edges not connecting to the boundary). Assuming $\ck$ contain at least one internal face, the graph $G$ contains at least one cycle. 

\begin{lemma}\label{lemmaInjectiveVE}
Let $G = (V, E)$ be a connected graph containing at least one cycle ($V,E$ denote the sets of vertices and edges). There exists an injective map $f: V \to E$ such that for every vertex $v \in V$, $f(v)$ is an edge incident to $v$.
\end{lemma}

\begin{proof}

Since $G$ is connected, it contains a spanning tree $T = (V, E_T)$, where $E_T \subset E$. A spanning tree on $|V|$ vertices has exactly $|V| - 1$ edges. Since $G$ contains at least one cycle, the set of edges $E \setminus E_T$ is non-empty. Let $e^* = \{u, v\}$ be an edge in $E \setminus E_T$.
Consider the subgraph $H = (V, E_H)$ where $E_H = E_T \cup \{e^*\}$.
The subgraph $H$ has $|V|$ vertices and $|V|$ edges (so an injection from $V$ to $E_H\subset E$ becomes possible). Since adding an edge to a spanning tree creates exactly one cycle, $H$ contains exactly one cycle, which we denote by $C$.

We define an orientation for every edge in $E_H$ to construct a directed graph $\vec{H}$: (1) Let the vertices of the cycle $C$ be $c_1, c_2, \dots, c_k$ in order. We orient the edges of the cycle cyclically: $c_1 \to c_2 \to \dots \to c_k \to c_1$. (2) Removing the edges of $C$ from $H$ leaves a forest where each connected component is a tree rooted at a vertex in $C$. For any vertex $x \in V \setminus V(C)$, there is a unique path in $H$ from $x$ to the cycle $C$. We orient all edges on this path in the direction towards $C$.

For any vertex $v$ in a directed graph, the out-degree $d_{out}(v)$ is the number of outgoing edges from $v$. In $\vec{H}$, we show that $d_{out}(v)=1$ for every vertex $v \in V$: 
\begin{itemize}
\item $v$ is on the cycle $C$:
    The vertex $v$ has exactly two incident edges in $C$. In the cyclic orientation, exactly one of these edges is directed away from $v$. Any other edges incident to $v$ in $H$ belong to the trees attached to $C$ and are directed towards $v$ (since $v$ is the "root" for those branches). Thus, $d_{out}(v) = 1$.
\item  $v$ is not on the cycle $C$:
    Let $P$ be the unique path from $v$ to $C$. Let $e$ be the first edge of this path, connecting $v$ to some neighbor $w$. By our construction, $e$ is directed from $v$ to $w$. Any other edge incident to $v$ in $H$ would be part of a path starting further away from $C$ and passing through $v$, and thus would be directed towards $v$. Therefore, $v$ has exactly one outgoing edge, so $d_{out}(v) = 1$.
\end{itemize}

Given $\vec{H}$ specified above, we define the map $f: V \to E$ as follows:
For each vertex $v \in V$, let $f(v)$ be the unique edge in $E_H$ that is oriented away from $v$ in $\vec{H}$.
Formally, if the outgoing edge from $v$ is directed $v \to w$, then $f(v) = \{v, w\}$.
By definition, $f(v)$ is incident to $v$.

To prove $f$ is injective, we must show that if $x \neq y$, then $f(x) \neq f(y)$.
Suppose for the sake of contradiction that $f(x) = f(y) = \epsilon$ for distinct vertices $x, y$. Given $x \neq y$, it must be that $\{x, y\} = \epsilon$.
This implies that in our orientation $\vec{H}$, $x$ orients the edge $\epsilon$ as $x \to y$ (since $f(x)=\epsilon$), whereas $y$ orients the edge $\epsilon$ as $y \to x$ (since $f(y)=\epsilon$). This would imply that the edge $\epsilon$ is directed both ways. But an edge has only one starting vertex in a directed graph \footnote{Even if $C$ is a cycle of length 2 (two parallel edges between $x$ and $y$), the edges are distinct. Let the edges be $e_1$ and $e_2$. The cycle orientation would be $x \xrightarrow{e_1} y \xrightarrow{e_2} x$. In this case, $f(x) = e_1$ and $f(y) = e_2$. Since $e_1 \neq e_2$, we have $f(x) \neq f(y)$.}. Thus, it is impossible for $f(x) = f(y)$ with $x\neq y$. The map $f$ is injective.

\end{proof}

\section{The partition function $\Xi[s_h]$}\label{meromorphic function}

Let $\O_N$ be the set of sequences $\mathbf{n}=\{n_k\}_{k\in\Z_+}$, $n_k=0,1,2,\cdots$, such that $n_k = 0$ for all $k > N$. Let $\O_{\rm fin}=\cup_{N=1}^\infty \O_N$ be the set of sequences $\mathbf{n}$ with finite support ($\mathbf{n}=\{n_k\}_{k\in\Z_+}$ such that $n_k \neq 0$ for finitely many $k$).

\begin{lemma}\label{lemmaA1}
Given a sequence $\{a_k\}_{k=1}^\infty$ satisfying $0\leq a_k<1$ and $\sum_{k=1}^\infty a_k< \infty$, the following relation holds
\be
\sum_{\mathbf{n}\in\O_{\rm fin}}\prod_{k=1}^\infty a_k^{n_k}=\prod_{k=1}^\infty\frac{1}{1-a_k}< \infty.
\ee
\end{lemma}

\begin{proof}
Consider the infinite product $ Q = \prod_{k=1}^\infty ({1 - a_k})^{-1} $ and take the logarithm $ \ln Q = \sum_{k=1}^\infty [-\ln(1 - a_k) ]$. Since $\sum_{k=1}^\infty a_k$ converges, $a_k \to 0$ as $k \to \infty$. For sufficiently large $k_s$, $a_k \leq \frac{1}{2}$ for $k\geq k_s$. We have $ 0\leq -\ln(1 - a_k) \leq 2a_k $ for $0\leq a_k \leq \frac{1}{2}$. It implies $\sum_{k\geq k_s}[ -\ln(1 - a_k)]\leq 2\sum_{k\geq k_s}a_k<\infty$, so the infinite product $Q$ converges.

Let $Q_N$ be the $N$-th partial product $ Q_N = \prod_{k=1}^N \frac{1}{1 - a_k} $. It is clear that $\lim_{N\to\infty }Q_N=Q$. Since $0\leq a_k < 1$, we can expand each factor as a convergent geometric series: 
\be
Q_N = \prod_{k=1}^N \left( \sum_{n_k=0}^\infty a_k^{n_k} \right)  = \sum_{\mathbf{n} \in \O_N} \prod_{k=1}^\infty a_k^{n_k}.
\ee 
Then take the limit,
\be
Q=\lim_{N\to\infty }Q_N=\sum_{\mathbf{n} \in \O_{\rm fin}} \prod_{k=1}^\infty a_k^{n_k}
\ee
\end{proof}

It is sufficient to only consider the sum over $\mathbf{n}\in\O_{\rm fin}$. Indeed,
let $\O$ be the set of all sequences $\mathbf{n}$. For any $\mathbf{n}\in \O\setminus\O_{\rm fin}$ (sequences with infinite support), $\prod_{k=1}^\infty a_k^{n_k}=0$ because $a_k \to 0$ as $k \to \infty$ ($a_k\leq 1/2$ for $k\geq k_0$ and $\lim_{m\to\infty}(1/2)^m=0$). As a result, we may formally write 
\be
Q=\sum_{\mathbf{n} \in \O} \prod_{k=1}^\infty a_k^{n_k},
\ee
despite that $\O$ is uncountable.

\begin{lemma}\label{lemmaA2}
Given a sequence $\{a_k\}_{k=1}^\infty$, $a_k\in\C$, satisfying $|a_k|<1$ and $\sum_{k=1}^\infty |a_k|< \infty$, the series $\sum_{\mathbf{n}\in\O_{\rm fin}}\prod_{k=1}^\infty a_k^{n_k}$ converges absolutely, and the following relation holds
\be
\sum_{\mathbf{n}\in\O_{\rm fin}}\prod_{k=1}^\infty a_k^{n_k}=\prod_{k=1}^\infty\frac{1}{1-a_k}< \infty.
\ee
\end{lemma}

\begin{proof}
The right-hand side convergence due to $|1-a_k|^{-1}\leq (1-|a_k|)^{-1}$. The rest of the proof is a trivial generalization from the above.

\end{proof}

\begin{theorem}
(1) Assume $E_k\sim O(k^N)$, $N\in\Z_+$, for large $k$ and $dE_k/dk\geq B$ for some $B>0$. The sum in
\be
\Xi(s_h) = \sum_{{\bf n}\in\O_{\rm fin}}
\prod_{k=1}^\infty
\left[ \lambda_h \tau_k^{(h)}(g_h) e^{-s_h E_k} \right]^{n_k}-1
\ee
converges absolutely for sufficiently large $\re(s_h)>0$.

(2) At a sufficiently large $\re(s_h)>0$, $\Xi(s_h)$ can be expressed as
\be
\Xi(s_h) =
\prod_{k=1}^\infty
\frac{1}{1-\lambda_h \tau_k^{(h)}(g_h) e^{-s_h E_k} }-1. \label{X1shA6}
\ee
\end{theorem}

\begin{proof}
(1) Let $a_k=|\lambda_h \tau_k^{(h)}(g_h) e^{-s_h E_k}|$, we only need to show $a_k<1$ and $\sum_{k=1}^\infty a_k< \infty$ for some $s_h$ with a large real part, by Lemma \ref{lemmaA1}. We denote by $x=\re(s_h)$. It is proven in \cite{spinfoamstack} that $|\tau_k^{(h)}(g_h)|\leq d_k^2$, so $a_k\leq |\lambda_h| d_k^2 e^{-x E_k}$. Given any $x>0$, the derivative of $ d_k^2 e^{-x E_k}$ in $k$ is
\be
e^{-x E_k}(k+1)\lt[2-(k+1) x\frac{dE_k}{dk}\rt]\leq e^{-x E_k}(k+1)\lt[2-(k+1) x B\rt]. \label{derivA5}
\ee
For all $x>B^{-1}$, Eq.\eqref{derivA5} is strictly negative for $k\in\Z_+$, so $|\lambda_h| d_k^2 e^{-x E_k}\leq 4|\lambda_h| e^{-x E_{1}}$. For all $x>\mathrm{Max}\lt\{B^{-1},E_{1}^{-1}\ln(4|\lambda_h|)\rt\}$, $a_k<1$ is satisfied, and $\sum_{k=1}^\infty a_k \leq |\l_h|\sum_{k=1}^\infty d_k^2 e^{-x E_k}<\infty$ since $E_k\sim O(k^N)$, $N\in\Z_+$, for large $k$. 

(2) This follows from Lemma \ref{lemmaA2}.
\end{proof}

\begin{theorem}
When analytic continuing the expression of $\Xi[s_h]$ in \eqref{X1shA6}, the resulting function is meromorphic on the right half plane $\re(s_h)>0$.
\end{theorem}

\begin{proof}
Consider the following function
\be 
\Xi(s_h) +1 = \prod_{k=1}^{\infty} \frac{1}{1 - \mathcal{T}_k(s_h)},\qquad \mathcal{T}_k(s_h) = \lambda_h \tau^{(h)}_k(g_h) e^{-s_hE_h}. 
\ee
To prove that $\Xi(s_h)$ is meromorphic, we must show that the denominator, the infinite product $D(s_h) = \prod_{k=1}^\infty (1 - \mathcal{T}_k(s_h))$, defines a holomorphic function. The function $ 1/D(s_h)$ will then be meromorphic, with poles occurring at the zeros of $D(s_h)$.

An infinite product $\prod_{k=1}^\infty (1 - \ct_k(s))$ converges to a holomorphic function on a domain $D$ if the series $\sum_{k=1}^\infty |\ct_k(s)|$ converges uniformly on compact subsets of $D$ (Theorem 15.6 in \cite{rudin1987real}). Let $s_h = x + iy$. We consider the domain $D = \{ s_h \in \mathbb{C} \mid x > 0 \}$ (the right half-plane).
For large $k$, $E_k \geq C k$ for some $C>0$.
We estimate 
$$ |\mathcal{T}_k(s_h)| \le |\lambda_h | d_k^2 e^{-C x k} $$
Consider any compact subset $K \subset D$. There exists a minimum real part $x_0 > 0$ such that for all $s_h \in K$, $\operatorname{Re}(s_h) \ge x_0$.
The series $\sum_{k=1}^\infty d_k^2 e^{-C x_0 k}$ is convergent.
By the Weierstrass M-test, the series $\sum_{k=1}^\infty \mathcal{T}_k(s_h)$ converges absolutely and uniformly on $K$.

\end{proof}

\section{Bound of $|\tau_k^{(h)}(g_h)|$}\label{Bound of tau k}

Given the principal series unitary irrep of $\Slc$ carried by the Hilbert space $\ch_{(k,\rho)}=\oplus_{m=k}^\infty\ch_m$ with $k\in\Z_+$ and $\rho\in R$ (in the direct sum $m=k,k+2,\cdots$), we denote by $D^{(k,\rho)}(g)$ the unitary operator representing $g\in\Slc$ on $\ch_{(k,\rho)}$. The operator $P_k$ is the orthogonal projection from $\ch_{(k,\rho)} $ to the SU(2) irrep $\ch_k$.

\begin{lemma}
	$D^{(k,\rho)}(g)P_k |\,\ff\,\rangle\in \ch_k$ for all $\ff\in\ch_k$ if an only if $g\in\Su$.
\end{lemma} 

\begin{proof} We only prove the ``only if'' direction, while the other direction is trivial. Suppose $g\in\Slc$ but not in $\Su$. We use the decomposition $g=k_1 e^{r \sigma_3/2} k_2$ with real $r\neq 0$ and $k_1,k_2\in\Su$. Then it suffices to show $D^{(k,\rho)}(e^{r \sigma_3/2})|\,v\,\rangle\not\in\ch_k$ for any $|v\rangle\in\ch_k$ and any real $r\neq 0$. It suffices to only consider the basis and let $|v\rangle=|k,m\rangle$. The action of $D^{(k,\rho)}(e^{r \sigma_3/2})$ on $|k,m\rangle$ leaves $m$ invariant but generally changes $k$. If $D^{(k,\rho)}(e^{r \sigma_3/2})|k,m\rangle\in\ch_k$, then $|k,m\rangle$ would have to be an eigenstate of $D^{(k,\rho)}(e^{r \sigma_3/2})$. But this contradicts to the fact that $D^{(k,\rho)}(e^{r \sigma_3/2})$ only has continuous spectrum. Indeed, use the representation $\ch_{(k,\rho)}\simeq L^2(\C,\rmd^2 z)$, the action $D^{(k,\rho)}(e^{r \sigma_3/2})$ on any $f(z,\bar{z})\in \ch_{(k,\rho)}$ is a dilation: 
\be
D^{(k,\rho)}(e^{r \sigma_3/2})f(z,\bar{z})=e^{-r(i\rho/2-1)}f(e^r z,e^r\bar{z}).
\ee 
Consider the coordinate transformation $z = e^{\xi + i\phi}$, where $\xi \in \mathbb{R}$ and $\phi \in [0, 2\pi)$. The measure transforms as $\rmd^2z = e^{2\xi} \rmd\xi \rmd\phi$.
We define a unitary map $V: L^2(\mathbb{C},\rmd^2z) \to L^2(\mathbb{R} \times S^1,\rmd\xi\rmd\phi)$ by: $
\psi(\xi, \phi) = [V f] (\xi, \phi) = e^{\xi} f(e^{\xi + i\phi},e^{\xi - i\phi})$ \footnote{
	The factor $e^\xi$ ensures unitarity: $
\int |f(z)|^2 d^2z = \int |e^{-\xi} \psi|^2 e^{2\xi} d\xi d\phi = \int |\psi|^2 d\xi d\phi$.
}.
The action of $D^{(k,\rho)}(e^{r \sigma_3/2})$ on $\psi$ is given by
\be
[V D^{(k,\rho)}(e^{r \sigma_3/2}) V^{-1} \psi](\xi, \phi) = e^{-i r \rho/2} \psi(\xi + r, \phi).
\ee
The generator of the boost is denoted by $K_3$ and $D^{(k,\rho)}(e^{r \sigma_3/2})=\exp(-i r K_3)$. Differentiating the above relation with respect to $r$ at $r=0$, we find the transformed generator $\tilde{K}_3 = V K_3 V^{-1}=i \partial_\xi+\rho/2$, which is a self-adjoint operator with continuous spectrum. Therefore, $D^{(k,\rho)}(e^{r \sigma_3/2})$ has only continuous spectrum for any $r\neq 0$ and thus cannot have normalizable eigenstate.

\end{proof}

\begin{theorem}\label{zetabound}
	$|\t_k^{(h)}|\leq d_k^2$, the equality holds if and only if $g_{ve}^{-1}g_{ve'}\in \Su$, for all $e,e'\subset\partial h$, $e\cap e'=v$, and $\overrightarrow{\prod}_{v\in\partial h} g_{ve}^{-1}g_{ve'} =\pm \mathbb{I}$.
	
\end{theorem}
	
\begin{proof}
For any state $\ff\in\ch_k\subset \ch_{(k,\rho)}$, we have $\Vert D^{(k,\rho)}(g_{ve}^{-1}g_{ve'})P_k |\,\ff\,\rangle\Vert =  \Vert \,\ff\,\Vert $ by the unitarity, then after the projection by $P_k$ back to $\ch_k$, $\Vert P_k D^{(k,\rho)}(g_{ve}^{-1}g_{ve'})P_k |\,\ff\,\rangle\Vert \leq  \Vert \,\ff\,\Vert $. The equality holds if and only if $D^{(k,\rho)}(g_{ve}^{-1}g_{ve'})P_k |\,\ff\,\rangle\in\ch_k$, equivalent $g_{ve}^{-1}g_{ve'}\in \Su$ by the above lemma. For any internal face $h$ with $N$ vertices,
\be
\lt\Vert \overrightarrow{\prod_{v=1}}^N P_k D^{(k,\rho)}(g_{ve}^{-1}g_{ve'})P_k |\,\ff\,\rangle\rt\Vert \leq  \lt\Vert \overrightarrow{\prod_{v=2}}^N P_k D^{(k,\rho)}(g_{ve}^{-1}g_{ve'})P_k |\,\ff\,\rangle\rt\Vert\leq \cdots \leq \Vert \,\ff\,\Vert
\ee
The equality $\Vert \overrightarrow{\prod}_{v} P_k D^{(k,\rho)}(g_{ve}^{-1}g_{ve'})P_k |\,\ff\,\rangle\Vert =  \Vert \,\ff\,\Vert $ holds if and only if all $g_{ve}^{-1}g_{ve'}\in \Su$. By using this inequality and the Schwartz inequality, we obtain
\be
	\lt|\lag \ff_1\lt|\overrightarrow{\prod_{v\in\partial h}}P_kD^{(k,\rho)}\lt(g_{ve}^{-1}g_{ve'}\rt)P_k\rt|  \ff_2\rag\rt|\leq \Vert \ff_1\Vert\lt\Vert \overrightarrow{\prod_{v\in\partial h}} P_k D^{(k,\rho)}(g_{ve}^{-1}g_{ve'})P_k |\,\ff_2\,\rangle\rt\Vert\leq \Vert \ff_1\Vert\Vert \ff_2\Vert=1,\label{schwartzineq}
\ee
for normalized states $\ff_1,\ff_2$. The equality holds if and only if for all $g_{ve}^{-1}g_{ve'}\in \Su$ and $\overrightarrow{\prod}_{v\in\partial h} D^k(g_{ve}^{-1}g_{ve'})\mid  \ff_2\rangle= e^{i\theta}\mid\ff_1 \rangle$ with $\theta\in\R$, where $D^k$ denotes the SU(2) unitary irrep on $\ch_k$. Moreover,
\be
	\lt|\tr_{(k,\rho)}\lt[\overrightarrow{\prod_{v\in\partial h}}P_kg_{ve}^{-1}g_{ve'}P_k\rt]\rt|\leq  \sum_{m=-k/2}^{k/2}\lt|\lag k,m\lt| \overrightarrow{\prod_{v\in\partial h}} P_k D^{(k,\rho)}(g_{ve}^{-1}g_{ve'})P_k  \rt| k,m\rag\rt|\leq d_k,\label{2ineq}
\ee
The second inequality in \eqref{2ineq} holds if and only if all $g_{ve}^{-1}g_{ve'}\in \Su$ and $\overrightarrow{\prod}_{v\in\partial h} g_{ve}^{-1}g_{ve'} = \exp(i\theta L_3)$ for some $\theta\in\R$, where $L_3|k,m\rangle = m |k,m\rangle$ \footnote{
	Denote $O=\overrightarrow{\prod}_{v\in\partial h} P_k D^{(k,\rho)}(g_{ve}^{-1}g_{ve'})P_k $,
	the equality $\sum_m |\langle k,m| O |k,m\rangle| = d_k$ and \eqref{schwartzineq} requires $|\langle k,m| O |k,m\rangle| = 1$ for every $m$. This implies that all $g_{ve}^{-1}g_{ve'}$ must be in SU(2). Let $G = \overrightarrow{\prod}_{v\in\partial h} g_{ve}^{-1}g_{ve'} \in \text{SU(2)}$. The operator $O$ now simply becomes the SU(2) representation matrix $D^{(k/2)}(G)$. The condition $|\langle k,m| D^{(k/2)}(G) |k,m\rangle|=1$ also means $|k,m\rangle$ must be an eigenvector of $D^{(k/2)}(G)$ (saturation of Schwartz inequality). Since this must hold for all basis vectors $|k,m\rangle$, the operator $D^{(k/2)}(G)$ must be diagonal in this basis. The basis $|k,m\rangle$ is the eigenbasis of the generator $L_3$, so $G$ must be a rotation about the z-axis, i.e., $G = \exp(i\theta L_3)$.
}. Then $\tr_{(k,\rho)}\lt[\overrightarrow{\prod}_{v\in\partial h} P_kg_{ve}^{-1}g_{ve'}P_k\rt]=\frac{\sin \left(  (k+1)\theta/2\right)}{\sin \left({\theta }/{2}\right)}$ is the SU(2) character. $|\tr_{(k,\rho)}\lt[\overrightarrow{\prod}_{v\in\partial h}P_kg_{ve}^{-1}g_{ve'}P_k\rt]|=d_k$ if and only if $\theta=0,2\pi$, i.e. $\overrightarrow{\prod}_{v\in\partial h} g_{ve}^{-1}g_{ve'} =\pm 1$.
	
\end{proof}

\begin{lemma}\label{lemmaAnalyQ}
	The function $\tau_k^{(h)}(g_h)$ is real analytic in $g_h$.
	
\end{lemma}

	\begin{proof}
	First of all, the canonical basis vector $|(k,\rho), k ',m\rangle$ is $K$-finite, where $K=\Su$ is the maximal compact subgroup of $\Slc$, then Harish-Chandra's analyticity theorem \cite{etingof2024representationsliegroups} \footnote{Harish-Chandra's analyticity theorem states that if $V$ carries a unitary irrep of a semisimple Lie group $G$ with maximal compact subgroup $K$ then every $K$-finite vector is a weakly analytic vector. A vector $v\in V$ is $K$-finite if it is contained in a finite-dimensional subrepresentation of $K$. A vector $v\in V$ is weakly analytic vector if $\langle u |g| v\rangle$ ($g\in G$) is analytic on $G$ (as a real manifold) for any $u\in V$.} implies that it is weakly analytic, so the Wigner $D$-function of the $\Slc$ unitary irrep $D^{(k,\rho)}_{k_1 m_1,k_2 m_2}(g)=\langle (k,\rho), k_1,m_1 | g | (k,\rho), k_2,m_2\rangle$ is an analytic function on $\Slc$. Consequently, $\t_k^{(h)}(g_h)$ is an analytic function of $g_h=\{g_{ve}\}_{e\subset h}$ for each $k$, since it is a polynomial of the $D$-functions.
	

	\end{proof}

\section{Inverse Laplace transform}\label{Inverse Laplace transform}

\begin{lemma}\label{KlimitTheta}
Consider the following integral with $T>0$ and $\L\in\R$,
\be
K_R(\Lambda) = \frac{1}{2\pi i} \int_{T-iR}^{T+iR} \frac{\rmd s}{s} e^{\Lambda s},
\ee
For $\L\neq 0$, 
\be
\lt|K_R(\L)-\Theta(\L)\rt|\leq \frac{e^{\Lambda T}}{\pi R |\Lambda|}.
\ee
For $\L=0$,
\be
K_R(0)=\frac{1}{\pi} \arctan\left(\frac{R}{T}\right)\to \frac{1}{2},\quad R\to\infty.
\ee
\end{lemma}

\begin{proof}

$\Lambda > 0$: We consider a rectangular contour $\mathcal{C}$ in the complex $s$-plane with vertices at $T-iR, T+iR, -M+iR, -M-iR$, where $M,T > 0$. The pole of the integrand at $s=0$ lies inside this contour. The integral along $\cc$ is given by the residue:
\be
\frac{1}{2\pi i}\oint_{\mathcal{C}} \frac{\rmd s}{s} e^{\Lambda s} = 1
\ee
The integral splits into four parts, each of which is along an edge of the rectangle. As $M \to \infty$, the integral over the left vertical segment vanishes. Thus
\be
K_R(\lambda) = 1 + \frac{1}{2\pi i} \left( \int_{-\infty+iR}^{T+iR} \frac{e^{\Lambda s}}{s} \rmd s + \int_{T-iR}^{-\infty-iR} \frac{e^{\Lambda s}}{s} \rmd s \right)
\ee
We bound the error terms (the horizontal integrals). On the top segment, $s = x + iR$, so $|s| \geq R$.
\be
\left| \int_{-\infty}^{T} \frac{e^{\Lambda(x+iR)}}{x+iR} \rmd x \right| \leq \frac{1}{R} \int_{-\infty}^{T} e^{\Lambda x} \rmd x = \frac{e^{\Lambda T}}{R \Lambda}
\ee
The same bound applies to the bottom segment. Thus:
\be
|K_R(\Lambda) - 1| \leq \frac{e^{\Lambda T}}{\pi R \Lambda} .
\ee

$\Lambda < 0$: We close the contour to the right, with vertices $T-iR, T+iR, M+iR, M-iR$, $M>T>0$. The pole at $s=0$ is outside this contour. The integral along the contour gives 0. Similar to the above, the horizontal integrals provide the error term as $M\to\infty$:
\be 
|K_R(\Lambda) - 0| \leq \frac{e^{\Lambda T}}{\pi R |\Lambda|} .
\ee

$\Lambda = 0$: 
\be
K_R(0) = \frac{1}{2\pi i} \int_{T-iR}^{T+iR} \frac{\rmd s}{s} = \frac{1}{2\pi} \int_{-R}^{R} \frac{\rmd y}{T+iy} 
\ee
Using $T+iy = \frac{T-iy}{T^2+y^2}$, the imaginary part is odd and integrates to 0.
\be
 K_R(0) = \frac{1}{2\pi} \int_{-R}^{R} \frac{T}{T^2+y^2} \rmd y = \frac{1}{\pi} \arctan\left(\frac{R}{T}\right) \to \frac{1}{2}, \quad R\to\infty.
\ee

\end{proof}

\begin{lemma}\label{expAtlemmaB2}
Given that
\be
\Xi(s) =
\prod_{k=1}^\infty
\frac{1}{1-\lambda \tau_k e^{-s E_k} }-1. \label{X1shA6C2}
\ee 
Assume $\{E_k\}_k$ to be positive integers and monotonically growing as polynomials, for any $T>0$ such that the contour of the integral
\be
\omega(A,T)=
\mathscr{P}\int_{T-i\infty}^{T+i\infty}
\frac{\rmd s}{2\pi i s}\, e^{A s}\Xi(s),\qquad A>0
\ee
does not pass any pole of $\Xi(s)$ 
\footnote{
	$T$ might not be greater that the real parts of all $\Xi(s)$'s poles and $T\neq \ln(|\lambda  \tau_k|)/E_k$ for all $k\in\Z_+$.
}, the integral satisfies the bound $|\o|\leq C(T) e^{AT}$, where $C(T)>0$ does not depend on $A$.

\end{lemma}

\begin{proof}
Due to $|\t_k|\leq d_k^2$, we can always classify the indices $k$ into two sets:
\be
K_{+} =  \{ k :\ |\lambda \tau_k| e^{-T E_k} < 1 \},\qquad K_{-} = \{ k :\  |\lambda\tau_k| e^{-T E_k} > 1 \}.
\ee
Since $E_k$ grows monotonically as $k$ grows, $K_{-}$ is a finite set, and $K_{+}$ contains all large $k$'s. There is no $k$ for $|\lambda\tau_k| e^{-T E_k} = 1$ because the poles are avoided by the contour. We decompose $\Xi(s) + 1$ into two partial products:
\be
\Xi(s) + 1 = P_-(s) P_+(s), \qquad P_\pm(s) = \prod_{k \in K_{\pm}} \frac{1}{1 - \lambda \tau_k e^{-s E_k}}. 
\ee
We expand each factor into a geometric series valid on the line $\text{Re}(s) = T$: 

\begin{itemize}

\item For $k \in K_{+}$, we have $|\lambda \tau_k e^{-s E_k}| < 1$ and $\sum_{k\in K_+}|\lambda \tau_k e^{-s E_k}|\leq |\lambda|\sum_{k\in K_+} d^2_k e^{-T E_k} <\infty$ at $\text{Re}(s)=T$. The expansion in Lemma \ref{lemmaA2} applies:
\be
P_+(s)= \sum_{\{n_k\}_{k\in K_+}} \prod_{k\in K_+}\lt(\lambda \tau_k e^{-s E_k} \rt)^{n_k}.
\ee
This series converges absolutely for $\re(s)=T$. 

\item For $k \in K_{-}$, we have $|\lambda \tau_k e^{-s E_k}| > 1$ at $\text{Re}(s)=T$. We expand in powers of the inverse $ \frac{1}{1 - u} = -\frac{u^{-1}}{1 - u^{-1}} = -\sum_{m=1}^\infty u^{-m} $ with $u = \lambda \tau_k e^{-s E_k}$, and the expansion is absolutely convergent. We obtain
\be
P_-(s)&=&\prod_{k\in K_-}\frac{1}{1 - \lambda \tau_k e^{-s E_k}} = \prod_{k\in K_-}\lt[ -\sum_{m_k=1}^\infty (\lambda \tau_k)^{-m_k} e^{s m_k  E_k} \rt]\nonumber\\
&= &\sum_{\{m_k\}_{k\in K_-}}\prod_{k\in K_-}\lt[ -(\lambda \tau_k)^{-1} e^{s E_k} \rt]^{m_k},\qquad m_k\geq 1.\label{PminusSeries}
\ee
We have interchanged the sum and product by the finiteness of $K_-$. The series \eqref{PminusSeries} converges absolutely for $\re(s)=T$. 

\end{itemize}

Combining these, $\Xi(s)+1$ can be written as an absolutely convergent series for $\re(s)=T$ \footnote{$ \sum_{\mu} |\beta_\mu e^{-s \ce_\mu}|\leq \sum_{\mathbf{n}\in\O}\prod_{k=1}^\infty a_k^{n_k}$ where $a_k=|-(\l\t_k)^{-1} e^{s E_k}|$ for $k\in K_-$ and $a_k=|\l\t_k e^{-s E_k}|$ for $k\in K_+$. Both conditions $a_k<1$ and $\sum_k a_k<\infty$ are satisfied.}:
\be
\Xi(s) +1= \sum_{\mu} \beta_\mu e^{-s \ce_\mu} ,\label{absconvergenceB16}
\ee
where we have introduced short-hand notations
\be
\mu\equiv \{m_k\}_{k\in K_-}\cup \{n_k\}_{k\in K_+},\quad 
\b_\mu\equiv \prod_{k\in K_-}\lt[ -(\lambda \tau_k)^{-m_k}\rt]\prod_{k\in K_+}\lt(\lambda \tau_k \rt)^{n_k}\quad
\ce_\mu =\sum_{k\in K_+} n_k E_k-\sum_{k\in K_-} m_k E_k.
\ee
The absolute convergence implies the uniform convergence on the vertical line with $\re(s)=T$.

Given that $\o$ is defined as the principal value, we consider the following integral on the finite segment $s \in [T-iR, T+iR]$
\be
\o_R = \frac{1}{2\pi i} \int_{T-iR}^{T+iR} \frac{\rmd s}{s}e^{As} \left( \sum_{\mu} \beta_\mu e^{-s \ce_\mu} - 1  \right)= \frac{1}{2\pi i} \lt[\sum_{\mu}\beta_\mu \int_{T-iR}^{T+iR} \frac{\rmd s}{s}e^{As} e^{-s \ce_\mu} - \int_{T-iR}^{T+iR} \frac{\rmd s}{s}e^{As} \rt],
\ee
and $\o = \lim_{R \to \infty} \o_R$. We can interchange the sum and the integral by the properties of uniformly convergent series.

By Lemma \ref{KlimitTheta}, $\lim_{R \to \infty} K_R(A-\ce_\mu) = \Theta(A-\ce_\mu)$, where $\Theta(x)$ is $1$ for $x>0$, $0$ for $x<0$, and $1/2$ for $x=0$. We now prove that $\lim_{R \to \infty} \sum_\mu \b_\mu K_R(A-\ce_\mu) = \sum_\mu \b_\mu \Theta(A-\ce_\mu)$. We estimate the error
\be
\left|\sum_\mu \b_\mu K_R(A-\ce_\mu) - \sum_\mu \b_\mu \Theta(A-\ce_\mu) \right| \leq \sum_{\mu} |\b_\mu| \lt|K_R(A-\ce_\mu) - \Theta(A-\ce_\mu)\rt| \equiv \Delta_R 
\ee
We split the sum into terms where $A \neq \ce_\mu$ and terms where $A =\ce_\mu$. For the terms with $A \neq \ce_\mu$, we using the bounds from Lemma \ref{KlimitTheta}:
\be 
\Delta_R^{\neq} \leq \frac{1}{\pi R} \sum_{\mu;A \neq \ce_\mu} |\b_\mu| \frac{e^{(A-\ce_\mu)T}}{|A-\ce_\mu|} = \frac{e^{AT}}{\pi R} \sum_{\mu;A \neq \ce_\mu} \frac{|\b_\mu| e^{-T \ce_\mu}}{|A-\ce_\mu|} 
\ee
The set of values $\{\ce_\mu\}_\mu$ has no finite accumulation point because it is a linear combination of integers with integer coefficients. This ensures a minimal distance $\delta>0$ between $A$ and $\ce_\mu\neq A$, i.e. $|A-\ce_\mu| \geq\delta$. Therefore,
\be
\Delta_R^{\neq} \leq \frac{e^{AT}}{\pi \delta  R} \sum_{\mu;A \neq \ce_\mu}{|\b_\mu| e^{-T \ce_\mu}} . \label{DeltaRneq}
\ee
On the other hand, for terms with $A = \ce_\mu$, we use $|\frac{1}{\pi} \arctan(R/T) - \frac{1}{2}|\leq C \frac{T}{R}$ for large $R$ and some $C>0$,
\be
\Delta_R^{=}\leq C \frac{T}{R}\sum_{\mu;A = \ce_\mu}|\b_\mu|=CT \frac{e^{AT}}{R}\sum_{\mu;A = \ce_\mu}|\b_\mu|e^{-T \ce_\mu}\label{DeltaReq}
\ee
Summing \eqref{DeltaRneq} and \eqref{DeltaReq}, 
\be
\Delta_R=\Delta_R^{\neq} +\Delta_R^=\leq C' (T)\frac{e^{AT}}{R} \sum_{\mu}{|\b_\mu| e^{-T \ce_\mu}},\qquad C'(T)=\max\lt\{CT,\frac{1}{\pi\delta}\rt\}>0, 
\ee
The sum converges due to the absolute convergence of $\Xi(s)$ in \eqref{absconvergenceB16} for $\re(s)=T$, so the total error vanishes as $R\to\infty$. We obtain $\o=\sum_\mu \b_\mu \Theta(A-\ce_\mu)-1$, then
\be
|\o|\leq \sum_\mu |\b_\mu| \Theta(A-\ce_\mu)+1\leq e^{AT}\lt(\sum_\mu |\b_\mu| e^{-\ce_\mu T}+1\rt)\equiv C(T)e^{AT}.
\ee
where we uses $\Theta(x)\leq e^{xT}$ and uses again the absolute convergence of \eqref{absconvergenceB16} for $\re(s)=T$ in the last step. Note that here we have to use the original value of $T$ in order to ensure the absolute convergence, although $\Theta(x)\leq e^{xT}$ is valid for all $T>0$.

\end{proof}

\begin{lemma}\label{smoothboundB3}
Denote by $\re(s_*(g_h))$ the maximum among the real parts of the poles of $\Xi(s)$ in the integrand of $\o_h^{\rm bos}$. Consider any domain $U$ of $g_h$, and if there exists $T>0$ such that $T>\re(s_*(g_h))$ for all $g_h\in U$, then there exists a function $C(g_h,T)$ continuous in $g_h\in U$ such that
\be
\lt|\o_h^{\rm bos}(A_h;g_h)\rt|\leq C(g_h,T) e^{A_hT},
\ee

\end{lemma}

\begin{proof}
$T$ is greater than the real parts of all poles, so $C(T)=\sum_\mu |\b_\mu| e^{-\ce_\mu T}+1$ in Lemma \ref{expAtlemmaB2} becomes 
$$
C(g_h,T)={\sum_{\{n_k\}_{k\in\Z_+}}}
\prod_{k=1}^\infty
\left[ \lt|\lambda_h \tau_k^{(h)}(g_h)\rt| e^{-T E_k} \right]^{n_k} +1=\prod_{k=1}^\infty \frac{1}{1-\lt|\lambda_h \tau_k^{(h)}(g_h)\rt| e^{-TE_k}} + 1.
$$
By $| \tau_k^{(h)}(g_h)|\leq d_k^2$ and the fact that $|\lambda_h| d_k^2 e^{-TE_k}<1$ for $k\geq k_s$ with a sufficiently large $k_s$, the following infinite product converges uniformly
\be
&&\prod_{k=k_s}^\infty \frac{1}{1-\lt|\lambda_h \tau_k^{(h)}(g_h)\rt| e^{-TE_k}}\leq \prod_{k=k_s}^\infty \frac{1}{1-\lt|\lambda_h\rt| d_k^2 e^{-TE_k}}
\ee
Therefore $C(g_h,T)$ is a continuous function of $g_h$, due to the uniform convergence and continuity of $\tau_k^{(h)}(g_h)$.

\end{proof}

\begin{theorem}
For $T>0$ greater than the real parts of all poles of $\Xi(s)$ in the right-half plane, for large $A$
\be
\o(A,T)=\lt[\sum_{s_*}\underset{s\to s_*}{\mathrm{Res}}\frac{e^{s A}}{s}\Xi(s)\rt]\lt[1+O(A^{-\infty})\rt]
\ee
where $s_*$ are the poles of $\Xi(s)$ having the largest real parts. $\mathrm{Res}_{s\to s_*}$ denotes the residue at $s=s_*$. $O(A^{-\infty})$ is the subleading contribution exponentially suppressed for large $A$.


\end{theorem}

\begin{proof}
We write $\l_h \t^{(h)}_k(g_h) = \rho_k e^{i\phi_k}$, $\phi_k\in[0,2\pi)$, $0\leq\rho_k\leq \l_h d_k^2$. $\Xi(s)$ has poles located at
\be
s(k,m)=\frac{1}{E_k}\lt[\ln \rho_k+i\lt(\phi_k+2\pi m\rt)\rt],\qquad k\in\Z_+,\ m\in\Z.
\ee
By $E_k\geq Bk$ for some $B>0$, the real part of the poles satisfies $\frac{\ln\rho_{k}}{E_{k}}\leq\frac{\ln\left(|\lambda| d_{k}^{2}\right)}{E_{k}}\leq\frac{\ln\left(|\lambda| d_{k}^{2}\right)}{Bk}$ which vanishes as $k\to\infty$, so $\frac{\ln\rho_{k}}{E_{k}}$ approaches to zero or negative as $k\to\infty$. For the poles whose real parts are positive and not close to zero, i.e. $\frac{\ln\rho_{k}}{E_{k}}>T_0$ with $0<T_0<T$, they only relate to finitely many $k$, say, $k<k_0$. Let us focus on these poles. Then we find a positive sequence $\{R_n\}_{n=1}^\infty$ such that for every $n$, $R_n\neq E_k^{-1}(\phi_k+2\pi m)$ for any $m\in\Z$, $k\in\Z_+$ and $k<k_0$, and $\lim_{n\to\infty}R_n=\infty$.

We consider $\o_n$ by truncating the integration contour to the vertical segment $[T-iR_n,T+iR_n]$, so $\o=\lim_{n\to\infty}\o_n$. We close the contour by adding edges $[T_0+iR_n,T+iR_n]$, $[T_0-iR_n,T-iR_n]$ and $[T_0-iR_n,T_0+iR_n]$ to form a rectangle, assuming $T_0\neq \frac{\ln\rho_{k}}{E_{k}}$ for any $k$. On the horizontal edges $[T_0+iR_n,T+iR_n]$ and $[T_0-iR_n,T-iR_n]$, $|\Xi(s)|$ are uniformly bounded: $|\Xi(s)|\leq C$ for some constant $C>0$, so 
\be
\lt|\int_{T_0\pm iR_n}^{T\pm iR_n}\frac{\rmd s}{s}e^{As}\Xi(s)\rt|\leq C\int_{T_0}^{T}\frac{\rmd x}{|x\pm iR_n|}e^{Ax}\leq \frac{C}{AR_n}\lt[e^{AT}-e^{AT_0}\rt]\to 0,\qquad n\to\infty.
\ee 
By Lemma \ref{expAtlemmaB2},
\be
\lt|\int_{T_0- iR_n}^{T_0+ iR_n}\frac{\rmd s}{s}e^{As}\Xi(s)\rt|\leq C(T_0)e^{AT_0}.
\ee
Denote the rectangular contour by $\cc$, 
\be
\o_n&=&\frac{1}{2\pi i}\oint_\cc\frac{\rmd s}{s}e^{As}\Xi(s)+\int_{T_0- iR_n}^{T_0+ iR_n}\frac{\rmd s}{s}e^{As}\Xi(s)+\int_{T_0+ iR_n}^{T+ iR_n}\frac{\rmd s}{s}e^{As}\Xi(s)-\int_{T_0- iR_n}^{T- iR_n}\frac{\rmd s}{s}e^{As}\Xi(s)\nonumber\\
&=&\sum_{s(k,m)}\underset{s\to s(k,m)}{\mathrm{Res}}\frac{e^{s A}}{s}\Xi\lt(s\rt)+O(e^{AT_0})+O(R_n^{-1}).
\ee
where we only sum the poles enclosed by the contour $\cc$. The dominant contribution to the sum comes from the poles $s_*$ with largest real parts: There exists a finite set $K_{*}$ such that $\frac{\ln\rho_{k_*}}{E_{k_*}}=\sup_{k}\frac{\ln\rho_k}{E_k}$ for any $k_*\in K_{*}$.
\be
s_*=s(k_*,m)=\frac{\ln\rho_{k_*}}{E_{k_*}}+i\frac{\phi_{k_*}+2\pi m}{E_{k_*}},\qquad k_*\in K_*.
\ee
By the limit $n\to\infty$,
\be
\o=\lt[\sum_{s_*}\underset{s\to s_*}{\mathrm{Res}}\frac{e^{s A}}{s}\Xi\lt(s\rt)\rt]\lt[1+O(A^{-\infty})\rt].
\ee
where the range of the sum covers entire $m\in\Z$.

\end{proof}

\section{Bound the derivatives of $\tau^{(h)}_k(g_h)$}\label{Bound the derivatives of tau}

Consider the principal series unitary irrep $(k,\rho)$ of $\Slc$ carried by the Hilbert space $\mathcal{H}_{(\rho,k)}=\oplus_{m=k}^\infty \mathcal{H}_m$, where $\mathcal{H}_m$, $m=k,k+2,\cdots$ is the $\Su$ irrep of spin $m/2$. Define the projection 
\be 
P_m: \mathcal{H}_{(k,\rho)} \to \mathcal{H}_m
\ee  
onto the subspace of spin $m/2$. We set $\rho=\gamma(k+2)$ by the simplicity constraint.

\begin{lemma}\label{lemmaPXPbound}
Let $X$ be a generator of $\mathfrak{sl}(2, \mathbb{C})$. The operator norm of the operator $ X P_m$ is bounded by a polynomial of $m,k$.
\end{lemma}

\begin{proof} 
For any generator $X \in \mathfrak{sl}(2, \mathbb{C})$, the image of a vector in $\ch_m$ lies in the direct sum of adjacent subspace, i.e.
\be
X:\ \ch_m \to \ch_{m-2} \oplus \ch_m \oplus \ch_{m+2} 
\ee
The operator $X P_m$  is bounded because it annihilates all states $ \mathcal{H}_{(\rho,k)}$ except the ones in the finite-dimensional space $\ch_m$. 
We define $O\equiv P_{m'} X P_m$ where $m' \in \{m-2, m, m+2\}$.

The Lie algebra $\mathfrak{sl}(2, \mathbb{C})$ is spanned by the rotation generators $L_i$ (generating $\mathfrak{su}(2)$) and the boost generators $K_i$ (where $i=1,2,3$). We consider the canonical basis $|j, \mu\rangle$ for $\ch_m$ (where $j=m/2$).
The generators $L_i$ act as angular momentum operators.
\be 
\| L_3 |j, \mu\rangle \| = |\mu| \le j\qquad  \| L_\pm |j, \mu\rangle \| = \sqrt{(j \mp \mu)(j \pm \mu + 1)} \le 2c' j ,\qquad c'>0 
\ee
Thus, $\|P_{m'} L_iP_m|j, \mu\rangle\| \le c' m$ for some constant $c'>0$.

The boost generators $K_i$ mix spins. The action on a basis vector $|j, \mu\rangle$ is
\be
K_3|j, \mu\rangle&= & -\alpha_{(j)} \sqrt{j^2-\mu^2}|j-1, \mu\rangle-\beta_{(j)} \mu|j, \mu\rangle +\alpha_{(j+1)} \sqrt{(j+1)^2-\mu^2}|j+1, \mu\rangle, \label{K3action}\\
K_{+}|j, \mu\rangle&= & -\alpha_{(j)} \sqrt{(j-\mu)(j-\mu-1)}|j-1, \mu+1\rangle -\beta_{(j)} \sqrt{(j-\mu)(j+\mu+1)}|j, \mu+1\rangle \nonumber\\
&&-\alpha_{(j+1)} \sqrt{(j+\mu+1)(j+\mu+2)}|j+1, \mu+1\rangle, \\
K_{-}|j, \mu\rangle&= & \alpha_{(j)} \sqrt{(j+\mu)(j+\mu-1)}|j-1, \mu-1\rangle -\beta_{(j)} \sqrt{(j+\mu)(j-\mu+1)}|j, \mu-1\rangle  \nonumber\\
&&+\alpha_{(j+1)} \sqrt{(j-\mu+1)(j-\mu+2)}|j+1, \mu-1\rangle,\label{Kmaction}
\ee
where
$$
\alpha_{(j)}=\frac{i}{j} \sqrt{\frac{\left(j^2-(k/2)^2\right)\left(j^2+(\rho/2)^2\right)}{4 j^2-1}}, \quad \beta_{(j)}=\frac{k \rho}{4j(j+1)}
$$
All coefficients on the right-hand sides in \eqref{K3action} - \eqref{Kmaction} are bounded by $\mathrm{Pol}(2j,k)$ for some polynomial function $\mathrm{Pol}$. Therefore, $\|P_{m'} K_iP_m|j, \mu\rangle\| \le  \mathrm{Pol}(m,k)$. So $\|O|j, \mu\rangle\|$ is polynomially bounded by $m,k$ for all generators in $sl(2,\C)$.

For any normalized state $\psi \in \mathcal{H}_{(\rho, k)}$, we can rescale its $m$-th spin component as $a_m \psi_m = P_m\psi \in \mathcal{H}_m$ with $\|\psi_m\| = 1$ (the normalization ensures $a_m \leq 1$, where $a_m$ is the normalization factor). Expanding $\psi_m$ in the canonical basis $|j, \mu\rangle$ (with $j = m/2$), we have $\psi_m = \sum_{\mu=-j}^j c_\mu |j, \mu\rangle$. The operator $O$ acts as
$
\|O\psi\| \leq \|O \psi_m\| \leq \sum_{\mu=-j}^j |c_\mu|\, \|O|j, \mu\rangle\| .
$
From earlier, we know that for each basis element $|j,\mu\rangle$, $\|O|j, \mu\rangle\| \leq \mathrm{Pol}(m,k)$. Thus,
$
\|O\psi\| \leq \mathrm{Pol}(m,k)\sum_{\mu=-j}^j |c_\mu| \leq ({m+1})\mathrm{Pol}(m,k),
$
by $\sum_{\mu=-j}^j |c_\mu|^2=1$. This estimate holds for any normalized state $\psi$. As a result,
\be
\|XP_m\psi\|\leq \sum_{m'=m-2}^{m+2}\|P_{m'}XP_m\psi\|\leq 3({m+1})\mathrm{Pol}(m,k)
\ee 
for all normalized state $\psi$.

\end{proof}

\begin{lemma}\label{lemmaXXXXXPbound}
The operator $X_1\cdots X_n P_k$ on $\mathcal{H}_{(k,\rho)}$, where $X_1,\dots, X_n$ are representations of Lie algebra generators, has the operator norm bounded by a polynomial of $k$.
\end{lemma}

\begin{proof}
Since $P_k$ annihilates any vector orthogonal to $\ch_k$, it suffices to consider a normalized vector $v_0 \in \ch_k$ ($j=k/2$). Let $v_0 \in H_k$ with $\|v_0\|=1$.
Define the sequence of vectors:
\be
v_1 = X_n v_0 ,\qquad  v_2 = X_{n-1} v_1, \qquad\cdots, v_n = X_1 v_{n-1}  
\ee

In general, $v_r$ lies in the subspace $V^{(r)} = \bigoplus_{l=0}^{r} \ch_{k+2l}$.
The maximum index involved in $v_r$ is $m_{\rm max}(r) = k + 2r$. We bound the norm iteratively. For any $v_r \in V^{(r)}$, $v_r = \sum_{l=0}^r v_{k+2l}$ where $v_{k+2l} \in \ch_{k+2l}$.
\be 
\|v_{r+1}\| &=&\|X_{n-r} v_r\| \le \sum_{l=0}^r \|X_{n-r} v_{k+2l}\| = \sum_{l=0}^r\|X_{n-r} P_{k+2l} v_r\|\nonumber\\
&\leq&\sum_{l=0}^r \mathrm{Pol}_0(k+2l,k)\| v_r\|\equiv  \mathrm{Pol}_{1}(r,k)\| v_r\|,
\ee
where $\mathrm{Pol}_0(m,k)=3({m+1})\mathrm{Pol}(m,k)$ and $\sum_{l=0}^r \mathrm{Pol}_0(k+2l,k)$ are two polynomials. Applying this inequality $n$ times:
\be
\|v_n\|= \|X_1\cdots X_n P_k v_0\| \le  \lt[\prod_{r=0}^{n-1} \mathrm{Pol}_{1}(r,k)\rt] \|v_0\| 
\ee
Thus, the norm of the operator $X_1\cdots X_n P_k$ is bounded by a polynomial of $k$.

\end{proof}

An $6n$-dimensional multi-index is $\a=\{\a_{A,i}\}_{A,i}$, $\a_{A,i}=0,1,2,\cdots$, where $A=1,\cdots,6$ and $i=1,\cdots,n$. The notation $|\a|$ is defined by the sum of components $|\a|=\sum_{A,i}\a_{A,i}$. We define the notation $D^\alpha_g$ for an arbitrary order-$|\alpha|$ multi-index derivative: For any function $f(g_1,\cdots,g_n)$, $g_i\in\Slc$,
\begin{eqnarray}
D^\alpha_g f(g_1,\cdots,g_n)=\frac{\partial^{|\alpha|}}{\partial t^{\a_{1,1}}_{1,1}\cdots \partial t_{6,n}^{\a_{6,n}}}f\lt(g_1(\vec t_1),\cdots,g_n(\vec t_n)\rt)\Big|_{\vec t=0},\qquad g_i(t_i)=e^{i\sum_{A=1}^6 t_{A,i} J^A}g_i ,
\end{eqnarray}
where $\{J^A\}_{A=1}^6$ are the generators of the Lorentz Lie algebra. The index $i$ corresponds to the label $(v,e)$ in the following discussion.

\begin{lemma}\label{lemmaBoundDtau}
For any $\a$, $D^\a_g \tau^{(h)}_k(g_h)$ are uniformly bounded by a polynomial of $k$ on any compact neighborhood of $g_h$.
\end{lemma}

\begin{proof}
The function $\tau^{(h)}_k(g_h)$ depends on $g_h$ through $g_{ve}^{-1}g_{ve'}$. The derivative with respect to each of $g_{ve},g_{ve'}$ gives $g_{ve}^{-1}X g_{ve'}$ for some $X \in \mathfrak{sl}(2,\mathbb{C})$, while the second derivative gives $g_{ve}^{-1}X Y g_{ve'}$ for some $X,Y\in \mathfrak{sl}(2,\mathbb{C})$. The trace in $\tau^{(h)}_k(g_h)$ is equivalent to the trace $\tr_k$ of operators on $\ch_k$. By using the relation $\tr_k(A)\leq d_k \Vert A \Vert_k$ and $\Vert A B \Vert_k\leq \Vert A \Vert_k\Vert B \Vert_k$,
\begin{eqnarray}
D^\a_g \tau^{(h)}_k(g_h)\leq d_k^2\prod_{v\in\partial h}\Vert A^{m_v}_v\Vert_k,\qquad A^{m_v}_v=P_kg_{ve}^{-1}X_1\cdots X_{m_v} g_{ve'}P_k
\end{eqnarray}
where $X_1,\cdots, X_{m_v}$ are representations of Lie algebra generators. By using the adjoint action $g^{-1} J^A g=\sum_B c^A_{\ B}(g)J^B $, where $c^A_{\ B}(g)$ is smooth in $g\in\Slc$ and thus is uniformly bounded by a constant on any compact neighborhood. It suffices to prove that 
\begin{eqnarray}
\Vert P_kg X_1\cdots X_{m}P_k\Vert_k
\end{eqnarray}
is uniformly bounded by a polynomial of $k$. Indeed, the unitarity of the representation implies $\Vert P_kg \psi \Vert\leq \Vert \psi \Vert$ for all $\psi\in\ch_{(k,\rho)}$. By Lemma \ref{lemmaXXXXXPbound}, $\|X_1\cdots X_{m}P_k\psi\|\leq C_m k^N \|\psi\|$ for some $C_m>0,\ N>0$ and any $\psi\in\ch_{(k,\rho)}$. Therefore,
\be
\Vert P_kg X_1\cdots X_{m}P_k\psi\Vert\leq \Vert X_1\cdots X_{m}P_k\psi\Vert \leq C_m k^N \|\psi\|
\ee
which shows that the operator norm is uniformly bounded by a polynomial of $k$.

\end{proof}

\begin{lemma}\label{lemmaSmoothQ}
Define 
\be
Q_h^M(k_0,m_h,g_h,\lambda_h)=\prod_{k\neq k_0,k<M}\frac{1}{1-\lambda_h\tau_k^{(h)}(g_h)e^{-s_h(k_0,m_h,g_h,\lambda_h)E_k}}.
\ee
The functions $Q_h^M(k_0,m_h,g_h,\lambda_h)$ and $D_{g_h}^\a Q^M_h(k_0,m_h,g_h,\lambda_h)$, $\forall \a$, converge uniformly in $g_h$ on any compact neighborhood $K$ in ${\mathscr{U}}_h$ as $M\to\infty$. Therefore, $Q_h(k_0,m_h,g_h,\lambda_h)=\lim_{M\to \infty}Q^M_h(k_0,m_h,g_h,\lambda_h)$ is a smooth function on ${\mathscr{U}}_h$.
\end{lemma}

\begin{proof} 
We omit the fixed parameters $k_0, m_h, \lambda_h$ in the argument list of $Q$ and $X$ for brevity, focusing on the dependence on $g_h \in U$. We write the function $Q_h^M$ as
\be
Q_h^M(g_h) = \prod_{k \neq k_0,k < M } \frac{1}{1 - \cx_k(g_h)} ,\qquad \cx_k(g_h) = \lambda_h \tau_k^{(h)}(g_h) e^{-s_h(g_h) E_k},
\ee
where $s_h(g_h)$ is given by $ s_h(g_h) = E_{k_0}^{-1} \ln[\lambda_h \tau_{k_0}^{(h)}(g_h)] + {2\pi}\mathrm{i} m_h {E_{k_0}^{-1}} $. We denote by $\rho_k = {E_k}/{E_{k_0}}$, which grows polynomially with $k$, and denote $\theta_k = 2\pi m_h \rho_k$, which is real and independent of $g_h$. Then $\cx_k(g_h)$ can be written as:
\be
\cx_k(g_h) = \lambda_h \tau_k^{(h)}(g_h) \left( \lambda_h \tau_{k_0}^{(h)}(g_h) \right)^{-\rho_k} e^{-\mathrm{i} \theta_k}.
\ee
The bound $\frac{1}{|\lambda_h|} e^{E_{k_0} \beta^*_h}<|\tau_k^{(h)}(g_h)| \leq d_k^2$ on the open neighborhood $\mathscr{U}_h$ implies the following bound of $|\cx_k(g_h)|$
\be
 |\cx_k(g_h)| \leq |\lambda_h| d_k^2 e^{-\beta^*_h E_k} 
\ee
Since $E_k$ grows polynomially, $\sum_{k} |\cx_k(g_h)|$ converges uniformly on $\mathscr{U}_h$. Therefore, $Q^M_h$ converges uniformly to $Q_h$. In addition, for sufficiently large $k$, $|\cx_k(g_h)|\leq1/2$ for all $g_h \in \mathscr{U}_h$.

Let $\tau_k$ denote $\tau_k^{(h)}(g_h)$ and $\tau_{k_0}$ denote $\tau_{k_0}^{(h)}(g_h)$, and $ \cx_k = e^{-\mathrm{i} \theta_k} \lambda_h \tau_k (\lambda_h \tau_{k_0})^{-\rho_k} $. Consider differentiating with respect to $g$ (notationally, $D_g\equiv D^\alpha _g$ for $|\a|=1$):
\be
D_{g} \cx_k = e^{-\mathrm{i} \theta_k} (\lambda_h \tau_{k_0})^{-\rho_k} B^{(1)}_k,\qquad B^{(1)}_k=\lambda_h \tau_k' - \rho_k \lambda_h \tau_k \frac{\tau_{k_0}'}{\tau_{k_0}}\qquad \tau_k'\equiv D_g\tau_k 
\ee
By $|{\tau_{k_0}}|^{-1} \leq |\lambda_h| e^{-E_{k_0} \beta^*_h} $, the term $(\lambda_h \tau_{k_0})^{-\rho_k}$ is bounded by $e^{-\beta^*_h E_k}$. On any compact neighborhood $K$ in $\mathscr{U}_h$, $\tau_k, \tau_k'$ are bounded by polynomials in $k$. $\tau_{k_0}, \tau_{k_0}'$ depend only on $k_0$ (fixed), so they are bounded uniformly in $g_h$ on $K$. $\rho_k = E_k/E_{k_0}$ grows polynomially in $k$. Thus, $B^{(1)}_k$ is bounded by a polynomial in $k$, say $P_1(k)$. We obtain
\be 
| D_{g} \cx_k| \leq  P_1(k) e^{-\beta^*_h E_k} 
\ee
This ensures $\sum_k |\partial_{g_h} \cx_k|$ converges uniformly.

Consider the second derivative $|\a|=2$ and denote $D^2_g=D^\a_g$ for any multi-index of length 2:
\be
D_{g}^2 \cx_k 
&=&e^{-\mathrm{i} \theta_k}(\lambda_h \tau_{k_0})^{-\rho_k} B^{(2)}_k,\qquad B^{(2)}_k= -\rho_k \frac{\tau_{k_0}'}{\tau_{k_0}} B_k^{(1)} + D_{g} B_k^{(1)} ,\\
D_{g} B_k^{(1)} &=& \lambda_h \tau_k'' - \rho_k \lambda_h \left( \tau_k' \frac{\tau_{k_0}'}{\tau_{k_0}} + \tau_k \frac{\tau_{k_0}'' \tau_{k_0} - (\tau_{k_0}')^2}{\tau_{k_0}^2} \right),\qquad \tau_{k}''=D_{g}^2  \tau_{k}\ .
\ee
The quantities $\tau_k, \tau_k', \tau_k''$, $\rho_k$ and $\rho_k^2$ are all polynomially bounded, so $B^{(1)}_k$ is bounded by a polynomial in $k$, say $P_2(k)$. Thus, we bound the second derivative as:
\be
|D_{g}^2 \cx_k| \leq  P_2(k) e^{-\beta^*_h E_k} , 
\ee
By applying the same reasoning recursively, we see that the estimate holds for all derivatives of arbitrary order indexed by the multi-index $\a$:
\be
|D_g^\alpha \cx_k| \leq P_{|\alpha|}(k) e^{-\beta^*_h E_k} , \qquad \forall\, \alpha,
\ee
where $P_{|\alpha|}(k)$ denotes a polynomial whose degree may depend on $|\alpha|$. This ensures $\sum_k |D_g^\alpha \cx_k|$ converges uniformly for all $\a$ on any compact neighborhood $K\subset\mathscr{U}_h$.

The derivative $D_g Q_h^M$ gives
\be
D_g Q_h^M = Q_h^M(g_h) S_1^M(g_h) ,\qquad S_1^M(g_h) = \sum_{k < M, k \neq k_0} \frac{D_g\cx_k}{1 - \cx_k}
\ee
We have shown that $|\cx_k'| \leq  P_1(k) e^{-\beta^*_h E_k}$. For $k$ sufficiently large, $|1 - \cx_k| \geq 1/2$, then $\left| \frac{D_g \cx_k}{1 - \cx_k} \right| \leq 2P_1(k) e^{-\beta^*_h E_k}$. By the Weierstrass M-test, the series $S_1(g_h) = \lim_{M \to \infty} S_1^M(g_h)$ converges uniformly on $K$. Since $Q_h^M$ converges uniformly to $Q_h$, the product $\partial_{g_h} Q_h^M$ converges uniformly to $Q_h S_1$.

Consider the second derivative of $Q_h^M$:
\be
 D^2_{g_h} Q_h^M = Q_h^M (S_1^M)^2 + Q_h^M S_2^M,\qquad S_2^M = \partial_{g_h} S_1^M = \sum_{k < M, k \neq k_0} \lt[\frac{D_g^2\cx_k}{1 - \cx_k} + \left( \frac{D_g\cx_k}{1 - \cx_k} \right)^2\rt]
\ee
For suifficiently large $k$, the term in $[\cdots]$ is bounded by:
\be
\left| \frac{D_g^2 \cx_k}{1 - \cx_k} \right| + \left| \frac{D_g \cx_k}{1 - \cx_k} \right|^2 \leq 2 |D_g^2 \cx_k| + 4 |D_g\cx_k|^2 \leq 2 P_2(k) e^{-\beta^*_h E_k} + 4 P_1(k)^2 e^{-2\beta^*_h E_k}
\ee
Again by the Weierstrass M-test, $S_2^M$ converges uniformly to a function $S_2$, so $D^2_{g_h} Q_h^M $ converge uniformly to $Q_h S_1^2 + Q_h^M S_2$ on $K$. It is straight-forward to generalize the proof of uniform convergence to arbitrarily higher order derivatives $D^\a_g Q_h^M$.

\end{proof}

\section{On the sum over $m$}\label{sum over m}

We compute the following sum over $m\in\Z$ in \eqref{leadingpolecontri}:
\be
\Fs(A,g)&=&\frac{1}{E_{k_0}}\sum_{m\in\Z}\frac{e^{iA m \nu_0}}{r_0(g)+im\, \nu_0}  Q_0\lt(m,g\rt),\qquad
Q_0\lt(m,g\rt)=\prod_{k \neq k_0}\frac{1} {1 - \lambda \tau_k(g) e^{-s_0(m,g) E_k}},\\
\nu_0&=&\frac{2\pi}{E_{k_0}},\qquad r_0(g)=\frac{\ln(\l\t_{k_0})(g)}{E_{k_0}},\qquad s_0(m,g)=r_0(g)+im\, \nu_0.
\ee
where the dependence on $h$ and $\lambda$ in the arguments of various functions is omitted. Note that here $\sum_{m\in\Z}\cdots = \lim_{M\to\infty}\sum_{m=-M}^M\cdots$, since it comes from the principal value integral. For any $g\in\mathscr{U}_h$, $|\lambda \tau_k(g) e^{-s_0(m,g)E_k}|< 1$ for $k\neq k_0$, because $\{s_0(m,g)\}_m$ have real parts larger than all other poles with $k\neq k_0$. We expand $Q(m)$ into a Dirichlet series. Let $\mathbf{n} = (n_k)_{k \neq k_0}$ be a sequence of non-negative integers (occupation numbers) with finite support.
\be
Q_0(m,g) &=& \prod_{k \neq k_0} \sum_{n_k=0}^\infty \lt[\lambda \tau_k(g) e^{-s_0(m,g) E_k}\rt]^{n_k} = \sum_{\mathbf{n}} C(\mathbf{n},g) e^{-s_0(m,g) \ce(\mathbf{n})}, \\
C(\mathbf{n},g) &=& \prod_{k \neq k_0} \lt[\lambda \tau_k(g)\rt]^{n_k},\qquad \mathcal{E}(\mathbf{n}) = \sum_{k \neq k_0} n_k E_k.
\ee
For any compact neighborhood $K\subset \mathscr{U}_h$, the series converges uniformly in $m$ and $g\in K$, because (1) $|\lambda \tau_k(g) e^{-s_0(m,g) E_k}|=|\lambda \tau_k(g)| e^{-\re[r_0(g)] E_k}<1$ is independent of $m$, and (2) $|\lambda \tau_k(g) e^{-s_0(m,g) E_k}|<|\lambda| d_k^2  e^{-\b_* E_k}$, and thus $\sum_{k\neq k_0}|\lambda \tau_k(g) e^{-s_0(m,g) E_k}|$ converges uniformly (by Weierstrass M-test).

We interchange $\sum_{m\in\Z}$ and $\sum_{{\bf n}}$
\be
\Fs(A,g)=\frac{1}{E_{k_0}}\sum_{\bf n}C({\bf n},g) e^{-r_0(g) \ce(\mathbf{n})}\sum_{m\in\Z}\frac{e^{i\lt[A-\ce({\bf n})\rt] m \nu_0}}{r_0(g)+im\, \nu_0} .
\ee
The interchange of sum is justified by (1) the uniform bound for the partial sum over $m$ (see Lemma \ref{unifBoundsumm} below):
\be
\lt|\sum_{m=-M}^M\frac{e^{i\lt[A-\ce({\bf n})\rt] m \nu_0}}{r_0+im\, \nu_0} \rt|\leq \mathscr{K},
\ee
where $\mathscr{K}$ does not depend on $m$ and ${\bf n}$; (2) The dominated convergence:
\be
\sum_{\bf n}\lt|C({\bf n}) e^{-r_0 \ce(\mathbf{n})}\sum_{m=-M}^M\frac{e^{i\lt[A-\ce({\bf n})\rt] m \nu_0}}{r_0+im\, \nu_0}\rt|\leq \mathscr{K}\sum_{\bf n}\lt|C({\bf n})\rt| e^{-\re(r_0) \ce(\mathbf{n})}<\infty.
\ee

\begin{lemma}\label{unifBoundsumm}

(1) $T_M(\theta) = \sum_{m=-M}^{M} \frac{e^{i m \theta}}{\sigma + i m} $, $\re(\sig)>0$, $\theta\in\R$, is uniformly bounded $|T_M(\theta)|\leq \sk(\sig)$, where $\sk(\sig)>0$ is independent of $M$ and $\theta$ and continuous on the right-half $\sigma$-plane.

(2) The periodic function $T(\theta)=\lim_{M\to\infty}T_M(\theta)=\frac{2\pi e^{\sig (2\pi -\theta)}}{e^{2\pi \sig}-1}$ for $0<\theta<2\pi$ and has jump discontinuities at multiples of $2\pi$. 

\end{lemma}

\begin{proof}

(1) We separate the $m=0$ term and pair the terms for $\pm m$ ($m > 0$):
\be
T_M(\theta) = \frac{1}{\sig}+ \sum_{m=1}^{M} \left( \frac{e^{i m \theta}}{\sigma + i m} + \frac{e^{-i m \theta}}{\sigma - i m} \right)=\frac{1}{\sigma} + 2\sigma \sum_{m=1}^{M} \frac{\cos(m\theta)}{\sigma^2 + m^2} + 2 \sum_{m=1}^{M} \frac{m \sin(m\theta)}{\sigma^2 + m^2}. 
\ee

For the term involving Cosine:
\be 
\left| 2\sigma \sum_{m=1}^{M} \frac{\cos(m\theta)}{\sigma^2 + m^2} \right| \le 2|\sigma| \sum_{m=1}^{\infty} \frac{1}{|\sigma^2 + m^2|}\equiv C_1(\sigma). 
\ee
By $|\sigma^2 + m^2|\geq |m^2-|\sigma|^2| $, for any compact neighborhood $K$ in the right-half $\sigma$-plane such that $|\sig|\leq R$, we choose $N$ such that $N^2>2R^2$, then for all $m>N$, $m^2>2R^2\geq 2|\sigma|^2$, and thus $|\sigma^2 + m^2|\geq m^2-|\sigma|^2> m^2/2$. Therefore, when we split the sum $\sum_{m=1}^{\infty} \frac{1}{|\sigma^2 + m^2|}=\sum_{m=1}^{N} \frac{1}{|\sigma^2 + m^2|}+\sum_{m=N+1}^{\infty} \frac{1}{|\sigma^2 + m^2|}$, the infinite sum is uniformly convergent: $\sum_{m=N+1}^{\infty} \frac{1}{|\sigma^2 + m^2|}\leq \sum_{m=N+1}^\infty \frac{2}{m^2}$. Therefore $C_1(\sigma)$ is continuous on the compact neighborhood $K$. Since $K$ is arbitrary, $C_1(\sigma)$ is continuous on the entire right-half $\sig$-plane.

For the term involving Sine, we use $\frac{m}{\sigma^2 + m^2} = \frac{1}{m} - \frac{\sigma^2}{m(\sigma^2 + m^2)}$:
\be 
2 \sum_{m=1}^{M} \frac{m \sin(m\theta)}{\sigma^2 + m^2}= 2 \sum_{m=1}^{M} \frac{\sin(m\theta)}{m} - 2\sigma^2 \sum_{m=1}^{M} \frac{\sin(m\theta)}{m(\sigma^2 + m^2)}. 
\ee
The second term is bounded by $C_2$ which is independent of $M$ and $\theta$:
\be 
\left| 2\sigma^2 \sum_{m=1}^{M} \frac{\sin(m\theta)}{m(\sigma^2 + m^2)} \right| \le 2|\sigma|^2 \sum_{m=1}^{\infty} \frac{1}{m |m^2 + \sigma^2|}\equiv C_2(\sigma) .
\ee
where $C_2(\sigma)$ is continuous on the right-half $\sig$-plane by the same argument as the above.

We are left to bound the sum $\sum_{m=1}^{M} \frac{\sin(m\theta)}{m}$.
Since this function is odd and $2\pi$-periodic in $\theta$, it suffices to consider $\theta \in [0, \pi]$. For $\theta = 0$, The sum is 0. For $\theta \in (0, \pi]$, we use $A_k(\theta) = \sum_{j=1}^k \sin(j\theta)= \frac{\cos(\theta/2) - \cos((k+1/2)\theta)}{2\sin(\theta/2)}$ and the bound:
\be
|A_k(\theta)| \le \frac{2}{2\sin(\theta/2)} = \frac{1}{\sin(\theta/2)}\le \frac{\pi}{\theta}. 
\ee
by $\sin(\theta/2) \ge \frac{\theta}{\pi}$ for $\theta \in (0, \pi]$. Split the sum $\sum_{m=1}^{M} \frac{\sin(m\theta)}{m}$ at index $N = \lfloor \frac{1}{\theta} \rfloor$. For $m \le N$, using $|\sin(x)| \le |x|$:
\be 
\left| \sum_{m=1}^{\min(M, N)} \frac{\sin(m\theta)}{m} \right| \le \sum_{m=1}^{N} \frac{m\theta}{m} = N\theta \le \frac{1}{\theta} \cdot \theta = 1. 
\ee    
For $m > N$, 
\be 
\sum_{m=N+1}^{M} \frac{\sin(m\theta)}{m} = \sum_{m=N+1}^{M} \frac{A_m (\theta)- A_{m-1}(\theta)}{m} = \frac{A_M(\theta)}{M} - \frac{A_N(\theta)}{N+1} + \sum_{m=N+1}^{M-1} A_m(\theta) \left( \frac{1}{m} - \frac{1}{m+1} \right). 
\ee
Use the bound $|A_m| \le \frac{\pi}{\theta}$:
\be
&&\qquad\qquad \left| \frac{A_M}{M} \right| \le \frac{\pi}{M\theta} \le \frac{\pi}{(N+1)\theta} \le \pi, \qquad  \left| \frac{A_N}{N+1} \right| \le \frac{\pi}{(N+1)\theta} \le \pi,\\
&&\left| \sum_{m=N+1}^{M-1} A_m \left( \frac{1}{m} - \frac{1}{m+1} \right) \right| \le \frac{\pi}{\theta} \sum_{m=N+1}^{M-1} \left( \frac{1}{m} - \frac{1}{m+1} \right) = \frac{\pi}{\theta} \left( \frac{1}{N+1} - \frac{1}{M} \right) \le \frac{\pi}{\theta(N+1)} \le \pi. 
\ee
Thus, the sum for $m > N$ is bounded by $3\pi$. Combining both parts, $|\sum_{m=1}^{M} \frac{\sin(m\theta)}{m}| \le 1 + 3\pi\equiv C_3$. 

Combining all parts, we have:
\be 
|T_M(\theta)| \le  \frac{1}{|\sigma|} + C_1(\sig) + C_2(\sig) + 2C_3  \equiv \sk(\sig). \label{Ksig}
\ee
where $\sk(\sig)$ is independent of $\theta$ and $M$ and continuous in $\sig$ for $\re(\sig)>0$.

(2) Taking the limit $M\to\infty$, $T(\theta)=\lim_{M\to\infty}T_M(\theta)$ is still bounded uniformly by $\mathscr{K}$ and is a periodic function with period $2\pi$. We expand $T(\theta)$ in Fourier series $T(\theta)=\sum_{m\in\Z} c_m e^{im\theta}$ and check the coefficients 
\be
c_m=\frac{1}{2\pi}\int_0^{2\pi }T(\theta)e^{-im\theta}\rmd\theta=\frac{e^{2\pi \sig}}{e^{2\pi \sig}-1}\int_0^{2\pi }e^{-(\sigma +im)\theta}\rmd\theta=\frac{1}{\sigma +i m}.
\ee

\end{proof}

Apply the above result, we obtain
\be
\Fs(A,g)&=&\frac{ e^{2\pi{r_{0}(g)}/{\nu_{0}}}}{e^{2\pi{r_{0}(g)}/{\nu_{0}}}-1}\sum_{{\bf n}}C({\bf n},g)e^{-r_{0}(g){\cal E}(\mathbf{n})}\exp\lt[-\frac{r_{0}(g)}{\nu_{0}}\psi(A,{\bf n})\rt],\label{FsSeries}\\
&&\qquad \psi(A,{\bf n})=\left[A-\mathcal{E}({\bf n})\right]\nu_{0}+2\pi Z(A,{\bf n}).
\ee
Here $Z(A,{\bf n})\in\Z$ is defined such that $\psi(A,{\bf n})\in(0,2\pi)$, so $\exp\lt[-\frac{r_{0}(g)}{\nu_{0}}\psi(A,{\bf n})\rt]$ is periodic in $A$ with period $2\pi/\nu_0=E_{k_0}$ and has jump discontinuities at the boundaries of each period, in particular, it satisfies the uniform bound
\begin{eqnarray}
\lt|\exp\lt[-\frac{r_{0}(g)}{\nu_{0}}\psi(A,{\bf n})\rt]\rt|\leq C
\end{eqnarray}
for all $\mathbf{n}$ and for all $g$ in any compact neighborhood $K\subset\mathscr{U}_h$ (by the continuity of $r_0(g)$).

\begin{lemma}\label{lemmaSmoothFs} (1) For any $A>0$, the function $\Fs(A,g)$ is smooth on $\mathscr{U}_h$.

(2) $D^\a_g \Fs(A,g)$ is uniformly bounded on $\R_{>0}\times K$, for any compact neighborhood $K\subset\mathscr{U}_h$
\end{lemma}

\begin{proof} 

(1) To prove that $\Fs(A,g)$ is smooth with respect to $g$ on $\mathscr{U}_h$, we want to show that the series \eqref{FsSeries} can be differentiated term-by-term to any order. This requires proving that the series of derivatives converges uniformly on any compact subset $K \subset \mathscr{U}_h$.

The function $\Fs(A,g)$ is the product of a smooth prefactor $P(g)$ and a series $S(g)$: $ \Fs(A,g) = P(g)  S(g) $ where
\be 
P(g) = \frac{e^{2\pi{r_{0}(g)}/{\nu_{0}}}}{e^{2\pi{r_{0}(g)}/{\nu_{0}}}-1}, \qquad S(g)& =& \sum_{{\bf n}} T_{\mathbf{n}}(g),\qquad T_{\mathbf{n}}(g) = C({\bf n},g) e^{-r_0(g) \Lambda_{\mathbf{n}}(A) },\\
\Lambda_{\mathbf{n}}(A) &=& \mathcal{E}(\mathbf{n}) + \frac{1}{\nu_0}\psi(A,{\bf n}) 
\ee
Note that $\Lambda_{\mathbf{n}}(A)$ is independent of $g$. Since $\psi(A, \mathbf{n}) \in (0, 2\pi)$ and $\mathcal{E}(\mathbf{n}) \geq 0$, we have the bounds $\mathcal{E}(\mathbf{n}) < \Lambda_{\mathbf{n}}(A) < \mathcal{E}(\mathbf{n}) + E_{k_0}$. 

We need to prove $S(g)$ to be smooth. Consider an arbitrary differential operator $D^\alpha_g$ of order $|\alpha|$ with respect to $g$ (see Section \ref{Bound the derivatives of tau} for the multi-index notation $D^\a_g$). We must bound $|D^\alpha_g T_{\mathbf{n}}(g)|$ on a compact set $K \subset \mathscr{U}_h$. The derivative is:
\begin{eqnarray}
D^\alpha_g T_{\mathbf{n}}(g) = \sum_{|\mu| \le |\alpha|} \binom{|\alpha|}{|\mu|} \left[ D^\mu_g C(\mathbf{n},g) \right] \left[ D^{\alpha-\mu}_g e^{-r_0(g) \Lambda_{\mathbf{n}}(A)} \right] .
\end{eqnarray}

Repeated differentiation of $e^{-r_0(g) \Lambda_{\mathbf{n}}}$ brings down factors of $\Lambda_{\mathbf{n}}$ and derivatives of $r_0(g)$. Since $r_0(g)$ is smooth, its derivatives are bounded on the compact set $K$. Thus, there exists a constant $C_1$ (depending on $\alpha, K$) such that
\be  
\left| D^{\rho}_g e^{-r_0(g) \Lambda_{\mathbf{n}}} \right| \le C_1 \left(1 + \Lambda_{\mathbf{n}}\right)^{|\rho|} e^{-\text{Re}(r_0(g)) \Lambda_{\mathbf{n}}}\leq C_2 \left(1 + \mathcal{E}(\mathbf{n})\right)^{|\rho|} e^{-(\beta^*+\eps) \mathcal{E}(\mathbf{n})} \qquad \forall g\in K,
\ee
due to $\mathcal{E}(\mathbf{n}) < \Lambda_{\mathbf{n}} < \mathcal{E}(\mathbf{n}) + E_{k_0}$, and since $\text{Re}(r_0(g)) > \beta^*$ everywhere on $\mathscr{U}_h$, it follows that on any compact $K \subset \mathscr{U}_h$ we have $\text{Re}(r_0(g)) \geq \beta^* + \eps$ for some $\eps > 0$.

Consider $ D^\mu_g C(\mathbf{n},g)$ where $C(\mathbf{n},g) = \prod_{k \neq k_0} u_k(g)^{n_k}$ and $u_k(g) = \lambda \tau_k(g)$. The occupation numbers $\mathbf{n}$ is of finite support, so the $n_k\neq 0$ only for finitely many $k$. When differentiating the product $\prod_{k \neq k_0} u_k(g)^{n_k}$ a finite number of times ($|\mu|$ times), the result is a finite sum of terms, each involving at most $|\mu|$ factors of the form $D_g u_k $ and a polynomial in $n_k$.
\begin{itemize}
\item  By Lemma \ref{lemmaBoundDtau}, derivatives of $\tau_k$ are polynomially bounded in $k$. Since $E_k$ grows polynomially in $k$, derivatives are also polynomially bounded in $E_k$.
\item  Differentiating $u_k^{n_k}$ introduces a polynomial of $n_k$. Since $E_k \ge 1$, we have $n_k \le n_k E_k \le \mathcal{E}(\mathbf{n})$. Thus, any polynomial in $n_k$ is bounded by a polynomial in the total energy $\mathcal{E}(\mathbf{n})$.
\end{itemize}
Combining these, we obtain
$$ |D^\mu_g C(\mathbf{n},g)| \le \sum _i P_\mu^{(i)}(\mathcal{E}(\mathbf{n})) \sup_{g \in K} |C_i(\mathbf{n},g)| \leq P_\mu(\mathcal{E}(\mathbf{n})) \prod_{k\neq k_0}\lt|\l d_k^2\rt|^{n_k}$$
where $C_i(\mathbf{n},g)$ removes a number of $u_k$'s from $C(\mathbf{n},g)$. We use the estimate $\sup_{g \in K} |\lambda \tau_k(g)| \leq |\lambda | d_k^2$ and $1\leq C|\lambda| d_k^2$ for some $C>0$. $P_\mu$ is a polynomial with positive coefficients.

We obtain the following bound of $D^\alpha_g T_{\mathbf{n}}(g)$:
\be
|D^\alpha_g T_{\mathbf{n}}(g)| \le P'_\alpha(\mathcal{E}(\mathbf{n})) e^{-(\beta^*+\eps)\mathcal{E}(\mathbf{n})} \prod_{k\neq k_0}\lt|\l d_k^2\rt|^{n_k}
\ee
where $P'_\alpha$ is a polynomial with positive coefficients. Given any $\delta > 0$, we can find a constant $C_\delta > 0$ such that $P'_\alpha(x) \leq C_\delta e^{\delta x}$ for all $x > 0$. We choose $\delta < \eps$.
\be
|D^\alpha_g T_{\mathbf{n}}(g)| \le C_\delta \prod_{k \neq k_0} \left[ |\l| d^2_k e^{-(\beta^*+\eps') E_k} \right]^{n_k},\qquad \eps' = \eps-\delta>0.
\ee
Summing $\prod_{k \neq k_0} \left( |\l |d^2_k e^{-(\beta^*+\eps') E_k} \right)^{n_k}$ over $\mathbf{n}$ converges by Lemma \ref{lemmaA1}, because $|\l| d^2_k e^{-(\beta^*+\eps') E_k}< 1$ (recall that $\b_*=\sup_{k\neq k_0}\b_k(\l)$) and $|\l| \sum_{k\neq k_0}d^2_k e^{-(\beta^*+\eps') E_k}< \infty$. Therefore, the series $\sum_{\mathbf{n}} D^\alpha_g T_{\mathbf{n}}(g)$ converges uniformly on $K$ by the Weierstrass M-test. 

(2) By $|D^\alpha_gS(g)|\leq \sum_{\mathbf{n}} |D^\alpha_g T_{\mathbf{n}}(g)|\leq C_\delta \sum_{\mathbf{n}}\prod_{k \neq k_0} [ \l d^2_k e^{-(\beta^*+\eps') E_k} ]^{n_k}$, the bound is valid on any compact neighborhood $K\subset\mathscr{U}_h$ and is independent of $A$ and $g$.

\end{proof}

\section{Interchange sum over $m$ and integral over $g$}\label{Interchange sum over m and integral over g}

We would like to justify the interchange of the integral over the compact neighborhood $\overline{\mathscr{U}}_{\rm int} \subset \text{SL}(2, \mathbb{C})^N$ and the summation over $m \in \mathbb{Z}$ in the expression for the stack amplitude $\mathscr{A}_{\mathcal{K}}$. The stack amplitude is given by:
\be 
\mathscr{A}_{\mathcal{K}} = \int_{\overline{\mathscr{U}}_{\rm int}} \mathrm{d}\Omega(g) \prod_h \omega_h^{\mathrm{bos}}(g) \prod_b \omega_b^{\mathrm{bos}}(g) ,
\ee
where the dominant term in each $\omega_h^{\mathrm{bos}}$ takes the following form (certain entries such as $k_0$ and $\lambda_h$ are suppressed for brevity):
\be
\omega_f^{\mathrm{bos}}(A_h; g_h) &=& \frac{e^{A_h r_h(g_h)}}{E_{k_0}} \left[ \sum_{m \in \mathbb{Z}} \frac{e^{i A_h m \nu_0}}{r_h(g_h) + i m \nu_0} Q_h(m, g_h) \right] ,\\
 Q_h\lt(m,g_h\rt)&=&\prod_{k \neq k_0}\frac{1} {1 - \lambda_h \tau_k^{(h)}(g_h) e^{-s_0(m,g_h) E_k}}.
\ee
with $r_h(g_h) = \frac{\ln[\lambda_h \tau_{k_0}^{(h)}(g_h)]}{E_{k_0}}$, $\nu_0 = \frac{2\pi}{E_{k_0}}$, and $s_0(m,g_h)=r_h(g_h) + i m \nu_0$. We define the partial sum for each internal face $h$, assuming the cutoff $M$ to be independent of $h$:
\be
S_{M,h}(g_h) = \sum_{m=-M}^M \frac{e^{i A_h m \nu_0}}{r_h(g_h) + i m \nu_0} Q_h( m, g_h)
\ee 
To prove the interchange, we show that $S_{M,f}(g)$ is uniformly bounded for all $M$ and $g \in \mathscr{U}_{\rm int}$ (here $g$ relates to $g_h$ by the projection $p_{h}(g)=g_h$). 

We have $|\lambda_h \tau_k^{(h)}(g_h) e^{-s_0(m,g_h) E_k}|\leq 1$ for $k\neq k_0$, because $\{s_0(m,g_h)\}_m$ have real parts larger than all other poles with $k\neq k_0$ on $\overline{\mathscr{U}}_{\rm int}$. We expand $Q_h$ into a Dirichlet series. Let $\mathbf{n} = (n_k)_{k \neq k_0}$ be a sequence of non-negative integers (occupation numbers) with finite support.
\be
Q_f(m, g_h)  &=& \prod_{k \neq k_0} \sum_{n_k=0}^\infty \lt[\lambda_h \tau_k^{(h)}(g_h) e^{-s_0(m,g_h)E_k}\rt]^{n_k} = \sum_{\mathbf{n}} C(\mathbf{n},g_h) e^{-s_0(m,g_h) \ce(\mathbf{n})}, \\
C(\mathbf{n},g_h) &=& \prod_{k \neq k_0} \lt[\lambda_h \tau^{(h)}_k(g_h)\rt]^{n_k},\qquad \mathcal{E}(\mathbf{n}) = \sum_{k \neq k_0} n_k E_k.
\ee
The series converges point-wisely in $g_h$ and uniformly in $m$, because $|\lambda_h \tau_k^{(h)}(g_h) e^{-s_0(m,g_h) E_k}|=|\lambda_h \tau_k^{(h)}(g_h)| e^{-\re(r_h(g_h)) E_k}\leq 1$ is independent of $m$.

Substituting the expansion into the partial sum:
\be 
S_{M,h}(g_h) = \frac{1}{\nu_0} \sum_{\mathbf{n}} C(\mathbf{n}, g_h) e^{-r_h(g_h) \mathcal{E}(\mathbf{n})} \left( \sum_{m=-M}^M \frac{e^{i [A_h - \mathcal{E}(\mathbf{n})] m \nu_0}}{\sigma(g_h) + i m} \right) ,
\ee
where $\sigma(g_h) = r_h(g_h) / \nu_0=\frac{1}{2\pi}\ln[\l_h\t_{k_0}^{(h)}(g_h)]$. 

From Lemma \ref{unifBoundsumm}, the inner sum $T_M(\theta, \sigma) = \sum_{m=-M}^M \frac{e^{im\theta}}{\sigma + im}$ is bounded by a constant $\sk(\sigma)$ independent of $M$ and $\theta$ (see \eqref{Ksig}). The bound $\sk(\sigma)$ is a continuous function of $\sigma$ for $\text{Re}(\sigma) > 0$. 
On the compact neighborhood $\overline{\mathscr{U}}_{\rm int}$, the function $\sigma(g_h)$ is continuous in $g_h$ and $\re(\sigma(g_h)) \geq \frac{1}{2\pi}E_{k_0}\b^*> 0$ for all $g \in \overline{\mathscr{U}}_{\rm int}$. Thus, $\sk(\sigma(g_h))$ is continuous in $g_h$ on the compact neighborhood $ \overline{\mathscr{U}}_{\rm int}$. We have a uniform bound:
\be 
\sup_{g \in \overline{\mathscr{U}}_{\rm int}, M, \mathbf{n}} \left| \sum_{m=-M}^M \frac{e^{i [A_h - \mathcal{E}(\mathbf{n})] m \nu_0}}{\sigma(g_h) + i m} \right| \leq \sk_{\max} < \infty 
\ee
where $\sk_{\max}$ is the maximum of $sk(\sigma(g_h))$ on $ \overline{\mathscr{U}}_{\rm int}$. Applying this bound to $S_{M,h}(g_h)$, we obtain
\be
|S_{M,h}(g_h)| \leq \frac{\sk_{\max}}{\nu_0} \sum_{\mathbf{n}} |C(\mathbf{n}, g_h)| e^{-\text{Re}(r_h(g_h)) \mathcal{E}(\mathbf{n})} 
\ee
The sum over $\mathbf{n}$ is the expansion of the absolute product:
\be
&&\sum_{\mathbf{n}} |C(\mathbf{n}, g_h)| e^{-\text{Re}(r_h(g)) \mathcal{E}(\mathbf{n})} = \prod_{k \neq k_0} \frac{1}{1 - |\lambda_h \tau_k^{(h)}(g_h)| e^{-\text{Re}(r_h(g_h)) E_k}}\\
&=& \prod_{k \neq k_0, k\leq k_s} \frac{1}{1 - |\lambda_h \tau_k^{(h)}(g_h)| e^{-\text{Re}(r_h(g_h)) E_k}}\prod_{k>k_s} \frac{1}{1 - |\lambda_h \tau_k^{(h)}(g_h)| e^{-\text{Re}(r_h(g_h)) E_k}}.\label{infiniteprodC}
\ee
Here we set $k_s$ such that for $k>k_s$, $|\lambda_h \tau_k^{(h)}(g_h)| e^{-\text{Re}(r_h(g_h)) E_k}\leq |\lambda_h |d_k^2 e^{-\b^* E_k}<1$, $g\in  \overline{\mathscr{U}}_{\rm int}$, so the infinite product of $k>k_s$ converges uniformly on $\overline{\mathscr{U}}_{\rm int}$, then the right-hand side is a bounded continuous function of $g$ on $\overline{\mathscr{U}}_{\rm int}$. Thus, there exists constant $B_h$ such that 
\be
|S_{M,h}(g_h)| \leq B_h
\ee 
for all $M$ and $g \in \overline{\mathscr{U}}_{\rm int}$.

The leading integrand for the stack amplitude is $I_M(g) = \cf(g) \prod_h S_{M,h}(g)$, where $\cf(g)=\prod_h \frac{e^{A_h r_h(g_h)}}{E_{k_0}}\prod_b\o_b(g_b)$ is a continuous function on $\overline{\mathscr{U}}_{\rm int}$. Since each $S_{M,h}$ is uniformly bounded by $B_h$ on on $\overline{\mathscr{U}}_{\rm int}$, the product is uniformly bounded:
\be 
|I_M(g)| \leq |\cf(g)|\prod_h B_h \leq \mathcal{G} 
\ee
where $\mathcal{G}$ is a constant since $\overline{\mathscr{U}}_{\rm int}$ is compact and the function is continuous. On a compact domain with a smooth measure $\mathrm{d}\Omega(g)$, the constant function $\mathcal{G}$ is integrable. By the dominated convergence theorem, we have:
\be 
\int_{\overline{\mathscr{U}}_{\rm int}} \mathrm{d}\Omega(g) \lim_{M \to \infty} I_M(g) = \lim_{M \to \infty} \int_{\overline{\mathscr{U}}_{\rm int}} \mathrm{d}\Omega(g) I_M(g) 
\ee
This justifies interchanging the integral over $\mathscr{U}$ and the summation over $m$ for each face.




\section{Parametrization of $\cc_{\rm int}^\cs/\cg_{\rm int}$}\label{AppendixParametrizationofcccscg}

Let $K$ be a connected 2-complex. Let $K^{(0)} = V$ be the set of vertices and $K^{(1)}$ be the 1-skeleton. Let $E$ be the set of oriented edges in $K^{(1)}$. We denote the inverse of an edge $e$ by $e^{-1}$.

Let $v_* \in V$ be a base point. Let $T$ be a maximal spanning tree in the 1-skeleton $K^{(1)}$. Since $T$ is a spanning tree, it contains all vertices $V$ and is contractible.

\begin{definition}
For any pair of vertices $u, v \in V$, let $p_T(u, v)$ denote the unique path in $T$ from $u$ to $v$. For any edge $e = (u, v) \in E \setminus T$ (edges not in the tree), we define the fundamental loop based at $v_*$ associated with $e$ as:
$$ \ell_e = p_T(v_*, u) \cdot e \cdot p_T(v, v_*) $$
where $\cdot$ denotes path concatenation.
\end{definition}

\begin{lemma}
The set of homotopy classes $\{ [\ell_e] \mid e \in E \setminus T \}$ generates the fundamental group $\pi_1(K, v_*)$.
\end{lemma}

\begin{proof}
Since $T$ is a contractible subcomplex of $K$, the quotient map $q: K \to K/T$ is a homotopy equivalence. In the quotient space $K/T$, (1)  All vertices $v \in V$ are identified to a single point, which we identify with $v_*$; (2)  All edges in $T$ are mapped to the constant path at $v_*$; (3) Each edge $e \in E \setminus T$ becomes a loop at $v_*$. Let's call this loop $\hat{e}$.

The 1-skeleton of $K/T$, denoted $(K/T)^{(1)}$, is a wedge of circles (one circle for each edge in $E \setminus T$). By the Van Kampen theorem (or standard properties of wedge sums), $\pi_1((K/T)^{(1)}, v_*)$ is the free group generated by the classes of these loops $\{[\hat{e}] \mid e \in E \setminus T\}$.

Consider the image of a fundamental loop $\ell_e$ under the quotient map $q$. Since $p_T(v_*, u)$ and $p_T(v, v_*)$ lie entirely in $T$, they map to the constant path at $v_*$. Therefore:
$$ q_{\#}([\ell_e]) = [\text{const} \cdot \hat{e} \cdot \text{const}] = [\hat{e}] $$
Since $q$ induces an isomorphism on fundamental groups, and the classes $\{[\hat{e}]\}$ generate $\pi_1((K/T)^{(1)})$, the classes $\{[\ell_e]\}$ generate $\pi_1(K^{(1)}, v_*)$.

Finally, $\pi_1(K, v_*)$ is obtained from $\pi_1(K^{(1)}, v_*)$ by adding relations corresponding to the boundaries of the 2-cells (faces) of $K$. Since $\{[\ell_e]\}$ generate the group $\pi_1(K^{(1)}, v_*)$, their images generate the quotient group $\pi_1(K, v_*)$.

\end{proof}

\begin{lemma}\label{HeenotinFlenotinct}
Pick up a minimal set $L$ of independent generators in $\pi_1(\ck_-,v_*)$ that are fundamental loops and denote the set of edges along these generators by $\Fl$. 
On $\cc_{\rm int}^\cs$ and under the gauge fixing $ H_e = 1, \forall e \subset \ct $, the holonomies $\{H_e\}_{e\not\in\Fl,e\not\in\ct}$ are uniquely determined by $\{H_l\}_{l\in\Fl,l\not\in\ct}$.
\end{lemma}

\begin{proof}
Let $\ct$ be the fixed maximal spanning tree in $\ck_-$. Let $E_{\rm int}$ be the set of edges in $\ck_-$. The set of edges not in the tree is denoted by $S = E_{\rm int} \setminus \ct$. The set of all fundamental loops is denoted by $\Gamma_{all} = \{ \ell_e \mid e \in S \}$ where $\ell_e = p_T(v_*, u) \cdot e \cdot p_T(v, v_*)$.

We fix the gauge partially by requiring $ H_e = 1,\ \forall e \subset \ct $. This condition determines the gauge $g_v$ uniquely for all $v \neq v_*$ relative to $g_{v_*}$.
In this gauge, for any $e \in S$, the holonomy of the fundamental loop is:
$$ H(\ell_e) = H(p_T) H_e H(p_T)^{-1} =  H_e $$
Thus, the variables describing the connection are exactly the matrices $\{H_e \mid e \in S\}$. The residual gauge freedom is the global conjugation by $g_{v_*} \in SU(2)$.

Let $L \subset \Gamma_{all}$ be a minimal set of independent generators for $\pi_1(\ck_-, v_*)$. Let $\Fl \subset S$ be the subset of edges corresponding to the loops in $L$. Let $e \in S \setminus \Fl$ be an edge whose fundamental loop $\ell_e$ is not in the minimal set $L$. Since $L$ generates $\pi_1(K)$, the homotopy class $[\ell_e]$ can be expressed as a product of classes in $L$ and their inverses. Let this word be $W_e$.
$$ [\ell_e] = W_e( [\ell_{l_1}], \dots, [\ell_{l_k}] ) \quad \text{in } \pi_1(\ck_-, v_*) $$
where $l_i\in\Fl$ and $l_i\not\in\ct$.

For the fundamental group $\pi_1(\ck_-,v_*)$, relations are generated by the boundaries of the 2-cells (faces), and they implies that the loop $\ell_e \cdot (W_e)^{-1}$ is contractible in $\ck_-$. Specifically, this loop is homotopic to a product of conjugates of face boundaries.
There exist paths $\alpha_j$ (along 1-skeleton and starting at $v_*$), faces $h_j$ with boundaries $\partial h_j$, and signs $\epsilon_j \in \{\pm 1\}$ such that:
$$ \ell_e \simeq \left( \prod_{j=1}^N \alpha_j \cdot (\partial h_j)^{\epsilon_j} \cdot \alpha_j^{-1} \right) \cdot \tilde{W}_e(\ell_{l_1}, \dots, \ell_{l_k}) $$
where $\tilde{W}_e$ is the loop representative of the word $W_e$. The product term represents the "trivial" part of the loop in the quotient group, formed by traversing to a face, going around it, and returning.

We apply the holonomy map $H$ to the homotopy relation derived above. For every face $h$, the loop holonomy is $H(\partial h) = s_h$, where $s_h =\pm 1$ generates the center $Z(\Su)$. Therefore, for any path holonomy $M = H(\alpha_j)$, we have:
$$ H(\alpha_j \cdot (\partial h_j)^{\epsilon_j} \cdot \alpha_j^{-1}) = M s_h^{\epsilon_j} M^{-1} = s_h^{\epsilon_j} $$

Now, we evaluate the holonomy of $\ell_e$:
$$ H(\ell_e) = \left( \prod_{j=1}^N s_{h_j}^{\epsilon_j} \right)  H( \tilde{W}_e ) $$
Let $s_e = \prod_{j=1}^N s_{h_j}^{\epsilon_j}=\pm1$. Since $H(\cdot)$ is multiplicative on path concatenation $ H(\tilde{W}_e) = W_e( H(\ell_{l_1}), \dots, H(\ell_{l_k}) ) $. In the fixed gauge, $H(\ell_e) = H_e$. Thus:
$$ H_e = s_e  W_e( H_{l_1}, \dots, H_{l_k} ) $$
The set of holonomies $\{H_e\}_{e\not\in\Fl,e\not\in\ct}$ are uniquely determined by the matrices $\{H_l\}_{l \in \Fl,l\not\in\ct}$.


\end{proof}

\bibliography{muxin.bib}

\end{document}